\documentclass[aps,pra,groupedaddress,nofootinbib,notitlepage,showpacs,floatfix]{revtex4-1}
\usepackage{}

\usepackage{graphicx,graphics,epsfig,subfigure,times,bm,bbm,amssymb,amsmath,amsfonts,amsthm,mathrsfs,MnSymbol}
\usepackage[matrix,frame,arrow]{xypic}
\usepackage[pdfstartview=FitH]{hyperref}
\usepackage[pdftex]{color}
\usepackage{hypernat}
\usepackage{braket}%Dirac Notation in QM
\usepackage{enumerate}
\input{Qcircuit.tex} %Quantum Circuits

\newtheorem{thm}{Theorem}
\newtheorem{mydef}{Definition}

\newtheorem{mytheorem}{Theorem}
\newtheorem{mylemma}{Lemma}

\newcommand{\bes} {\begin{subequations}}
\newcommand{\ees} {\end{subequations}}
\newcommand{\bea} {\begin{eqnarray}}
\newcommand{\eea} {\end{eqnarray}}

\newcommand{\C}{\ensuremath{\mathbb{C}}} %complex numebrs
\newcommand{\ii}{\ensuremath{{i}}}%i:=\sqrt{-1}
\newcommand{\abs}[1]{\ensuremath{\left|#1\right|}} %absolute value
\newcommand{\norm}[1]{\ensuremath{\left\|#1\right\|}} %norm
\newcommand{\opU}{\ensuremath{{{U}}}}
\newcommand{\opH}{\ensuremath{{{H}}}}
\newcommand{\beq}{\begin{equation}}
\newcommand{\dfsh}{\tilde{\mathcal{H}}}

\newcommand{\eeq}{\end{equation}}

\newcommand{\ignore}[1]{}
\newcommand{\mc}[1]{\mathcal{#1}}

\def\a{\alpha}
\def\b{\beta}

\def\d{\delta}

\def\r{\rho}
\def\s{\sigma}

\def\ps{\psi}
\def\o{\omega}

\def\ox{\otimes}
\def\>{\rangle}
\def\<{\langle}
\def\Tr{\mathrm{Tr}}

\newcommand{\ketb}[2]{|{#1}\>\<#2|}

\newcommand{\opI}{\ensuremath{{{I}}}}
\newcommand{\opA}{\ensuremath{{{A}}}}
\newcommand{\opB}{\ensuremath{{{B}}}}
\newcommand{\opX}{\ensuremath{{{X}}}}
\newcommand{\opY}{\ensuremath{{{Y}}}}
\newcommand{\opZ}{\ensuremath{{{Z}}}}

\newcommand{\ee}{\ensuremath{{e}}} %e:=lim_{n\to\infty}(1+1/n)^n

\def\lp{\left(}
\def\rp{\right)}
\def\ls{\left[}
\def\rs{\right]}
\def\lb{\left\{}
\def\rb{\right\}}

\def\dgr{\dagger}

\begin{document}

\title{Review of Decoherence Free Subspaces, Noiseless Subsystems, and Dynamical Decoupling}
\author{Daniel A. Lidar}
\affiliation{Departments of Electrical Engineering, Chemistry, and Physics, and Center
for Quantum Information Science \& Technology, University of Southern
California, Los Angeles, California 90089, USA}
\begin{abstract}
Quantum information requires protection from the adverse affects of decoherence and noise. This review provides an introduction to the theory of decoherence-free subspaces, noiseless subsystems, and dynamical decoupling. It addresses quantum information preservation as well 
as 
protected computation.
\end{abstract}

\maketitle

\section{Introduction}
The protection of quantum information is a central task in quantum information processing \cite{nielsen2000quantum}. Decoherence and noise are obstacles which must be overcome and managed before large scale quantum computers can be built. This review provides an introduction to the theory of decoherence-free subspaces, noiseless subsystems, and dynamical decoupling, among the key tools in the arsenal of decoherence mitigation strategies. It is based on lectures given by the author at the University of Southern California as part of a graduate course on quantum error correction, and as such is not meant to be a comprehensive review, nor to supply an exhaustive list of references. Rather, the goal is to get the reader quickly up to speed on a subset of key topics in the field of quantum noise avoidance and suppression. For previous reviews overlapping with some of the theoretical topics covered here see, e.g., Refs.~\cite{Lidar:2003fk,Yang:2011:2,Lidar-Brun:book}.

The review is structured as follows. Section~\ref{sec:DFS} introduces decoherence-free subspaces (DFSs). Section~\ref{sec:CD} defines and analyzes the collective dephasing model, and explains how to combine the corresponding DFS encoding with universal quantum computation. Section~\ref{sec:CDec} considers the same problem in the context of the more general collective decoherence model, where general noise afflicts all qubits simultaneously. Section~\ref{noiseless} introduces and analyzes noiseless subsystems (NSs), a key generalization of DFSs which underlies all known methods of quantum information protection. The NS structure is illustrated with the three-qubit code against collective decoherence, including computation over this code. We then proceed to dynamical decoupling (DD). Section~\ref{sec:DD} introduces the topic by analyzing the protection of a single qubit against pure dephasing and against general decoherence, using both ideal (zero-width) and real (finite-width) pulses.  Section~\ref{sec:DD-symm} briefly discusses DD as a symmetrization procedure. Section \ref{sec:DD-DFS} discusses combining DD with DFS in the case of two qubits. Section~\ref{CDD} addresses concatenated dynamical decoupling (CDD), a method to achieve high-order decoupling. In the final technical Section~\ref{sec:DD-rep}, we come full circle and connect dynamical decoupling to the representation theory ideas underlying noiseless subsystems theory, thus presenting a unified view of the the approaches. Concluding remarks and additional literature entries are presented in Section~\ref{sec:conc}.

\section{Decoherence-Free Subspaces}
\label{sec:DFS}

Let us begin by assuming that we have two systems: $S$ and $B$, defined by the Hilbert spaces $\mathcal{H}_S$ and $\mathcal{H}_B$, respectively. In general, the dynamics of these two systems are generated by
\beq
H=H_S+H_B+H_{SB},
\eeq
where $H_S$ and $H_B$ are the Hamiltonians corresponding to the pure dynamics of systems $S$ and $B$, respectively, and $H_{SB}$ is the interaction between the two systems. Using the Kraus operator sum representation (OSR) \cite{nielsen2000quantum}, we can effectively study the reduced dynamics of $S$ for an initial state $\rho_S(0)$, where
\beq
\rho_S(0)\mapsto \rho_S(t)=\sum_\alpha K_\alpha(t) \rho_S(0) K^{\dagger}_\alpha(t)
\label{eq:rhomap}
\eeq
after the partial trace over system $B$ is completed. The Kraus operators $K_\alpha(t)$ satisfy the relation $\sum_\alpha K^{\dagger}_\alpha(t)K_\alpha(t)=I_S$ $\forall t$, where $I_S$ is the identity operator on the system, $S$. The OSR generally results in non-unitary evolution in the system Hilbert space. Therefore, let us define decoherence as follows:

\begin{mydef}
An open system undergoes decoherence if its evolution is not unitary. Conversely, an open system which undergoes purely unitary evolution (possibly only in a subspace of its Hilbert space) is said to be decoherence-free.
\end{mydef}

We would now like to study how we can avoid decoherence.

\subsection{A Classical Example}
Let us begin with a simple classical example.  Assume we have three parties: Alice, Bob, and Eve.  Alice wants to send a message to Bob, and Evil Eve (the Environment) wants to mess that message up.  Let's also assume that the only way in which Eve can act to mess up the message is by, with some probability, flipping all of the bits of the message.  If Alice were to send only one bit to Bob there would be no way of knowing if that bit had been flipped.  But let's say Alice is smarter than Eve and decides to send two bits.  She also communicates with Bob beforehand and tells him that if he receives a 00 or 11 he should treat it as a ``logical" 0 ($\bar{0}$), and if he receives a 01 or 10 to treat it as a ``logical" 1 ($\bar{1}$).  If this scheme is used, Eve's ability to flip both bits has no effect on their ability to communicate:
\bes
\bea
\bar{0} = \begin{cases}
00\\11
\end{cases} \mapsto\ \begin{cases}
11\\00
\end{cases} = \bar{0} \\
\bar{1} = \begin{cases}
01\\10
\end{cases} \mapsto\ \begin{cases}
10\\01
\end{cases} = \bar{1}
\eea
\ees
In this example we use parity conservation to protect information.  The logical 0 is even parity and the logical 1 is odd parity.  Encoding logical bits in parity in this way effectively hides the information from Eve's bit flip error.

It is easy to see that the same strategy works for $N$ bits when all Eve can do is to flip all bits simultaneously. Namely, Alice and Bob agree to encode their logical bits into the bit-string pairs $x_1x_2\dots x_N$ and $y_1y_2\dots y_N$, where $y_i = x_i\oplus 1$ (addition modulo $2$), i.e., $y_i=0$ if $x_i=1$ and $y_i=1$ if $x_i=0$. This encoding strategy yields $N-1$ logical bits given $N$ physical bits, i.e., the code rate (defined as the number of logical bits to the number of physical bits) is $1-1/N$, which is asymptotically close to $1$ in the large $N$ limit.

\subsection{Collective Dephasing DFS}
\label{CD-DFS}

Let us now move to a genuine quantum example by analyzing in detail the operation of the simplest decoherence-free subspace (DFS).
Suppose that a system of $N$ qubits is coupled to a bath in a symmetric way,
and undergoes a dephasing process. Namely, qubit $j$ undergoes the
transformation 
\beq
|0\rangle _{j}\rightarrow |0\rangle _{j}\qquad |1\rangle _{j}\rightarrow
e^{i\phi }|1\rangle _{j},
\eeq%
which puts a random phase $\phi $ between the basis states $|0\rangle $ and $%
|1\rangle $ (eigenstates of $\sigma _{z}$ with respective eigenvalues $+1$
and $-1$). Notice how the phase $\phi $---by assumption---has no space ($j$)
dependence, i.e., the dephasing process is invariant under qubit
permutations. Suppose the phases have a distribution $p_\phi$ and define the matrix $R_{z}(\phi )=\mathrm{%
diag}\left( 1,e^{i\phi }\right) $ acting on the $\{|0\rangle ,|1\rangle \}$
basis. If each qubit is initially in an arbitrary pure state 
$|\psi \rangle _{j}=a|0\rangle _{j}+b|1\rangle _{j}$ 
then the random process outputs a state $R_{z}(\phi )|\psi \rangle _{j}$ with probability $p_\phi$, i.e., it yields a pure state ensemble $\{R_{z}(\phi )|\psi \rangle _{j},p_\phi\}$, which is equivalent to the density matrix 
$\rho _{j}= \sum_\phi p_\phi R_{z}(\phi )|\psi \rangle _{j}\bra{\psi}R_{z}^\dgr(\phi )$. Clearly, this is in the form of a Kraus OSR, with Kraus operators $K_\phi = \sqrt{p_\phi}R_{z}(\phi )$, i.e., we can also write $\rho _{j}= \sum_\phi K_\phi |\psi \rangle _{j}\bra{\psi}K_\phi^\dgr$. Thus, each of the qubits will decohere. To see this explicitly, let us assume that $\phi$ is continuously distributed, so that 
\beq
\rho _{j}=\int_{-\infty }^{\infty }p(\phi )R_{z}(\phi )|\psi \rangle _{j}\langle
\psi |R_{z}^{\dagger }(\phi )\,d\phi ,
\eeq%
where $p(\phi )$ is a probability density, and we assume the initial state
of all qubits to be a product state. For a Gaussian distribution, $p(\phi
)=\left( 4\pi \alpha \right) ^{-1/2}\exp (-\phi ^{2}/4\alpha )$, it is
simple to check that 
\beq
\rho _{j}=\left( 
\begin{array}{cc}
|a|^{2} & ab^{\ast }e^{-\alpha } \\ 
a^{\ast }be^{-\alpha } & |b|^{2}%
\end{array}%
\right) .
\eeq%
The decay of the off-diagonal elements in the computational basis is a
signature of decoherence.

Let us now consider what happens in the two-qubit Hilbert space. The four
basis states undergo the transformation 
\begin{eqnarray}
|0\rangle _{1}\otimes |0\rangle _{2} &\rightarrow &|0\rangle _{1}\otimes
|0\rangle _{2}  \notag \\
|0\rangle _{1}\otimes |1\rangle _{2} &\rightarrow &e^{i\phi }|0\rangle
_{1}\otimes |1\rangle _{2}  \notag \\
|1\rangle _{1}\otimes |0\rangle _{2} &\rightarrow &e^{i\phi }|1\rangle
_{1}\otimes |0\rangle _{2}  \notag \\
|1\rangle _{1}\otimes |1\rangle _{2} &\rightarrow &e^{2i\phi }|1\rangle
_{1}\otimes |1\rangle _{2} .
\end{eqnarray}%
Observe that the basis states $|0\rangle _{1}\otimes |1\rangle _{2}$ and $%
|1\rangle _{1}\otimes |0\rangle _{2}$ acquire the same phase, hence
experience the same error. Let us define encoded states by $|0_{L}\rangle
=|0\rangle _{1}\otimes |1\rangle _{2}\equiv |01\rangle $ and $|1_{L}\rangle
=|10\rangle $. Then the state $|\psi _{L}\rangle =a|0_{L}\rangle
+b|1_{L}\rangle $ evolves under the dephasing process as 
\beq
|\psi _{L}\rangle \rightarrow a|0\rangle _{1}\otimes e^{i\phi }|1\rangle
_{2}+be^{i\phi }|1\rangle _{1}\otimes |0\rangle _{2}=e^{i\phi }|\psi
_{L}\rangle ,
\eeq%
and the overall phase thus acquired is clearly unimportant. This means that
the 2-dimensional subspace DFS$_{2}(1)=\mathrm{Span}\{|01\rangle ,|10\rangle
\}$ of the 4-dimensional Hilbert space of two qubits is \emph{%
decoherence-free}. The subspaces DFS$_{2}(2)=\mathrm{Span}\{|00\rangle \}$
and\index{DFS (decoherence-free subspace)} DFS$_{2}(0)=\mathrm{Span}\{|11\rangle \}$ are also (trivially) DF, since
they each acquire a global phase as well, $1$ and $e^{2i\phi }$
respectively. Since the phases acquired by the different subspaces differ,
there is no coherence \emph{between} the subspaces. You might want to pause at this point to think about the similarities and differences between the quantum case and the case of two classical coins.

For $N=3$ qubits a similar calculation reveals that the subspaces 
\bes
\begin{eqnarray}
\mathrm{DFS}_{3}(2) &=&\mathrm{Span}\{|001\rangle ,|010\rangle ,|100\rangle
\},\quad \\
\mathrm{DFS}_{3}(1) &=&\mathrm{Span}\{|011\rangle ,|101\rangle ,|110\rangle
\}
\end{eqnarray}%
\ees
are DF, as well the (trivial)\ subspaces $\mathrm{DFS}_{3}(3)=\mathrm{Span}%
\{|000\rangle \}$ and $\mathrm{DFS}_{3}(0)=\mathrm{Span}\{|111\rangle \}$.

By now it should be clear how this generalizes. Let $\lambda _{N}$ denote
the number of $0$'s\textrm{\ }in a computational basis state (i.e., a
bitstring) over $N$ qubits. Then it is easy to check that any subspace
spanned by states with constant $\lambda _{N}$ is DF against collective
dephasing, and can be denoted $\mathrm{DFS}_{N}(\lambda _{N})$ in accordance
with the notation above. The dimensions of these subspaces are given by the
binomial coefficients: $d_{N}\equiv \dim [\mathrm{DFS}_{N}(\lambda _{N})]={%
\binom{N}{\lambda _{N}}}$ and they each encode $\log _{2}d_{N}$ qubits. It
might seem that we lost a lot of bits in the encoding process. However,
consider the encoding rate $r\equiv \frac{\log _{2}d_{N}}{N}$, defined as
the number of output qubits to the number of input qubits. Using Stirling's
formula $\log_2 x!\approx (x+\frac{1}{2})\log_2 x-x$ we find, for the case of
the highest-dimensional\index{DFS (decoherence-free subspace)} DFS ($\lambda _{N}=N/2$, for even $N$):%
\beq
r\overset{N\gg 1}{\approx }1-\frac{1}{2}\frac{\log _{2}N}{N},
\label{a3:eq:coll-deph}
\eeq%
where we neglected $1/N$ compared to $\log (N)/N$. Thus,
the rate approaches $1$ with only a logarithmically small correction. 

\subsection{Decoherence-free subspaces in the Kraus OSR}
Let us now partition the Hilbert space into two subspaces $\mathcal{H}_S=\mathcal{H}_G\oplus\mathcal{H}_N$. The subspace of the Hilbert space not affected is denoted by $\mathcal{H}_G$, the ``good'' portion, and $\mathcal{H}_N$ denotes the decoherence-affected, ``noisy" subspace. The point of this decomposition is to establish a general condition under which $\mathcal{H}_G$ remains unaffected by the open system evolution, and evolves unitarily. If we can do this then we are justified in calling  $\mathcal{H}_G$ a DFS. We now make two assumptions:
\begin{itemize}
\item Assume that it is possible to partition the Kraus operators as
\beq
K_\alpha(t)=g_\alpha U\oplus N_\alpha = \left(
\begin{array}{cc}
g_\alpha U & 0 \\
0 & N_\alpha
\end{array}
\right) ,
\label{K-DFS}
\eeq
such that $U$ defines a unitary operator acting solely on $\mathcal{H}_G$, $g_\alpha\in\C$, and $N_{\alpha}$ is an arbitrary (possibly non-unitary) operator acting solely on $\mathcal{H}_N$. 
\item Assume that the initial state is partitioned in the same manner, as
\beq
\rho_S(0)=\rho_G(0)\oplus\rho_N(0) = \left(
\begin{array}{cc}
\rho_G(0) & 0 \\
0 & \rho_N(0)
\end{array}
\right) ,
\eeq
where $\rho_G(0) :\mathcal{H}_G \mapsto \mc{H}_G$ and $\rho_N(0) :\mathcal{H}_N \mapsto \mc{H}_N$.
\end{itemize}

Note that by the normalization of the Kraus operators,
\beq
\sum_\alpha |g_\alpha|^2=1, \quad \sum_\alpha N^{\dagger}_\alpha N_\alpha =I_N.
\eeq
Under these two assumptions the Kraus OSR [Eq.~\eqref{eq:rhomap}] then becomes
\beq
\rho_S\mapsto\rho_S'=\sum_{\alpha}\left(g_\alpha U\oplus N_\alpha\right)\rho_S(0)\left(g^*_\alpha U^{\dagger}\oplus B^{\dagger}_\alpha\right)=
\left(
\begin{array}{cc}
U\rho_G(0) U^{\dagger} & 0 \\
0 & \sum_\alpha N_\alpha \rho_N(0) N^{\dagger}_\alpha
\end{array}
\right) .
\eeq
The remarkable thing to notice about this last result is that $\r_G$ evolves purely unitarily, i.e., it satisfies the definition of decoherence-freeness. If we disregard the evolution in the ``noisy" subspace we thus have the following result:

\begin{mytheorem}
If the two assumptions above hold then the evolution of an open system that is initialized in the ``good" subspace $\mathcal{H}_G$ is decoherence-free, i.e., $\mathcal{H}_G$ is a DFS.
\end{mytheorem}

\subsection{Hamiltonian DFS}
Discussing DFS in terms of Kraus operators works well, but we'd like to develop a bottom-up understanding of the DFS concept, using Hamiltonian evolution.  Assume we are given a system in which our computation is occurring, and a bath that is connected to the system.  The Hamiltonian governing the whole system can be written as usual as
\beq
\label{eq:DFShamiltonian}
H = H_{S} \otimes I_{B} + I_{S} \otimes H_{B} + H_{SB},
\eeq
where $H_{S}$ acts only on the system we are interested in, $H_{B}$ acts only on the bath, and $H_{SB}$ governs the interaction between the two.  Assume also, without loss of generality, that the interaction Hamiltonian can be written as
\beq
\label{eq:HSB}
H_{SB} = \sum_{\alpha} S_{\alpha} \otimes B_{\alpha},
\eeq
where each $S_\alpha$ is a pure-system operator and each $B_\alpha$ is a pure-bath operator. The Hilbert space can be written $\mathcal{H} = \mathcal{H}_S \otimes \mathcal{H}_B$, and $\mathcal{H}_S = \mathcal{H}_G \otimes \mathcal{H}_N$ where
\bes
\bea
\mathcal{H}_G &=& \text{Span} \lbrace \ket{\gamma_i} \rbrace\\
\mathcal{H}_N &=& \text{Span} \lbrace \ket{\nu_k} \rbrace\\
\mathcal{H}_B &=& \text{Span} \lbrace \ket{\beta_j} \rbrace
\eea
\ees
To formulate a theorem we also need the following assumptions:
\begin{enumerate}
  \item The system state is initialized in the \textit{good} subspace:
  \beq
  \rho_{S} = \rho_{G}\oplus 0 = \sum_{i,j} r_{ij} \ketb{\gamma_i}{\gamma_j}\oplus 0
  \label{eq:stateInit}
  \eeq
  \item The basis states of the good subspace are eigenvectors of the interaction Hamiltonian:
  \beq
  S_{\alpha} \ket{\gamma_i} = c_{\alpha} \ket{\gamma_i},  c_{\alpha} \in \mathbb{C}
  \label{eq:HSBonGood}
  \eeq
  \item The basis states of the good subspace, when acted on by the system Hamiltonian, remain in the good subspace:
  \beq
  H_S \ket{\gamma_i} \in \mathcal{H}_G
  \label{eq:HSonGood}
  \eeq
\end{enumerate}

With these assumptions in hand and with $U\left(t\right) = e^{-iHt}$, we can posit the following theorem.
\begin{thm}
\label{thm:hamiltonianDFS} 
Assuming 1-3, the evolution of the open system can be written as
\bes
\bea
\rho_S(t) = \Tr_B \lbrack U(t) \left( \rho_S \left(0\right) \otimes \rho_B \left(0\right) \right)U^{\dagger}(t) \rbrack &=& {U}_S(t) \rho_G \left(0\right) {U}_S^{\dagger}(t) \\
{U}_S(t) \ket{\gamma_i} \in \mathcal{H}_G
\eea
\ees
where $U_S(t)= e^{-iH_S t}$.
\end{thm}
\begin{proof}
Using equations \eqref{eq:HSB} and \eqref{eq:HSBonGood} we can write
\bes
\bea
\left( I_S \otimes H_B + H_{SB}\right) \ket{\gamma_i}_S \otimes \ket{\beta_j}_B &=&
\ket{\gamma_i}\otimes \left( H_B \ket{\beta_j}\right) + \sum_{\alpha} c_{\alpha} \ket{\gamma_i} \otimes B_{\alpha} \ket{\beta_j} \\
&=& \ket{\gamma_i} \otimes \left( \sum_{\alpha} c_{\alpha} B_{\alpha} + H_B\right) \ket{\beta_j} \\
&=& \ket{\gamma_i} \otimes H_{B'}\ket{\beta_j},
\eea
\ees
where $H_{B'}$ acts only on the bath.  Applying this to Eq.~\eqref{eq:DFShamiltonian} we find that the complete Hamiltonian can be decomposed into a portion that acts only on the system and a portion that acts only on the bath.
\beq
H \ket{\gamma_i} \otimes \ket{\beta_j} = \left( H_S \otimes I_B + I_S \otimes H_{B'} \right) \ket{\gamma_i} \otimes \ket{\beta_j} .
\eeq
If we plug this form of the Hamiltonian into the unitary evolution matrix we get
\beq
U\left(t\right) = e^{-i\left( H_S \otimes I_B + I_S \otimes H_{B'} \right)t} = U_S(t) \otimes U_C(t) ,
\eeq
where $U_x(t) = \exp(-itH_x)$, $x=S,C$.

To find $\rho(t)$ we apply this unitary to $\rho(0)$ with $\rho_S(0) = \rho_G(0)\oplus 0$.
\bes
\bea
\rho(t) &=& U(\rho_S(0)\otimes\rho_B(0))U^{\dagger}\\
				&=& U_S(\rho_G(0)\oplus 0)U_S^{\dagger}\otimes U_C\rho_B(0)U_C^{\dagger}
\eea
\ees
We find the state of our system of interest $\rho_S(t)$ by taking the partial trace.
\bes
\bea
\rho_S(t) &=& \Tr_B\lbrack U_S(\rho_G(0)\oplus 0)U_S^{\dagger}\otimes U_C\rho_B(0)U_C^{\dagger} \rbrack \\
&=& U_S (\rho_G(0)\oplus 0) U_S^{\dagger} \\
&=& U_S^G\rho_G(0) (U_S^G)^{\dagger} ,
\eea
\ees
where in the last equality we projected $U_S$ to the good subspace, i.e., $U_S^G \equiv U_S |_{\mc{H}_G}$.
\end{proof}
Thus Theorem~\ref{thm:hamiltonianDFS} guarantees that if its conditions are satisfied, a state initialized in the DFS will evolve unitarily.

\subsection{Deutsch's Algorithm}
As a first example of error avoidance utilizing the DFS construction, we can consider the first known algorithm that offers a quantum speed-up. Deutsch's Algorithm \cite{nielsen2000quantum} presents a simple decision problem in which the goal is to decide whether a function is constant or balanced. Let $f(x):\{0,1\}^n\mapsto\{0,1\}$ (in decimal notation $x\in \{0,1,\ldots,2^n-1\}$), denote the function, where if
\beq
f(x)=\left\{
\begin{array}{ll}
0, & \forall\, x \textrm{ or}\\
1, & \forall\, x\\
\end{array}\right.
\label{eq:Ceqs}
\eeq
then the function is called constant and if
\beq
f(x)=\left\{
\begin{array}{cc}
0, & \text{half the inputs}\\
1, & \text{half the inputs}\\
\end{array}\right.
\label{eq:Beqs}
\eeq
then the function is called balanced. Classically, we find that in the worst case, making a decision on whether $f(x)$ is constant or balanced requires a minimum of $2^n/2+1$ total queries to $f$. Deutsch and Jozsa showed that the exponential cost in $f$-queries is drastically reduced -- to just a single query! -- by considering a quantum version of the algorithm.

The decision problem can be recast in terms of the following quantum circuit:
\bigskip
\[
\Qcircuit @C=1em @R=0.7em {
      \lstick{\ket{0}} & {/^n} \qw & \qw &\gate{W^{\otimes n}} & \qw & \qw &	 \multigate{1}{U_f} & \qw & \qw &\gate{W^{\otimes n}} & \qw & \qw & \meter  \\
      \lstick{\ket{1}} & \qw & \qw &\gate{W} & \qw & \qw &	\ghost{U} & \qw & \qw & \qw & \qw & \qw & \qw\\
      								 & &\lstick{\ket{\psi_1}}& & &\lstick{\ket{\psi_2}}	 & & & \lstick{\ket{\psi_3}} & & & \lstick{\ket{\psi_4}}
}
\]
\newline\newline
where each classical bit $n$ corresponds to a qubit. The unitary operator $U$ performs the query on $f(x)$ by
\beq
U_f:\ket{x}\ket{y}\mapsto \ket{x}\ket{y\oplus f(x)}\textrm{ (addition mod 2)} ,
\label{eq:U_f}
\eeq
where the first register ($x$) contains the first $n$ qubits and the second register ($y$) contains the last qubit. Here $W$ represents the Hadamard gate: $W\ket{0}=\ket{+}$ and $W\ket{1}=\ket{-}$, where $\ket{\pm}=(\ket{0}\pm\ket{1})/\sqrt{2}$.

In order to illustrate Deutsch's algorithm for the quantum circuit above, consider the single qubit version ($n=1$). In this case there are four functions, two of which are constant and two of which are balanced: $\{f_0(x)=0, f_1(x)=1\}$ (constant), $\{f_2(x)=x,f_3(x)=\bar{x}\}$ (balanced), where the bar denotes bit negation. Clearly two classical queries to $f$ are required to tell whether $f$ is constant or balanced.

Initially the total system state is given by $\ket{\psi_1}=\ket{0}\ket{1}$. Applying the Hadamard gate, the system state becomes $\ket{\psi_2}=\ket{+}\ket{-}$. Applying the unitary operator $U_f$, the resulting state is
\bes
\begin{eqnarray}
\ket{\psi_3}=U_f\ket{+}\ket{-}&=&U_f\left[\frac{1}{2}\left(\ket{00}-\ket{01}+\ket{10}-\ket{11}\right)\right]\\
&=&\frac{1}{2}\left(\ket{0,0\oplus f(0)}-\ket{0,1\oplus f(0)}+\ket{1,0\oplus f(1)}-\ket{1,1\oplus f(1)}\right)\\
&=&\frac{1}{2}\left(\ket{0,f(0)}-\ket{0,\bar{f}(0)}+\ket{1,f(1)}-\ket{1,\bar{f}(1)}\right).
\end{eqnarray}
\ees
Applying each constant and balanced function to $\ket{\psi_3}$, we find
\beq
\ket{\psi_3}=\left\{
\begin{array}{cc}
f_0:&\ket{00}-\ket{01}+\ket{10}-\ket{11}=+\ket{+-}\\
f_1:&\ket{01}-\ket{00}+\ket{11}-\ket{10}=-\ket{+-}\\
f_2:&\ket{00}-\ket{01}+\ket{11}-\ket{10}=+\ket{--}\\
f_3:&\ket{01}-\ket{00}+\ket{10}-\ket{11}=-\ket{--}
\end{array}\right.
\eeq
for $f_j(x)\in\{0,1,x,\bar{x}\}$ as defined by Eqs.~(\ref{eq:Ceqs}) and (\ref{eq:Beqs}). The remaining Hadamard gate yields the final state
\beq
\ket{\psi_4}=\left\{
\begin{array}{cc}
f_0:&+\ket{0-}\\
f_1:&-\ket{0-}\\
f_2:&+\ket{1-}\\
f_3:&-\ket{1-}
\end{array}\right.
\label{psi4}
\eeq
and the characteristic of the function is determined by measuring the $1$st qubit: a result of $0$ indicates a constant function, a result of $1$ a balanced function. Thus, remarkably, we find that the quantum version only requires a {\em single} query to the function, while the classical case requires two queries (this scenario is the original Deutsch algorithm).

The circuit depicted above can be subjected to a similar analysis in the $n$-qubit case (the Deutsch-Jozsa algorithm) and the conclusion is that the quantum version of the algorithm still requires only a single $f$-query, thus resulting in an exponential speed-up relative to its classical counterpart in the worst case.

\subsection{Deutsch's Algorithm With Decoherence}
To gain an understanding of how a DFS works we can look at the Deutsch Problem with added decoherence.  We can consider the circuit diagram for the single qubit Deutsch algorithm, but introduce a dephasing element as follows, where the dotted box denotes dephasing on the top qubit only:
\bigskip
\[
\Qcircuit @C=1em @R=0.7em {
      \lstick{\ket{0}} & \qw & \qw &\gate{W} & \qw & \gate{Z} & \qw &	 \multigate{1}{U_f} & \qw & \qw &\gate{W} & \qw & \qw & \meter  \\
      \lstick{\ket{1}} & \qw & \qw &\gate{W} & \qw & \gate{I} & \qw &	 \ghost{U} & \qw & \qw & \qw & \qw & \qw & \qw\\
      								 & &\lstick{{\r_1}}& & &\lstick{\r_2}	& & \lstick{{\r_2'}} & & \lstick{{\r_3}} & & & \lstick{{\r_4}} \gategroup{1}{6}{2}{6}{0.5em}{.}
}
\]
\newline\newline
The Kraus Operators governing the dephasing of $ \rho_{2} $ are:
\bes
\bea
\label{eq:dfsK0}
K_{0} &=& \sqrt{1-p} I_{1} \otimes I_{2}\\
\label{eq:dfsK1}
K_{1} &=& \sqrt{p} Z_{1} \otimes I_{2}
\eea
\ees
With probability $(1-p)$ nothing happens.  However with probability $p$ the first qubit experiences dephasing.
If we follow the density matrix states through the algorithm we can see the effect this dephasing has on our result.  As before we have $\rho_{1} = \ket{01} \bra{01}$ and, $\rho_{2} = \ket{+-} \bra{+-}$. By applying the Kraus operators we find the state after dephasing to be
\bes
\bea
\rho_{2}'	&=& K_{0} \rho_{2} K_{0}^{\dagger} + K_{1} \rho_{2} K_{1}^{\dagger}\\
				&=& (1-p)\rho_{2} + p (Z \otimes I) \ket{+-} \bra{+-} (Z \otimes I)^{\dagger} .
\eea
\ees
It is easy to check that $ Z \ket{+} = \ket{-} $, thus we find
\beq
\rho_{2}'	=	(1-p) \ket{+-} \bra{+-} + p \ket{--} \bra{--} .
\eeq
Using Eq.~\eqref{eq:U_f} we can compute $\rho_3$
\beq
\rho_{3} = \begin{cases}
(1-p) \ket{+-} \bra{+-} + p \ket{--} \bra{--} \quad f_0,f_1 \textrm{ (constant)}\\
(1-p) \ket{--} \bra{--} + p \ket{+-} \bra{+-} \quad f_2,f_3 \textrm{ (balanced)}
\end{cases} .
\eeq
After applying the final Hadamard we see that the state we will measure is
\beq
\rho_{4} = \begin{cases}
(1-p) \ket{0-} \bra{0-} + p \ket{1-} \bra{1-} \quad f_0,f_1 \textrm{ (constant)}\\
(1-p) \ket{1-} \bra{1-} + p \ket{0-} \bra{0-} \quad f_2,f_3 \textrm{ (balanced)}
\end{cases}.
\eeq
If we now measure the first qubit to determine whether the function is constant or balanced, with probability $p$ we will misidentify the outcome.  For example, if we obtain the outcome $1$, with probability $p$ this could have come from the constant case. But, according to Eq.~\eqref{psi4} the outcome $1$ belongs to the balanced case.

It is possible to overcome this problem by use of a DFS.  Let us again modify the original circuit design.  We can add a third qubit and then define logical bits and gates.
\bigskip
\[
\Qcircuit @C=1em @R=0.7em {
      \lstick{\ket{\bar{0}}} & \qw {/^2}& \qw &\gate{W_L} & \qw & \gate{ZZ} & \qw &	 \multigate{1}{U_{fL}} & \qw & \qw &\gate{W_L} & \qw & \qw & \meter  \\
      \lstick{\ket{1}} & \qw & \qw &\gate{W} & \qw & \gate{I} & \qw &	 \ghost{U} & \qw & \qw & \qw & \qw & \qw & \qw\\
      								 & &\lstick{{\r_1}}& & &\lstick{\r_2}	& & \lstick{{\r_2'}} & & \lstick{{\r_3}} & & & \lstick{{\r_4}} \gategroup{1}{6}{2}{6}{0.5em}{.}
}
\]
\bigskip

Now the $Z$ dephasing acts simultaneously on both top qubits, that comprise the logical qubit $\ket{\bar{0}}$. In this case the Kraus operators are
\bea
K_{0} = \sqrt{1-p} III,\qquad
K_{1} = \sqrt{p} ZZI\ ,
\eea
where $ ZZI = \sigma^{z} \otimes \sigma^{z} \otimes I $.  Recall the requirements for a DFS.  The Kraus operators, as in Eq.~\eqref{K-DFS}, must be of the form
\beq
K_{\alpha} = \left(
\begin{array}{c|c}
g_{\alpha}U & 0\\ \hline
0 & B_{\alpha}
\end{array}
\right)
\label{K-DFS2}
\eeq
and the state must be initialized in a \textit{good} subspace, i.e., $ \rho_S = \rho_{G} \oplus \rho_{N}$, where the direct sum reflects the same block structure as in Eq.~\eqref{K-DFS2}.  If these conditions are met then $ \rho_S' = \sum K_{\alpha} \rho_S K_{\alpha}^{\dagger} = U \rho_G U^{\dagger} \oplus \r'_N$.  In other words the evolution of $\r_G$ is entirely unitary.  We start by checking the matrix form of the Kraus operators.  $K_{0}$ is simply the identity and trivially satisfies this condition.  We can check the $ZZ$ portion of $K_{1}$ since that is what will act on our logical qubit.
\beq
ZZ  =
\begin{pmatrix}
1 & 0 & 0 & 0\\
0 & -1 & 0 & 0\\
0 & 0 & -1 & 0 \\
0 & 0 & 0 & 1
\end{pmatrix}
\begin{array}{c}
00 \\ 01 \\ 10 \\ 11
\end{array}.
\eeq
This obviously doesn't fit the required matrix format, in that there is no block of $1$'s like in $K_0$.  However with a simple reordering of the basis states we obtain the following matrix
\beq
ZZ  =
\begin{pmatrix}
1 & 0 & 0 & 0\\
0 & 1 & 0 & 0\\
0 & 0 & -1 & 0 \\
0 & 0 & 0 & -1
\end{pmatrix}
\begin{array}{c}
00 \\ 11 \\ 01 \\ 10
\end{array}.
\eeq
Now we have a $2\times 2$ matrix of $1$'s, so both $ZZ$ and the identity matrix act as the same unitary on the subspace spanned by $\ket{00}$ and $\ket{11}$. The full matrix, $K_{1}$, then takes the form
\beq
K_{1}  =\sqrt{p}
\begin{pmatrix}
I_{2} & 0 & 0 & 0\\
0 & I_{2} & 0 & 0\\
0 & 0 & -I_{2} & 0 \\
0 & 0 & 0 & I_{2}
\end{pmatrix}
\begin{array}{c}
00i \\ 11i \\ 01i \\ 10i
\end{array}
\eeq
where $I_2$ denotes the $2\times 2$ identity matrix, $i = \left\lbrace 0,1 \right\rbrace$, and $K_{0}$ is the $8\times 8$ identity matrix.  Thus we see that both $K_{0}$ and $K_{1}$ have the same upper block format, namely $U = I_{4x4}$, where $g_{0} = \sqrt{1-p}$ and $g_{1} = \sqrt{p}$.  Now we can define our logical bits $\ket{\bar{0}} = \ket{00}$ and $\ket{\bar{1}} = \ket{11}$.  With these states we can construct our logical Hadamard.
\beq
\left(W_{L}\right)_{4x4} = \left(\begin{array}{c|c}
W & 0 \\ \hline
0 & V
\end{array}\right) .
\eeq
Here the logical Hadamard acts as a regular Hadamard on our logical qubits.
\bes
\bea
W_{L} \ket{\bar{0}} = \ket{+_{L}}\\
W_{L} \ket{\bar{1}} = \ket{-_{L}}
\eea
\ees
where $\ket{\pm_{L}} = \frac{1}{\sqrt{2}} \left( \ket{\bar{0}} \pm \ket{\bar{1}} \right)$. The other unitary action of $H_{L}$, namely $V$, we don't care about.  Similarly we can construct a logical $U_{f}$.
\beq
\left(U_{fL}\right)_{8x8} = \left(\begin{array}{c|c}
U_{f} & 0 \\ \hline
0 & V'
\end{array}\right) .
\eeq
Again, $U_{fL}$ acts as $U_{f}$ on our logical bits, and $V'$ we don't care about. Neither $V$ nor $V'$ affect our logical qubits in any way.
Now that we have set our system up we can apply the Deutsch algorithm again to see if the DFS corrects the possibility of misidentifying the result.
Our system begins in the state $\rho_{1} = \ket{\bar{0}} \ket{1}\bra{\bar{0}}\bra{1}$ and after applying the logical Hadamard we get $\rho_{2} = \ket{+_{L}} \ket{-}\bra{+_L}\bra{-}$.  Now we can apply the Kraus operators to see the effect of the decoherence. $K_{0}$ has no effect other than to multiply the state by $\sqrt{1-p}$ because it is proportional to the identity matrix. It is enough to examine the effect of $K_{1}$ on the state $\ket{+_{L}}$.
\bes
\bea
K_{1}\ket{+_{L}} &=& \dfrac{1}{\sqrt{2}}\left( K_{1}\ket{\bar{0}} + K_{1}\ket{\bar{1}} \right)\\
				&=&  \dfrac{1}{\sqrt{2}}\sqrt{p}\left( I \ket{\bar{0}} + I \ket{\bar{1}} \right)\\
%				&=& \dfrac{1}{\sqrt{2}}\left(\ket{\bar{0}} + \ket{\bar{1}} \right)\\
				&=& \sqrt{p} \ket{+_{L}}
\eea
\ees
Therefore $\rho_{2}' = \rho_{2}$.  The decoherence has no effect on our system and the rest of the algorithm will proceed without any possibility of error in the end.

\section{Collective Dephasing}
\label{sec:CD}

\subsection{The model}
Consider the example of a spin-boson Hamiltonian. In this example, the system of qubits could be the spins of $N$ electrons trapped in the periodic potential well of a crystalline lattice. The bath is the phonons of the crystal (its vibrational modes). We also assume that the system-bath interaction has \emph{permutation symmetry} in the sense that the interaction between the spins and phonons is the same for all spins, e.g., because the phonon wavelength is long compared to the spacing between spins. This assumption is crucial for our purpose of demonstrating the appearance of a DFS. If the potential wells are deep enough then the motional degrees of freedom of the electrons can be ignored.
Let $i$ denote the index for the set of $N$ electrons in the system (the same as the index for the set of occupied potential wells in the solid), let $k$ denote the vibrational mode index, $b_k^{\dagger}\ket{n_1,\dots,n_k,\dots}=\sqrt{n_k+1}\ket{n_1,\dots,n_k+1,\dots}$ is the action of the creation operator for mode $k$ on a Fock state with occupation number $n_k$, $b_k\ket{n}=\sqrt{n_k}\ket{n_1,\dots,n_k-1,\dots}$ is the action of the annihilation operator for mode $k$; $b_k^{\dagger}b_k$ is the number operator, i.e., $b_k^{\dagger}b_k\ket{n_1,\dots,n_k,\dots}=n_k\ket{n_1,\dots,n_k,\dots}$. With $\sigma_i^z$ the Pauli-$z$ spin operator acting on the $i^{th}$ spin, the system-bath Hamiltonian is
\beq
{H}_{SB}= \sum_{i,k} g_{i,k}^z \sigma_i^z\otimes(b_k + b_k^\dagger) + h_{i,k}^z \sigma_i^z\otimes b_k^\dagger b_k.
\eeq
The permutation symmetry assumption implies
\beq
g_{i,k}^z=g_k^z \ ,\ h_{i,k}^z=h_k^z ,
\eeq
i.e., the coupling constants do not depend on the qubit index. The system-bath Hamiltonian can then be written
\bes
\bea
{H}_{SB}&=& \sum_{i} \sigma_i^z \otimes \sum_{k}g_k^z (b_k + b_k^\dagger) + h_k^z b_k^\dagger b_k\\
&=& S_z \otimes B_z ,
\eea
\ees
where
\beq
S_z \equiv \sum_{i} \sigma_i^z,\quad B_z \equiv \sum_{k}g_k^z (b_k + b_k^\dagger) + h_k^z b_k^\dagger b_k.
\eeq
If these conditions are met, the bath acts identically on all qubits and system-bath Hamiltonian is invariant under permutations of the qubits' order. The operator $S_z$ is a \emph{collective} spin operator.

\subsection{The DFS}
We will now see that this model results in a DFs that is essentially identical to the one we saw in Section~\ref{CD-DFS}.

First consider the case of $N=2$. In light of the DFS condition Eq.~\eqref{eq:HSBonGood}:
\beq
N=2\Rightarrow S_z=Z\otimes I+I\otimes Z.
\eeq
Thus
\bea
\begin{array}{ccll}		
	\ket{00} &\stackrel{S_z}{\rightarrow}& 2 \cdot \ket{00}&\Rightarrow c_z = 2\\
	\ket{01} &\stackrel{S_z}{\rightarrow}& \ket{01} - \ket{01} = 0&\Rightarrow c_z = 0\\
	\ket{00} &\stackrel{S_z}{\rightarrow}& \ket{10} - \ket{10} = 0&\Rightarrow c_z = 0\\
	\ket{00} &\stackrel{S_z}{\rightarrow}& -2 \cdot \ket{11}&\Rightarrow c_z = -2\\	
\end{array}
\eea
It follows that the DFS's for the two spins are:
\bea
\begin{matrix}
	\dfsh_{N=2}(2)&=&\{\ket{00}\}\\
	\dfsh_{N=2}(0)&=&\textrm{Span}\{\ket{01},\ket{10}\}\\
	\dfsh_{N=2}(-2)&=&\{\ket{11}\}\\
\end{matrix}
\eea
where we used the notation $\dfsh_N(c_z)$ to denote the ``good'' subspace $\mathcal{H}_G$ for $N$ qubits, with eigenvalue $c_z$.

In the $\dfsh_{N=2}(0)$ DFS, there are two states, so we have an encoded qubit:
\bes
\label{eq:DFS0110}
\bea
  \ket{\overline{0}}&=\ket{01}\quad\textnormal{logical 0}\\
  \ket{\overline{1}}&=\ket{10}\quad\textnormal{logical 1}
\eea
\ees
For three spins we have:
\beq
N=3\Rightarrow S_z=Z\otimes I\otimes I\ +\ I\otimes Z\otimes I\ +\ I\otimes I\otimes Z
\eeq
Similarly, it follows that the DFS's for the three spins are:
\bea
\begin{array}{lcc}
	\dfsh_{N=3}(3)&=&\{\ket{000}\}\\
	\dfsh_{N=3}(1)&=&\textrm{Span}\{\ket{001},\ket{010},\ket{100}\}\\
	\dfsh_{N=3}(-1)&=&\textrm{Span}\{\ket{011},\ket{101},\ket{110}\}\\
	\dfsh_{N=3}(-3)&=&\{\ket{111}\}\\
\end{array}
\eea

We find that there are two possible encoded qutrits for $N=3$, one in $\dfsh_{N=3}(1)$ and the other in $\dfsh_{N=3}(-1)$.

In general, the DFS $\dfsh_N(c_z)$ is the eigenspace of each eigenvalue of $S_z$. It is easy to see that the number of spin-ups ($0$'s) and the number of spin-downs ($1$'s) in each eigenstate is constant throughout a given eigenspace. This corresponds to the value of total spin projection along $z$. In fact, for arbitrary $N$,
\beq
c_z = \# 0-\# 1.
\eeq

Figure~\ref{fig:Bratteli}, known as the \emph{Bratteli Diagram}, shows the eigenvalues ($y$-axis) of $S_z$ as $N$ ($x$-axis) increases. It represents the constellation of DF subspaces in the parameter-space, $(N,c_z)$.
\begin{figure}[t]
		\includegraphics{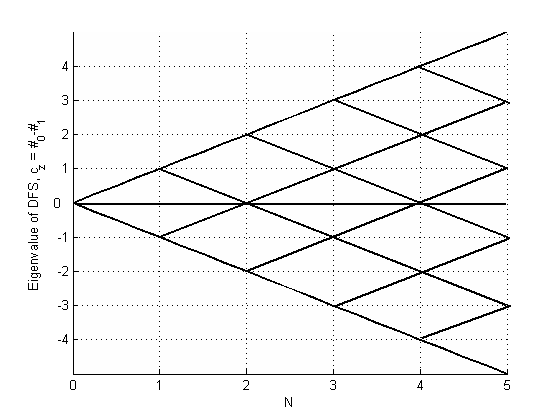}
		\caption{Bratteli diagram showing DFSs for each $N$. Each intersection point corresponds to a DFS. The number of paths leading to a given point is the dimension of the corresponding DFS. The diagram shows, e.g., that there is a $2$-dimensional DFS for $N=2$ at $c_z = 0$, yielding one encoded qubit.
}
	\label{fig:Bratteli}
\end{figure}
Each intersection point in the figure represents a DFS. Each upward stroke on the diagram indicates the addition of one new spin-up particle, $\ket{0}$, to the system, while each downward stroke indicates the addition of a new spin-down particle, $\ket{1}$, to the system. The number of paths from the origin to a given intersection point on the diagram is therefore exactly the  dimension of the eigenspace with eigenvalue $c_z$, which is given by
\beq 
\dim(\dfsh_N(c_z))={N \choose \# 0},
\eeq
and note that since $N=\# 0+\# 1$, we have $\# 0 = (N+c_z)/2$. The highest dimensional DFS for each $N$ is thus
\bea
	\max_{c_z}\{\dim(\dfsh_N(c_z)) \}= \left\{
  \begin{array}{ll}
  	{N \choose \frac{N}{2}}  & N \textrm{even}\\
  	{N \choose {\frac{N\pm 1}{2}}}& N \textrm{odd}
  \end{array}
  \right.
  \label{eq:maxcode}
\eea
Given a $D$-dimensional DFS, $\dfsh$:
\beq
 \textrm{\# of DFS qubits in }\dfsh =  \log_2 D  .
 \eeq
We can use this to calculate the rate of the DFS code, i.e.,
\bes
\begin{align}
r &\equiv \frac{\textrm{\# of DFS qubits in }\dfsh}{\textrm{\# of physical qubits}} = \frac{\log_2 D}{N} \\
&\overset{N\gg 1}{\approx }1-\frac{1}{2}\frac{\log _{2}N}{N},
\label{eq:rate}
\end{align}
\ees
where in the second line we used Eq.~\eqref{a3:eq:coll-deph}. For a DFS of given dimension it may be preferable to think in terms of qudits rather than qubits. For example, for the DFS $\dfsh_3(-1)$, the dimension $D=3$, and so this DFS encodes one qutrit.

Any superposition of states in the same DFS will remain unaffected by the coupling to the bath, since they all share the same eigenvalue of $S_z$, and hence only acquire a joint overall phase under the action of $S_z$. But a superposition of states in different DFSs, i.e., eigenspaces of $S_z$, will not evolve in a decoherence-free manner, since they will acquire relative phases due to the different eigenvalues of $S_z$.

We don't \emph{have} to build a multi-qubit DFS out of the largest good subspace. In some cases it makes sense to sacrifice the code rate to gain simplicity or physical realizability. We shall see an example of this in the next subsection.

\subsection{Universal Encoded Quantum Computation}
\label{sec:UEQC}

From here on, encoded qubits will be called `logical qubits'. To perform arbitrary single qubit operations, we need to be able to apply any two of the Pauli operators on the logical qubits. On our example system, $\dfsh_2(0)$, we have:
\begin{eqnarray*}
	\ket{\bar{0}}&\equiv&\ket{01}\\
	\ket{\bar{1}}&\equiv&\ket{10}\\
	\ket{\bar{\phi}}&:&\textnormal{DFS encoded logical qubit, $a\ket{\bar{0}}+b\ket{\bar{1}}$}\\
	\bar{U}&:&\textnormal{Logical operator on the DFS qubits}\\
\end{eqnarray*}
Thus the logical Pauli-$z$ operator is
\bea
\left.
	\begin{array}{lcr}
		\bar{Z}\ket{\bar{0}} &=& \ket{\bar{0}}\\
		\bar{Z}\ket{\bar{1}} &=&-\ket{\bar{1}}
	\end{array} \right\}
	\Rightarrow \bar{Z} = Z\otimes I \Rightarrow
	\left\{
	\begin{array}{lrr}
		\ket{01}\stackrel{Z\otimes I}{\rightarrow}&\ket{01}=&\ket{\bar{0}}\\
		\ket{10}\stackrel{Z\otimes I}{\rightarrow}&-\ket{10}=&-\ket{\bar{1}}
	\end{array}
	\right. ,
\eea
and the logical Pauli-$x$ operator is
\bea
\left.
	\begin{array}{lcr}
		\bar{X}\ket{\bar{0}} &=& \ket{\bar{1}}\\
		\bar{X}\ket{\bar{1}} &=& \ket{\bar{0}}
	\end{array} \right\}
	\Rightarrow \bar{X} = X\otimes X \Rightarrow
	\left\{
	\begin{array}{lrr}
		\ket{01}\stackrel{X\otimes X}{\rightarrow}&\ket{10}=&\ket{\bar{1}}\\
		\ket{10}\stackrel{X\otimes X}{\rightarrow}&\ket{01}=& \ket{\bar{0}}
	\end{array}
	\right. .
\eea

In general, suppose we have $2N$ physical qubits all experiencing collective dephasing. We can pair them into $N$ logical qubits, each pair in $\dfsh_2(0)$, and perform $Z$ or $X$ logical operations on the $i^{\rm th}$ logical qubit using the following operators:
\bes
\label{eq:logXZ}
\bea
\bar{Z}_i &\equiv& Z_{2i-1}\otimes I_{2i} ,\\
\bar{X}_i &\equiv& X_{2i-1}\otimes X_{2i}.
\eea
\ees
Note that by using this pairing we obtain a code whose rate is $N/(2N) = 1/2$, which is substantially less than the highest rate possible, when we use the codes specified by Eq.~\eqref{eq:maxcode}. However, the sacrifice is well worth it since we now have simple and physically implementable encoded logical operation involving at most 2-body interactions. Moreover, it is easy to see that the tensor product of two-qubit DFSs is itself still a DFS (just check that all basis states in this tensor product space still satisfy the DFS condition, i.e., have the same eigenvalue under the action of $S_z$).

We can define arbitrary rotations about the logical $X$ or $Z$ axis as $R_{\bar{X}}(\theta) = \exp[i\bar{X}\theta]$ and $R_{\bar{Z}}(\phi) = \exp[i\bar{Z}\phi]$. An arbitrary single logical qubit rotation (an arbitrary element of SU(2)) can then be obtained using the Euler angle formula, as a product of three rotations:
$R_{\bar{X}}(\theta_2)R_{\bar{Z}}(\phi)R_{\bar{X}}(\theta_1)$.

To generate arbitrary operators on multiple qubits, we need to add another gate to the generating set: the controlled phase gate. The logical controlled phase gate can be generated from  $\bar{Z}_i\otimes \bar{Z}_j \equiv Z_{i} \otimes Z_{2j-1}$. Thus a Hamiltonian of the form
\beq
\bar{H}_S = \sum_{i}\omega_{Z_i}(t) \bar{Z_i} + \sum_{i} \omega_{X_i}(t)\bar{X_{i}} + \sum_{i\ne j} \Omega_{ij}(t) \bar{Z_{i}}\otimes \bar{Z_{j}}
\eeq
not only does not take the encoded information outside the DFS $\dfsh_2(0)$ of each of the $N$ encoded qubits, i.e., satisfies the DFS preservation condition Eq.~\eqref{eq:HSonGood}, it is also sufficient to generate a universal set of logical gates over the logical DFS qubits. Moreover, this Hamiltonian is composed entirely of one- and two-body physical qubit operators, so it is physically implementable.

\section{Collective Decoherence and Decoherence Free Subspaces}
\label{sec:CDec}

The collective dephasing model can be readily modified to give the more general collective decoherence model. The interaction Hamiltonian has the following form:
\begin{eqnarray}
	H_{SB} &=&\sum_{i,k} \left[
			g_{i,k}^z\sigma_i^z\otimes(b_k + b_k^\dagger) +
			g_{i,k}^+\sigma_i^+\otimes b_k +
			g_{i,k}^-\sigma_i^-\otimes b_k^\dagger
			\right] \label{eq:HSB_CD}
\end{eqnarray}
where
\beq
	\sigma^{\pm}= \frac{1}{2} \left(\sigma_x \mp i\sigma_y\right) \label{sigma_pm} ,
\eeq
corresponds to the raising $(+)$ and lowering $(-)$ operators respectively, i.e.,
\bes
\begin{eqnarray}
\sigma^{+} \ket{0} &=& \ket{1} ,\ \sigma^{+} \ket{1} = 0 \\
\sigma^{-} \ket{0} &=& 0,  \ \sigma^{-} \ket{1} = \ket{0}
\end{eqnarray}
\ees
where $0$ here corresponds to the null vector and should not be confused with the $\ket{0}$ state. Thus, $\sigma^+ = | 1 \rangle \langle 0 |/2$ and $\sigma^- = | 0 \rangle \langle 1 |/2$ and the factor of $1/2$ is important for the rules of angular momentum addition we shall use below. Thus all the Pauli matrices in this section also include factors of $1/2$. 

The first term in the summation of \eqref{eq:HSB_CD} corresponds to an energy conserving (dephasing) term while the second and third terms correspond to energy exchange via, respectively, phonon absorption/spin excitation, and spin relaxation/phonon emission.

By assuming that all qubits are coupled to the same bath, thereby introducing a permutation symmetry assumption, we have 
\bes
\begin{eqnarray}
	\quad g_{ik}^\alpha&=&g_{k}^\alpha, \quad \forall k, \quad \alpha \in \{+,-,z\}\\
	\Rightarrow H_{SB}&=&\sum_i{\sigma_i^+} \otimes \overbrace{\sum_k{g_k^+b_k}}^{B_+}
		+ \sum_i{\sigma_i^-} \otimes \overbrace{\sum_k{g_k^- b_k^\dagger}}^{B_-}
		+\sum_i{\sigma_i^z} \otimes \overbrace{\sum_k{g_k^z (b_k+b_k^\dagger)}}^{B_z}\\
	&=&\sum_{\alpha \in \{+,-,z\}}{S_\alpha \otimes B_\alpha},
\end{eqnarray}
\label{eq:CD}
\ees
where
\beq
S_\alpha= \displaystyle \sum_{i=1}^{N}{\sigma_i^\alpha} \label{eq:totspin}
\eeq
is the total spin operator acting on the entire system of $N$ physical qubits.
We can derive the following relations directly from the commutation relations of the Pauli matrices:
\beq
	\left.
	\begin{array}{lcr}
		\left[S_\pm, S_z\right]&=& \pm2 S_\pm\\
		\left[S_-, S_+\right]&=& S_z
	\end{array}
	\right\}\parbox[c]{0.5 \linewidth}{commutation relations for SL(2) triple}\\
	\label{eq:commute}
\eeq
where SL(2) is a Lie algebra \cite{Hammermesh:grouptheorybook}.

We wish to define the total angular momentum operator $\vec{S}^2$ in terms of the angular momenta operators around each axis. It will be convenient to define the vector of angular momenta: $\vec{S} = \left( S_x, S_y, S_z\right)$ where $S_x \equiv S_+ + S_-$ and $S_y\equiv i(S_+ - S_-)$. We note that $ \vec{S}^2 \equiv \vec{S}\cdot \vec{S} = \sum_{\alpha \in \{x,y,z\}} {S_\alpha^2}$ satisfies $[\vec{S}^2, S_z] = 0$. Since $\vec{S}^{2}$ and $S_{z}$ commute and are both Hermitian, they are simultaneously diagonalizable, i.e., they share a common orthonormal eigenbasis.

Recalling some basic results from the quantum theory of angular momentum, we note that for the basis $\left\{ \ket{S,m_S}\right\}$ where $S$ represents the total spin quantum number of $N$ spin-$1/2$ particles and $m_S$ represents the total spin projection quantum number onto the $z$-axis, we can show that
\bes
\bea
\vec{S}^2\ket{S,m_S} &\equiv& S(S+1)\ket{S,m_S}\\
	S_z\ket{S,m_S} &\equiv& m_S\ket{S,m_S}
\eea
\ees
where
\beq S \in \left\{0, \frac{1}{2}, 1, \frac{3}{2},2, \dots , \frac{N}{2}\right\}\eeq
and
\beq m_S \in \left\{-S,-S+1,\ldots, S-1,S\right\}.
\eeq
Keeping in mind that the basis states of the good subspace are eigenvectors of the interaction Hamiltonian and also satisfy Eq.~\eqref{eq:HSBonGood}
%  \beq
%S_{\alpha} \ket{\gamma_i} = c_{\alpha} \ket{\gamma_i},  c_{\alpha} \in \mathbb{C} \label{eq:HSBonGood}
%\eeq
for $\alpha \in \{+,-,z\}$, let us examine the cases $N = 1,2,3$ and $4$ in turn.

\subsection{One Physical Qubit}
For a single physical qubit ($N = 1$), the basis $\left\{\ket{0}, \ket{1}\right\}$ corresponds to that of our familiar spin-$\frac{1}{2}$ particle, with $S=\frac{1}{2}$ and $m_S = \pm \frac{1}{2}$. We identify our logical zero and one states as follows

\bes
\begin{eqnarray}
\ket{0} &=& \ket{S=\frac{1}{2},m_S = \frac{1}{2}} \\
\ket{1} &=& \ket{S=\frac{1}{2},m_S = -\frac{1}{2}} 
\end{eqnarray}
\ees

\subsection{Two Physical Qubits}
For two physical qubits ($N = 2$), which we label $A$ and $B$, with individual spins $S_A =\frac{1}{2}$ and $S_B =\frac{1}{2}$, we first note that the prescription for adding angular momentum (or spin) given $\vec{S_A}$ and $\vec{S_B}$, is to form the new spin operator $\vec{S} = \vec{S_A} + \vec{S_B}$ with eigenvalues
\beq S \in \left\{\left|S_A - S_B \right|, \dots , S_A + S_B \right\}\eeq
with the corresponding spin projection eigenvalues
\beq m_S \in \left\{-S, \ldots, S\right\}.\eeq

Thus, for two physical qubits, we see that the total spin eigenvalues $S_{\left(N=2\right)}$ can only take the value $0$ or $1$. For $S_{\left(N=2\right)} =0$, we see that $m_S$ can only take the value $0$ (singlet subspace) whereas when $S_{\left(N=2\right)} =1$, $m_S$ can take any one of the three values $-1,0,1$ (triplet subspace). For our singlet subspace,
\beq
\ket{S_{\left(N=2\right)} = 0, m_S = 0} = \frac{1}{\sqrt{2}}\left(\ket{01} - \ket{10}\right)
\eeq
we see that $S_z \ket{S_{\left(N=2\right)} = 0, m_S = 0} =0$ for our system operator $S_z = \sigma_{1}^{z} \otimes I + I \otimes \sigma_{2}^{z}$. In fact, $S_{\alpha} \ket{S_{\left(N=2\right)} = 0, m_S = 0}=0$ for $\alpha \in \left\{+,-,z\right\}$ where $S_{\alpha} = \sum_{i} \sigma_{i}^{\alpha}$. Similarly, it can also be shown that $\vec{S}^{2} \ket{S_{\left(N=2\right)} = 0, m_S = 0} =0$. Since the singlet state clearly satisfies condition \eqref{eq:HSBonGood}, we conclude that $\ket{S_{\left(N=2\right)} = 0, m_S = 0}$ is by itself a one-dimensional DFS. However, we also note that the triplet states are not eigenstates of $S_z,S_+, S_-$ and thus violate Eq.~\eqref{eq:HSBonGood}.

\subsection{Three Physical Qubits}
For three physical qubits ($N = 3$), let us label the physical qubits $A$, $B$ and $C$ each with corresponding total spins $S_A = \frac{1}{2}$, $S_B = \frac{1}{2}$ and $S_C = \frac{1}{2}$. If we think of this system as a combination of a pair of spins ($A$ and $B$) with another spin $C$, we can again apply our rule for adding angular momenta which gives us from combining our pair of physical qubits into a $S_{\left(N=2\right) } = 0$ system with a spin-$\frac{1}{2}$ particle, eigenvalues of the total spin operator of
\begin{equation*}
S_{\left(N=3\right)} = \left|0-\frac{1}{2}\right|,\dots, \left|0+\frac{1}{2}\right| = \frac{1}{2}
\end{equation*}
with corresponding spin projection eigenvalues $m_S = \pm \frac{1}{2}$. If instead we chose to combine our pair of physical qubits $A$ and $B$ into a $S_{\left(N=2\right) } = 1$ system with a spin-$\frac{1}{2}$ particle, the eigenvalues of the total spin operator would be
\begin{equation*}
S_{\left(N=3\right)} = \left|1-\frac{1}{2}\right|,\dots, \left|1+\frac{1}{2}\right| = \frac{1}{2}, \frac{3}{2}
\end{equation*}
with corresponding spin projection eigenvalues $m_S = \pm \frac{1}{2}$ for $S_{\left(N=3\right)} = \frac{1}{2}$ or $m_S = \pm \frac{1}{2}, \pm \frac{3}{2}$ for $S_{\left(N=3\right)} = \frac{3}{2}$. These distinct cases arise because there are $2$ distinct ways we can get a total spin of $S = \frac{1}{2}$ from a system with $3$ physical qubits, either with two of the qubits combined as a spin-$1$ system and then combined with the spin-$\frac{1}{2}$ particle or alternatively with two qubits combined as a spin-$0$ system and subsequently combined with the remaining spin-$\frac{1}{2}$ particle.

\subsection{Generalization to $N$ physical qubits}
The extension of this idea of combining spin angular momenta is straightforward. There is an inductive method of building up from the above procedure to higher $N$. Suppose we wish
to build up the spin states of $N$ physical qubits. We would first build up the states
for a set of $N-1$ physical qubits and then couple the spin of the last qubit.

Suppose we consider the case with $N=4$ physical qubits. We can create a Bratteli diagram for
this scenario (Figure~\ref{fig:full-Bratteli}). The decoherence free states
lie on the axis where $S=0$. There are two possible paths to build up the
states from $N=0$ to $N=4$. So we can construct a qubit with each logical
state $\ket{\bar{0}} $ and $\ket{\bar{1}}$ equal to a decoherence free state indexed by the path label $\lambda$.
\begin{align}
\ket{\bar{0}}  &  =%
{\includegraphics[
natheight=1.266100in,
natwidth=4.266100in,
height=0.192in,
width=0.5803in
]%
{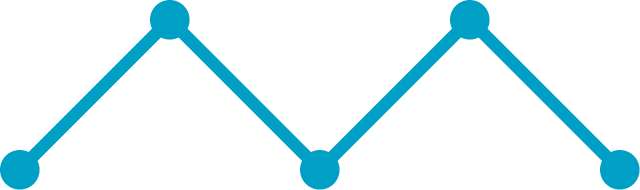}%
}%
=\ket{S=0,m_{S}=0,\lambda=0} \\
\ket{\bar{1}}  &  =%
{\includegraphics[
natheight=2.253700in,
natwidth=4.253100in,
height=0.32in,
width=0.5812in
]
{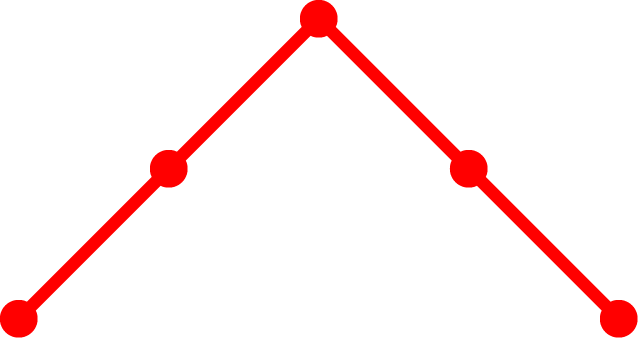}%
}
=\ket{S=0,m_{S}=0,\lambda=1}
\end{align}%
\begin{figure}
[ptb]
\begin{center}
\includegraphics[
natheight=10.820500in,
natwidth=10.420100in,
height=3.096in,
width=4.3976in
]
{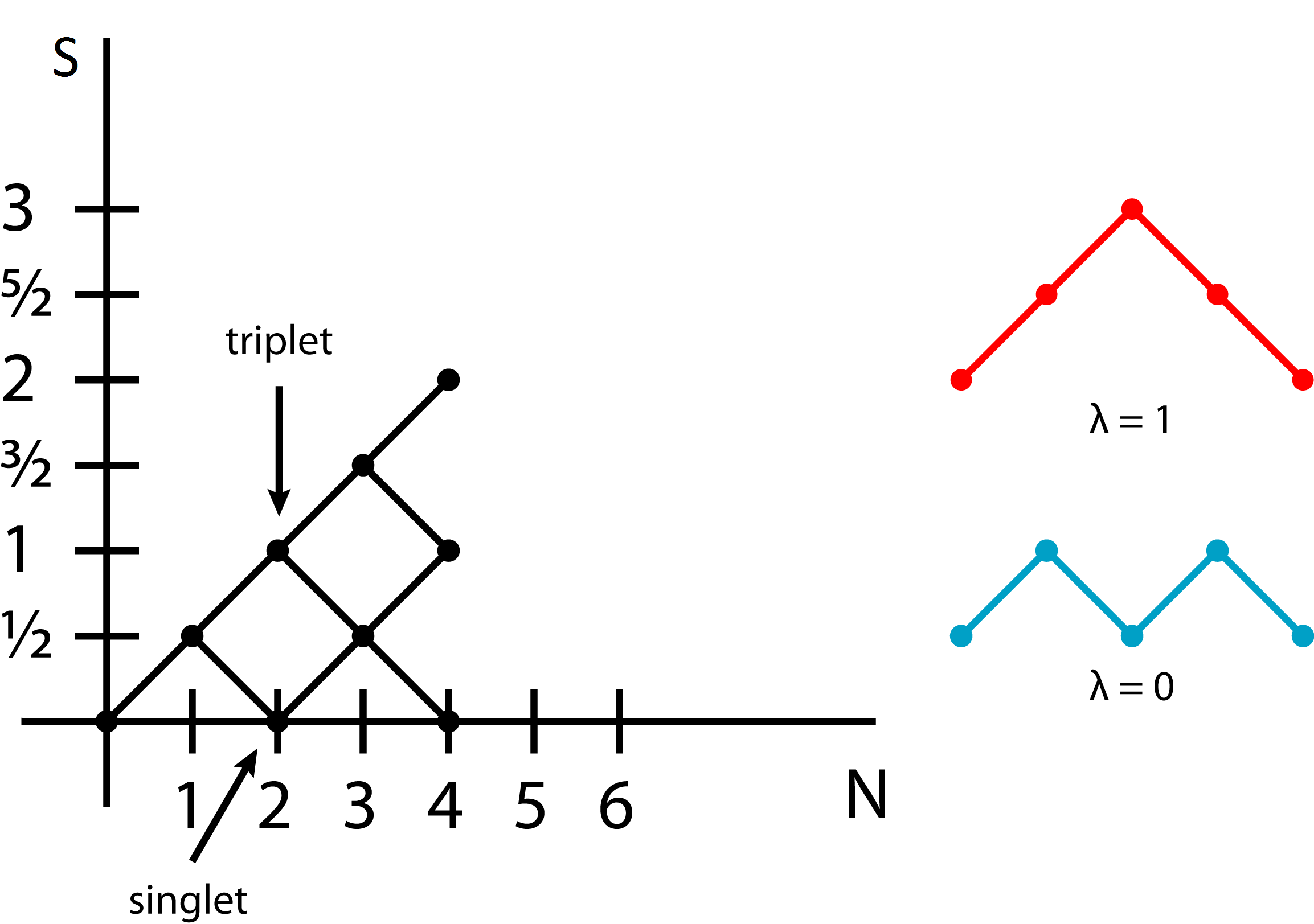}
\caption{Example Bratteli diagram for $N=4$ physical qubits. The decoherence free
states lie on the points of the axis where $S=0$. There are two ways of
getting to $S=0$ when $N=4$ because there are two possible paths starting from
$N=0$. So we can realize a qubit by setting the logical computational basis
states $\ket{\bar{0}} $ and $\ket{\bar{1}}$ to these two decoherence free states. The parameter $\lambda$ indexes the two
possible paths for $N=4$ physical qubits.}%
\label{fig:full-Bratteli}%
\end{center}
\end{figure}

Before proceeding, let us define the singlet state $\ket{ s} _{ij}$ over the
$i^{\rm th}$ and $j^{\rm th}$ qubit as
\beq
\ket{s}_{ij}\equiv \frac{1}{\sqrt{2}} \left(\ket{0
_{i} 1_{j}}- \ket{1_{i} 0_{j}}\right)
\eeq
and the three triplet states as
\bes
\begin{eqnarray}
\ket{t^{-}}_{ij} &\equiv& \ket{1_{i} 1_{j}} = \ket{S=1,m_S=-1}  \\
\ket{t^{0}}_{ij} &\equiv&\frac{1}{\sqrt{2}} \left(\ket{0_{i} 1_{j}}+ \ket{1_{i} 0_{j}}\right) = \ket{S=1,m_S=0}  \\
\ket{t^{+}}_{ij} &\equiv& \ket{0_{i} 0_{j}} = \ket{S=1,m_S=1}
\end{eqnarray}
\ees

For $N=4$ physical qubits, the logical zero is given by
\bes
\label{eq:dfs-0}
\begin{align}
\left\vert {\bar{0}}\right\rangle  &  =
{\includegraphics[
natheight=1.266100in,
natwidth=4.266100in,
height=0.192in,
width=0.5803in
]
{logical-1-path.png}%
}\\
&  =\ket{ \text{singlet}} \otimes \ket{\text{singlet}} \\
&  = \ket{s}_{12}\otimes \ket{s}_{34} \\
&=\frac{1}{2}
\left(
|0101\rangle-|0110\rangle-|1001\rangle+|1010\rangle
\right).
\end{align}
\ees
On the other hand the logical one will be later seen to be given by
\bes
\label{eq:logical1_N4}
\begin{align}
\ket{\bar{1}} &={\includegraphics[
natheight=2.253700in,
natwidth=4.253100in,
height=0.32in,
width=0.5812in
]
{logical-0-path.png}%
} \\
&=\frac{1}{\sqrt{3}} \left[\ket{t^{+}}_{12} \otimes \ket{t^{-}}_{34} + \ket{t^{-}}_{12} \otimes \ket{t^{+}}_{34} -\ket{t^{0}}_{12} \otimes \ket{t^{0}}_{34}\right] \\
&= \frac{1}{\sqrt{3}}
\left(
|1100\rangle+|0011\rangle-\frac{1}{2}|0101\rangle-\frac{1}{2}|0110\rangle-\frac{1}{2}|1001\rangle-\frac{1}{2}|1010\rangle
\right)
\end{align}
\ees

In a similar fashion, for $N=6$ physical qubits, we have for the logical zero
\beq
\ket{\bar{0}} = \ket{s}_{12} \otimes \ket{s}_{34} \otimes \ket{s}_{56}
\eeq

Note that permutations of qubit labels are permissible and can be used to define alternative basis states. Actually, we shall see in the next section that such permutations can be used to implement logical operations on the logical qubits.

\subsection{Higher Dimensions and Encoding Rate}
Clearly, more paths exist as $N$ grows. We thus have more logical states available as we increase $N$ because these states correspond to the distinct paths leading to each intersection point on the horizontal axis. There exists a combinatorial formula for the number of paths to each point in the Bratteli diagram with $S=0$ for a given $N$ and hence for the dimension $d_{N}$ of the DFS $\tilde{\mathcal{H}}(N)$ of $N$ spin-$\frac{1}{2}$ physical qubits
\beq
d_{N} \equiv\dim\left(\tilde{\mathcal{H}}\left(N\right)  \right)
=\frac{N!}{\left(N/2\right)!\left(N/2+1\right)!}
\label{eq:dim-DFS}
\eeq
As in the case of collective dephasing [Eq.~\eqref{eq:rate}] we can determine the encoding rate from the above formula. The encoding rate $r_{N}$ is the number of logical qubits $N_{L}$ we obtain divided by the the number of
physical qubits $N$ we put into the system. We can construct logical qubits from the logical
states in the DFS $\tilde{\mathcal{H}}(N)
$, and the number $N_{L}$\ of logical qubits is logarithmic in the number of
logical states of $\tilde{\mathcal{H}}(N)  $ with $N_{L}=\log
_{2}(d_{N})  $. So the encoding rate $r_{N}$ is
\beq
r_{N}\equiv \frac{\textrm{\# of DFS qubits in }\dfsh(N)}{\textrm{\# of physical qubits}}  =\frac{N_{L}}{N}=\frac{\log_{2}d_{N}}{N}.
\eeq
It can be shown using Stirling's approximation
\beq
\log_{2}N!\approx\left(N+1/2\right)
\log_{2}N-N
\eeq
for $N\gg1$, that the rate
\beq
r_{N}\approx 1-\frac{3}{2}\frac
{\log_{2}N}{N}
\eeq
for $N \gg 1 $ and hence that the rate $r_{N}$ asymptotically approaches unity
\beq
\lim_{N\rightarrow\infty}r_{N}=1.
\eeq
This implies that when $N$ is very large, remarkably we get about as many logical qubits out of our system as physical qubits we put into the system.

\subsection{Logical Operations on the DFS of Four Qubits}

How can we compute over a DFS? Suppose we group the qubits into blocks of length of $4$, and encode each block into the logical qubits given in Eqs.~\eqref{eq:dfs-0} and \eqref{eq:logical1_N4}.

Now, $\forall x,y\in \{0,1\}$, define the exchange operation $E_{ij}$ on the state $|x\rangle_i\otimes|y\rangle_j$ to be:
\beq
E_{ij}
\left(
|x\rangle_i\otimes|y\rangle_j
\right)
\equiv |y\rangle_i\otimes|x\rangle_j
\label{eq:Eij}
\eeq

Thus, $E_{ij}$ has the following matrix representation in the standard basis of two qubits:
\beq
E_{ij}=
\left(
\begin{array}{cccc}
  1 & 0 & 0 & 0 \\
  0 & 0 & 1 & 0 \\
  0 & 1 & 0 & 0 \\
  0 & 0 & 0 & 1 \\
\end{array}
\right)
\eeq
and it is easy to see that
\beq
\left[
H_{SB},E_{ij}
\right]=0
\hspace{10mm}
\forall i,j
\eeq
for the collective decoherence case. The exchange operator has a natural representation using the so-called Heisenberg exchange Hamiltonian, namely
\beq
H_{\textrm{Heis}} = \sum_{ij} J_{ij} \vec{S}_i\cdot \vec{S}_j
\eeq
where $\vec{S}_i \equiv (X_i,Y_i,Z_i)$, with $X,Y,Z$ the regular Pauli matrices (without the prefactor of $1/2$) and the $J_{ij}$ are controllable coefficients that quantify the magnitude of the coupling between spin vectors $\vec{S}_i$ and $\vec{S}_j$. Indeed, we easily find by direct matrix multiplication that
\beq
E_{ij} = \frac{1}{2}(\vec{S}_i\cdot \vec{S}_j + I)
\eeq
where $I$ is the $4\times 4$ identity matrix. Since the identity matrix will only give rise to an overall energy shift it can be dropped from $H_{\textrm{Heis}}$. The nice thing about this realization is that the Heisenberg exchange interaction is actually very prevalent as it arises directly from the Coulomb interaction between electrons, and has been tapped for quantum computation, e.g., in quantum dots \cite{Loss:1998zr,Mizel:2004ly}

We will now show that $E_{ij}$'s can be used to generate the encoded $X$, $Y$, and $Z$ operators for the encoded qubits in our $4$-qubit DFS case.

Consider the operator $\left(-E_{12}\right)$ and its action on the two encoded states given in Eqs.~\eqref{eq:dfs-0} and \eqref{eq:logical1_N4}:
\bes
\bea
\left(-E_{12}\right)|\bar{0}\rangle &=& -\left(E_{12}\right)\frac{1}{2}\left(|0101\rangle-|0110\rangle-|1001\rangle+|1010\rangle\right)\\
                                   &=& -\frac{1}{2}\left(|1001\rangle-|1010\rangle-|0101\rangle+|0110\rangle\right)\\
                                   &=& |\bar{0}\rangle\\
\left(-E_{12}\right)|\bar{1}\rangle &=& -\left(E_{12}\right)\frac{1}{\sqrt{3}}\left(|1100\rangle+|0011\rangle-\frac{1}{2}|0101\rangle-\frac{1}{2}|0110\rangle-\frac{1}{2}|1001\rangle-\frac{1}{2}|1010\rangle\right)\\
                                   &=& -\frac{1}{\sqrt{3}}\left(|1100\rangle+|0011\rangle-\frac{1}{2}|1001\rangle-\frac{1}{2}|1010\rangle-\frac{1}{2}|0101\rangle-\frac{1}{2}|0110\rangle\right)\\
                                   &=& -|\bar{1}\rangle
\eea
\ees
Therefore $\left(-E_{12}\right)$ acts as a $Z$ operator on the encoded qubits.

Similarly, we may check that $\frac{1}{\sqrt{3}}\left(E_{23}-E_{13}\right)$ acts like an $X$ operator on the encoded qubits.

Thus, we may define one set of $X$, $Y$, and $Z$ operators for the DFS in our case to be:
\bes
\bea
\bar{\sigma}^z & &\equiv\left(-E_{12}\right)\\
\bar{\sigma}^x & &\equiv\frac{1}{\sqrt{3}}\left(E_{23}-E_{13}\right)\\
\bar{\sigma}^y & &\equiv\frac{i}{2}\left[\bar{\sigma}^x,\bar{\sigma}^z\right]
\eea
\ees

As we saw in Section~\ref{sec:UEQC}, with the $\bar{\sigma}^x$ and $\bar{\sigma}^z$ operations we can construct arbitrary qubit rotations via the Euler angle formula:
\beq
\exp(i\theta\hat{n}\cdot\overrightharpoon{\bar{\sigma}})=\exp(i\alpha\bar{\sigma}^x)\exp(i\beta\bar{\sigma}^z)\exp(i\gamma\bar{\sigma}^x)
\eeq

To perform universal quantum computation we also need to construct entangling logical operations between the logical qubits. This too can be done using entirely using exchange operations. See \cite{Bacon:2000qf} for the original construction of such a gate between the logical qubits of the $4$-qubit DFS code, and 
\cite{Woodworth:2006bh} for a more recent and efficient construction.

\section{Noiseless/Decoherence Free Subsystems}
\label{noiseless}
\subsection{Representation theory of matrix algebras}
We begin this section by stating a theorem in representation theory of matrix algebras.

Recall the general form of the system-bath Hamiltonian, $H_{SB}=\sum_{\alpha} S_\alpha \otimes B_\alpha$.
Let $\mathcal{A}=\{S_\alpha\}$ be the  algebra generated by all the system operators $S_\alpha$ (all sums and products of such operators). 

\begin{thm}[\cite{Knill:2000dq}]
\label{th:rep}
Assume that $\mathcal{A}$ is $\dagger$-closed (i.e. $A\in\mathcal{A}\Rightarrow A^{\dagger}\in\mathcal{A}$) and that $I\in\mathcal{A}$. Then
\beq
\mathcal{A}\cong\bigoplus_J I_{n_J}\otimes {\cal{M}}_{d_J}(\mathbb{C}).
\label{eq:A-decomp}
\eeq
The system Hilbert space can be decomposed as
\beq
\mathcal{H}_S=\bigoplus_J \mathbb{C}^{n_J}\otimes\mathbb{C}^{d_J},
\label{eq:decomp}
\eeq
Consequently the subsystem factors $\mathbb{C}^{n_J}$'s are unaffected by decoherence.
\end{thm}
Here ${\cal{M}}_{d}(\mathbb{C})$ denotes the algebra of complex-valued $d\times d$ irreducible matrices $\{M_{d}(\mathbb{C})\}$, while as usual $I$ is the identity matrix. The number $J$ is the label of an irreducible representation (irrep) of $\mathcal{A}$, $n_J$ is the degeneracy of the $J$th irrep, and $d_J$ is the dimension of the $J$th irrep. Irreducibility means that the matrices $\{M_{d}(\mathbb{C})\}$ cannot be further block-diagonalized.

Each left factor $\mathbb{C}^{n_J}$ is called a ``subsystem" and the corresponding right factor $\mathbb{C}^{d_J}$ is called a ``gauge". Their tensor product forms a proper subspace of the system Hilbert space.

Note that the central conclusion of Theorem \ref{th:rep}, that it is possible to safely store quantum information in each of the left factors, or ``subsystems" $\mathbb{C}^{n_J}$, is a direct consequence of the fact that every term in  $\mathcal{A}$ acts trivially (as the identity operator) on these subsystem factors. 
These components $\mathbb{C}^{n_J}$ are called \emph{noiseless subsystems} (NS).

\emph{The DFS case arises when $d_J=1$}: Then $\mathbb{C}^1$ is just a scalar and $\mathbb{C}^{n_J}\otimes \mathbb{C}^1= \mathbb{C}^{n_J} $, i.e., the summand $\mathbb{C}^{n_J}\otimes \mathbb{C}^1$ reduces to a proper subspace.

Also note that it follows immediately from Eq.~\eqref{eq:decomp} that the dimension of the full system Hilbert space $\mathcal{H}_S=(\mathbb{C}^2)^{\otimes N}=\mathbb{C}^{2^N}$  can be decomposed as
\begin{align}
2^N= \sum_J n_J d_J.\label{eq:system_dim}
\end{align}

The technical conditions of the theorem are easy to satisfy. To ensure that $I \in  \mathcal{A}$ just modify the definition of $H_{SB}$ so that it includes also the pure-bath term $I\otimes H_B$. And to ensure that $\mathcal{A}$ is $\dag$-closed we can always redefine the terms in $H_{SB}$, if needed, as follows in terms of new Hermitian operators
\bes
\begin{align}
S'_{\a}\otimes B'_{\a} &\equiv \frac{1}{2}(S_{\a}\otimes B_{\a} + S^\dag_{\a}\otimes B^\dag_{\a}) \\
S''_{\a}\otimes B''_{\a} &\equiv \frac{i}{2}(S_{\a}\otimes B_{\a} - S^\dag_{\a}\otimes B^\dag_{\a}) .
\end{align}
\ees
Then $S_{\a}\otimes B_{\a} = S'_{\a}\otimes B'_{\a} -i S''_{\a}\otimes B''_{\a}$ and by writing $H_{SB}$ in terms of $S'_{\a}\otimes B'_{\a}$ and $S''_{\a}\otimes B''_{\a}$ we have ensured that $\mathcal{A}$ is $\dag$-closed.

As an application of Theorem \ref{th:rep}, we now know that in the right basis (the basis which gives the block-diagonal form \eqref{eq:A-decomp}), every  system operator  $S_\alpha$ has a matrix representation in the form $S_\alpha=\bigoplus_J I_{n_J}\otimes M^{\a}_{d_J}$:
\beq
S_\alpha=
\left[
\begin{array}{cccc}
  Q^{\a}_1 & 0 & 0 & \cdots \\
  0 & Q^{\a}_2 & 0 & \cdots \\
  0 & 0 & Q^{\a}_3 & \cdots \\
  \vdots & \vdots & \vdots & \ddots \\
\end{array}
\right]
\eeq
Each block $Q^{\a}_J=I_{n_J}\otimes M^{\a}_{d_J}$ is $n_J d_J\times n_J d_J$-dimensional, and is of the form
\beq
Q^{\a}_J=
\left[
\begin{array}{cccc}
  M^{\a}_{d_J} & 0 & 0 & \cdots \\
  0 & M^{\a}_{d_J} & 0 & \cdots \\
  0 & 0 & M^{\a}_{d_J} & \cdots \\
  \vdots & \vdots & \vdots & \ddots \\
\end{array}
\right]
\eeq

Since the system-bath interaction $H_{SB}$ is just a weighted sum of the $S_{\a}$ (the weights being the $B_{\a}$), it is also an element of $\mc{A}$, and hence has the same block-diagonal form. The same applies to any function of $H_{SB}$ that can be written in terms of sums and products, so in particular $e^{-it H_{SB}}$. Thus the system-bath unitary evolution operator also has the same block-diagonal form, and its action on the NS factors $\mathbb{C}^{n_J}$ is also trivial, i.e., proportional to the identity operator.

\subsection{Computation over a NS}
As we saw in our study of DFSs, operators that do not commute with $H_{SB}$ will induce transitions outside of the DFS. Thus we had to restrict our attention to system Hamiltonians $H_S$ which preserve the DFS [Eq.~\eqref{eq:HSonGood}]. For the same reason we now consider the \emph{commutant} $\mathcal{A}'$ of $\mathcal{A}$, defined to be the set
\beq
\label{eq:cmmt}
\mathcal{A}'=\{X:[X,A]=0,\forall A\in \mathcal{A} \}.
\eeq
This set also forms a $\dag$-closed algebra and is reducible to, over the same basis as $\mathcal{A}$,
\beq
\label{eq:cmmt2}
\mathcal{A}'\cong\bigoplus_J  \mathcal{M}_{n_J}(\mathbb{C}) \otimes I_{d_J}.
\eeq
These are the logical operations for performing quantum computation: they act non-trivially on the noiseless subsystems $\mathbb{C}^{n_J}$.

\subsection{Example: collective decoherence revisited}
\subsubsection{General structure}
Let's return to the collective decoherence model.
Recall that collective decoherence on $N$ qubits is characterized by the system operators
$
S_\a =\sum_{i=1}^N \sigma_i^\a,
$
for $\a\in\{x,y,z\}$.
In this case, the system space is
\beq
\label{eq:colld1}
    \mathcal{H}_S=\bigoplus_{J=0(1/2)}^{N/2} \mathbb{C}^{n_J}\otimes\mathbb{C}^{d_J},
\eeq
where $J$ labels the total spin, and the sum is from $J=0$ or $J=1/2$ if $N$ is even or odd, respectively.
For a fixed $J$, there are $2J+1$ different eigenvalues of $m_J$, and hence
\beq
\label{eq:colld2}
d_J=2J+1.
\eeq
By using angular momentum addition rules, one can prove that
\begin{align}
\label{eq:colld3}
n_J=\frac{(2J+1)N!}{(N/2+1+J)!(N/2-J)!},
\end{align}
which is equal to the number of paths from the origin to the vertex $(N,J)$ on the Bratteli diagram (Fig.~\ref{fig:full-Bratteli}), and generalizes the DFS dimensionality formula, Eq.~\eqref{eq:dim-DFS}.

We have
\beq
    \mathcal{H}_S^{(N)}=\mathbb{C}^{n_0}\otimes \mathbb{C}^1 \oplus \mathbb{C}^{n_1}\otimes \mathbb{C}^3 \oplus \cdots
\eeq
for $N$ even,
and
\beq
    \mathcal{H}_S^{(N)}=\mathbb{C}^{n_{1/2}}\otimes \mathbb{C}^2 \oplus \mathbb{C}^{n_{3/2}}\otimes \mathbb{C}^4 \oplus \cdots,
\eeq
for $N$ odd. For example, when $N=3$: $
n_{1/2}=\frac{2\cdot 3!}{3!1!}=2$.

\emph{The DFS case arises when $J=0$ (so that $d_J=1$)}: Then $\mathbb{C}^1$ is just a scalar and $\mathbb{C}^{n_0}\otimes \mathbb{C}^1= \mathbb{C}^{n_0} $, i.e., the left (subsystem) factor has a dimension equal to the number of paths and the right (gauge) factor is just a scalar. In this case the summand $\mathbb{C}^{n_0}\otimes \mathbb{C}^1$ reduces to a proper subspace.

The noiseless subsystems corresponding to different values of $J$ for a given $N$ can be computed
by using the addition of angular momentum, as illustrated below.

\subsubsection{The three qubit code for collective decoherence}

The smallest $N$ which encodes one qubit in a noiseless subsystem is $N=3$. In this case,
\beq
    \mathcal{H}_S^{(N=3)}=\mathbb{C}^{2}\otimes \mathbb{C}^2 \oplus \mathbb{C}^{1}\otimes \mathbb{C}^4,
    \label{eq:HN=3}
\eeq
Thus we can encode one qubit in the first factor $\mathbb{C}^2$ of $J=1/2$.
The two paths of $\ket{\bar{0}}$ and $\ket{\bar{1}}$
are respectively {\includegraphics[
natheight=1.266100in,
natwidth=4.266100in,
height=0.192in,
width=0.5803in
]
{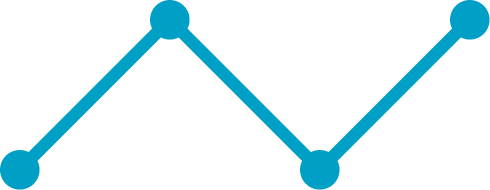}%
} ($\lambda=0$) and
{\includegraphics[
natheight=1.266100in,
natwidth=4.266100in,
height=0.192in,
width=0.5803in
]
{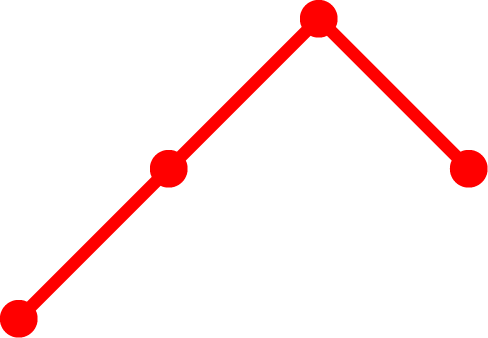}%
} ($\lambda=1$). The end points of these two paths each have two spin projections $m_J=\pm 1/2$ (since they correspond to a total spin $J=1/2$). Using the state notation $\ket{J, \lambda, m_J}$, we thus have
\bes
\label{eq:3code}
\bea
\ket{\bar{0}}=&\a \ket{1/2,0,-1/2}+\b\ket{1/2,0,1/2}= \ket{1/2,0}\otimes ( \a\ket{-1/2}+\b\ket{1/2}),\\
\ket{\bar{1}}=&\a \ket{1/2,1,-1/2}+\b\ket{1/2,1,1/2}= \ket{1/2,1}\otimes ( \a\ket{-1/2}+\b\ket{1/2}),
\eea
\ees
where $\a$ and $\b$ are completely arbitrary. Or using the vector form, we have
\beq
\ket{\bar{0}}=\left(\begin{array}{c}1 \\ 0\end{array} \right)\otimes \left(\begin{array}{c}\a \\ \b\end{array} \right),\quad
\ket{\bar{1}}=\left(\begin{array}{c}0 \\ 1\end{array} \right)\otimes \left(\begin{array}{c}\a \\ \b\end{array} \right).
\eeq
Suppose we want to encode a state $\ket{\psi}=a\ket{0}+b\ket{1}$. The encoded state is
\beq
\ket{\bar{\psi}}=a\ket{\bar{0}}+\b\ket{\bar{1}}\\
=\left(\begin{array}{c}a \\ b\end{array} \right)\otimes \left(\begin{array}{c}\a \\ \b\end{array} \right),
\eeq
where we only care about the encoded information $a$ and $b$. Notice how this last result precisely corresponds to the $\mathbb{C}^{2}\otimes \mathbb{C}^2$ term in Eq.~\eqref{eq:HN=3}. Thus, $\a$ and $\b$ are a ``gauge amplitudes''; their precise values don't matter.

The interaction Hamiltonian restricted to the system $S$  is of the form
\bes
\begin{align}
\left.H_{SB}\right|_S=& \bigoplus_{J=1/2}^{3/2} I_{n_J}\otimes  \mathcal{M}_{d_J}\\
=&I_2\otimes  \mathcal{M}_2\oplus I_1\otimes  \mathcal{M}_4\\
=&\left[\begin{array}{c|c}
I_2\otimes  \mathcal{M}_2&\\
\hline
&  \mathcal{M}_4\\
\end{array}\right].
\end{align}
\ees
What this means is that the term $I_2\otimes  \mathcal{M}_2$ acts on $\ket{\bar{\psi}}$ and leaves its first factor alone (this is good since that's where we store the qubit), but applies some arbitrary matrix $M_2$ to the second factor (we don't care). $ \mathcal{M}_4$ acts on the $\mathbb{C}^{1}\otimes \mathbb{C}^4$ subspace, where we don't store any quantum information.

We can check that the dimensions satisfy Eq.~\eqref{eq:system_dim}:
\begin{align}
\sum_{J=1/2}^{3/2}n_Jd_J= n_{1/2}d_{1/2}+n_{3/2}d_{3/2}=2\cdot 2+1\cdot 4=8=2^3.
\end{align}

Let's now find explicit expressions for the basis state of the three-qubit noiseless subsystem. Recall that $| 0 \rangle = | J = 1/2, m_J = 1/2 \rangle$, $|1 \rangle = | 1/2, -1/2 \rangle$,
the singlet state $\ket{s}=\ket{0,0}=\frac{1}{\sqrt{2}}(\ket{01}-\ket{10})$ and the triplet states are $| t_+ \rangle = | 1, 1 \rangle = | 0 0 \rangle$, $| t_- \rangle = | 1, 1 \rangle = | 1 1 \rangle$, and $| t_0 \rangle = | 1, 0 \rangle = \frac{1}{2}(| 01 \rangle + | 10 \rangle)$.
We now derive the four $J=1/2$ states by using the addition of angular momentum and Clebsch-Gordan coefficients.
\bes
\begin{align}
  \ket{1/2,0,-1/2}=&\ket{s} \otimes \ket{m_3=-1/2}\\
  =&\frac{1}{\sqrt{2}}(\ket{011}-\ket{101})\\
  \ket{1/2,0,1/2}=&\ket{s}\otimes \ket{0}\\
  =&\frac{1}{\sqrt{2}}(\ket{010}-\ket{100})\\
  \ket{1/2,1,-1/2}=&\frac{1}{\sqrt{3}}(\sqrt{2}\ket{J_{12}=1,m_{J_{12}}=-1}\otimes \ket{m_3=1/2}-\ket{J_{12}=1,m_{J_{12}}=0}\otimes \ket{m_3=-1/2})  \\
  =& \frac{1}{\sqrt{6}}(2\ket{110}-\ket{011}-\ket{101})\\
  \ket{1/2,1,1/2}=&\frac{1}{\sqrt{3}}(\ket{J_{12}=1,m_{J_{12}}=0}\otimes \ket{m_3=1/2}-\sqrt{2}\ket{J_{12}=1,m_{J_{12}}=1}\otimes \ket{m_3=-1/2})  \\
  =& \frac{1}{\sqrt{6}}(\ket{010}+\ket{100}-2\ket{001}).
  \end{align}
\ees
These are the basis states that appear in Eq.~\eqref{eq:3code}, so they complete the specification of the three-qubit code.

\subsubsection{Computation over the three-qubit code}
Consider the permutation operator $E_{ij}=\frac{1}{2}(I+\vec{\sigma_i}\cdot \vec{\sigma_j})$
such that $E_{ij}\ket{x}_i\ket{y}_j=\ket{y}_i\ket{x}_j$ for $x,y\in\{0,1\}$.
We have
\bes
\label{eq:3basis}
\begin{align}
E_{12}\ket{1/2,0,-1/2}=&\frac{1}{\sqrt{2}}(-\ket{011}+\ket{101})=-\ket{1/2,0,-1/2}\\
E_{12}\ket{1/2,0,1/2}=&\frac{1}{\sqrt{2}}(\ket{100}-\ket{010})=-\ket{1/2,0,1/2}\\
E_{12}\ket{1/2,1,-1/2}=&\frac{1}{\sqrt{6}}(2\ket{110}-\ket{011}-\ket{101})=\ket{1/2,1,-1/2}\\
E_{12}\ket{1/2,1,1/2}=&\frac{1}{\sqrt{6}}(\ket{010}+\ket{100}-2\ket{001})=\ket{1/2,1,1/2}.
\end{align}
\ees
Thus $E_{12}$ works as a logical $-{\sigma}^z$, in the sense that
\begin{align}
E_{12}=\begin{bmatrix}-1&&&\\&-1&&\\&&1&\\&&&1\end{bmatrix}=-\sigma^z\otimes I = -\bar{\sigma}^z
\end{align}
in the ordered basis of the four $J=1/2$ states given in Eq.~\eqref{eq:3basis}. Again, this agrees with the $\mathbb{C}^{2}\otimes \mathbb{C}^2$ structure of the Hilbert subspace where we store our qubit.

Similarly, one can easily verify that
\begin{align}
\frac{1}{\sqrt{3}}(E_{13}-E_{23})=\sigma^x\otimes I=\bar{\sigma}^x.
\end{align}
Then $\bar{\sigma}^y$ can be obtained from
\begin{align}
2i\bar{\sigma}^y=[\bar{\sigma}^z, \bar{\sigma}^x].
\end{align}

Finding the explicit form of the encoded CNOT is a complicated problem. See \cite{Kempe:2001uq} for a constructive approach using infinitesimal exchange generators, and \cite{DiVincenzo:2000kx} for a numerical approach that yields a finite and small set of exchange-based gates.

\section{Dynamical Decoupling}
\label{sec:DD}

As we saw in the discussion of noiseless subsystems, the error algebra $\mathcal{A}=\{S_\alpha\}$ is isomorphic to
a direct sum of $n_J$ copies of $d_J\times d_J$ complex matrix algebras: 
$\mathcal{A}\cong\bigoplus_J I_{n_J}\otimes \mathcal{M}_{d_J}(\mathbb{C})$,
where $n_J$ is the degeneracy of the $J$th irrep and $d_J$ is the dimension of the $J$th irrep. We can store quantum information in a factor $\mathbb{C}^{n_J}$ when $n_J>1$. However, from general principles (Noether's theorem) we know that degeneracy requires a symmetry, and in our case we would only have $n_J>1$ when the system-bath coupling has some symmetry. When there's no symmetry at all, $n_J = 1$ for all $J$'s, and a DFS or NS may not exist. Starting in this section, we discuss how to ``engineer" the system-bath coupling to have some symmetry.

To sum up, the idea of a DFS/NS is powerful: we can use naturally available symmetries to encode and hide quantum information, and we can compute over the encoded, hidden information. But often such symmetries are imperfect, and we need additional tools to protect quantum information. Such an approach, which adds active intervention to the passive DFS/NS approach, is dynamical decoupling.

\subsection{Decoupling single qubit pure dephasing}

\subsubsection{The ideal pulse case}
Consider a single qubit system with the pure dephasing system-bath coupling Hamiltonian
\beq
H_{SB}=\sigma^z\otimes B^z
\eeq
and system Hamiltonian
\beq
H_S=\lambda(t)\sigma^x.
\eeq
We assume that $\lambda(t)$ is a fully controllable field, e.g., several pulses of a magnetic or electric field applied to the system. Assume these pulses last for a period of time $\delta$, and with strength $\lambda$, and
\beq
\delta\lambda=\frac{\pi}{2}.
\eeq
Assume that at $t=0$, we turn on the pulse for a period of time $\delta$, then let the system and bath interact for a period of time $\tau$, and repeat this procedure, as shown in Fig. ~\ref{Fig:pulse}. In the \emph{ideal} case, $\delta\rightarrow0$ and $\lambda\rightarrow\infty$ while still satisfying $\delta\lambda=\frac{\pi}{2}$, which means the pulses are a series of delta functions. For simplicity we temporarily assume that $H_{B}=0$.
 \begin{figure}[!ht]
		\includegraphics[width=80mm]{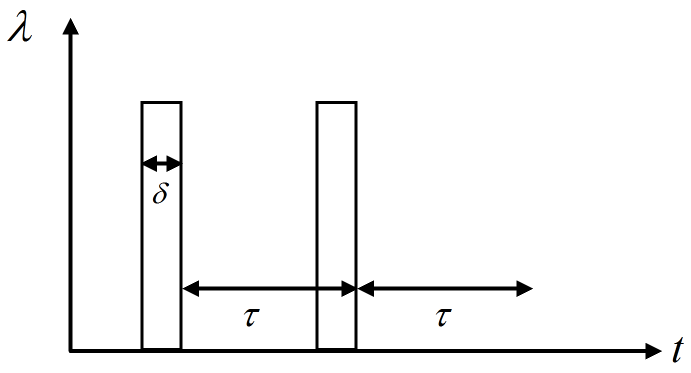}
		\caption{Schematic of a dynamical decoupling pulse sequence. Pulses have width $\d$ and intervals of duration $\tau$. The modulation function $\lambda (t)$ is responsible for switching the pulses on and off.}
		\label{Fig:pulse}
\end{figure}

To formalize this ``ideal pulse'' scenario, let us define the system-bath ``pulse-free" evolution operator $f_\tau$, and the unitary transformation caused by the pulse, $X$, as follows:
\bes
\begin{align}
f_\tau &\equiv e^{-i\tau H_{SB}}.    \label{eq:ftao}\\
     X &\equiv  e^{-i\delta\lambda\sigma^x} \otimes I_B = e^{-i\frac{\pi}{2}\sigma^x} \otimes I_B
        = -i\s^x\otimes I_B.
       \label{eq:Xpulse}
\end{align}
\ees
 In the case of an ideal pulse ($\delta\rightarrow0, \lambda\rightarrow\infty$) there is no system-bath interaction during the time the pulse is turned on, since the duration of the pulse is $0$. Then the joint system-bath evolution operator at time $t=2\tau$ is (dropping overall factors of $i$ and minus signs)
\begin{eqnarray}
Xf_\tau Xf_\tau  =\sigma^x e^{-i\tau H_{SB}}\sigma^x e^{-i\tau H_{SB}} =e^{-i\tau\sigma^x H_{SB} \sigma^x}e^{-i\tau H_{SB}}.
                 \label{eq:DDevolution}
\end{eqnarray}
where in the second equality we used the identity
\beq
Ue^A U^\dagger = e^{UAU^\dagger},
\label{eq:UAU}
\eeq
valid for any operator $A$ and unitary $U$.

On the other hand, since the Pauli matrices are Hermitian and every pair of distinct Pauli matrices anticommutes,
\beq
\{\s_\a,\s_\b\}=0,\quad \a\neq\b,
\label{eq:ac}
\eeq
where the anti-commutator is defined as
\beq
\{A,B\}\equiv AB+BA
\eeq
for any pair of operators $A$ and $B$, it follows that the sign of $H_{SB}$ is flipped:
\bes
\begin{eqnarray}
\sigma^xH_{SB}\sigma^x&=&\sigma^x\sigma^z\sigma^x\otimes B^z \\
&=&-\sigma^z\otimes B^z,\\
&=& - H_{SB}.
                      \label{eq:HSB_Trans}
\end{eqnarray}
\ees
This means that the evolution under $H_{SB}$ has been effectively time-reversed!

Indeed, if we now substitute Eq.~\eqref{eq:HSB_Trans} into Eq.~\eqref{eq:DDevolution} we obtain
\beq
Xf_\tau Xf_\tau = e^{+i\tau H_{SB}}e^{-i\tau H_{SB}} = I.
\eeq
Thus, the bath has no effect on the system at the instant $t=2\tau$. In other words, for a fleeting instant, at $t=2\tau$, the the system is completely decoupled from the bath. Clearly, if we were to repeat Eq.~\eqref{eq:DDevolution} over and over, the system would ``stroboscopically'' decouple from the bath every $2\tau$.

\subsubsection{The real pulse case}

Unfortunately, in the real world, pulses cannot be described by $\delta$ functions, because that would require infinite energy. Generally, the pulse must be described by some continuous function $\lambda(t)$ in the time domain, which may or may not be a pulse. Then, during the period when the pulse is applied to the system, the system-bath Hamiltonian cannot be neglected, so we must take it into account. Keeping the assumption $H_{B}=0$ for the time being, we have to modify the pulse to
\beq
X=e^{-i\delta(\lambda\sigma^x+H_{SB})}.
\eeq
If $\lambda\gg\|H_{SB}\|$ and $\delta\lambda=\pi/2$ (we'll define the norm momentarily), then it's true that $X\approx\sigma^x\otimes I_B$, i.e., we can approximate the ideal pulse case of Eq.~\eqref{eq:Xpulse}. Let us now see how good of an approximation this is.

To deal with the real pulse case, we first recall the Baker-Campell-Hausdorff (BCH) formula (see any advanced book on matrices, e.g., \cite{Bhatia:book}):
\beq
e^{\epsilon(A+B)}=e^{\epsilon A}e^{\epsilon B}e^{(\epsilon^2/2)[A,B]+\mathcal{O}(\epsilon^3)},
\label{eq:BCH}
\eeq
for any pair of operators $A$ and $B$. Now, set $\epsilon=-i\delta$, $A=\lambda\sigma^x$, $B=H_{SB}=\sigma^z\otimes B^z$. Then the real pulse is
\beq
X=\underbrace{e^{-i\delta\lambda\sigma^x}}_{\textrm{ideal pulse}}\underbrace{e^{-i\delta H_{SB}}}_{\textrm{OK}}\underbrace{e^{-\delta^2\lambda[\sigma^x, H_{SB}]/2+\mathcal {O}(\delta^3)}}_{\textrm{does damage}}
\label{eq:realX}
\eeq
The first exponential is just the ideal pulse, and the second is OK as well (we will see that shortly), but the third term will cause the pulse sequence to operate imperfectly.
Let's analyze the pulse sequence subject to this structure of the real pulse.

First, let us define the operator norm \cite{Bhatia:book}:
\beq
\|A\| \equiv \sup_{\ket{\psi}} \frac{\norm{A\ket{\psi}}}{\norm{\ket{\psi}}} = \sup_{\ket{\psi}} \frac{\sqrt{\bra{\psi}A^\dagger A\ket{\psi}}}{\norm{\ket{\psi}}},
\label{eq:op-norm}
\eeq
i.e., the largest singular value of $A$ (the largest eigenvalue of $|A| = \sqrt{A^\dag A}$), which reduces to the absolute value of the largest eigenvalue of $A$ when $A$ is Hermitian. The operator norm is an example of a unitarily invariant (ui) norm: If $U$ and $V$ are unitary, and $A$ is some operator, a norm is said to be unitarily invariant if
\beq
\norm{U A V}_{\textrm{ui}}=\norm{A}_{\textrm{ui}}.
\label{eq:ui-norm}
\eeq
Such norms are submultiplicative over products and distributive over tensor products \cite{Bhatia:book}:
\beq
\|AB\|_{\textrm{ui}} \leq \|A\|_{\textrm{ui}}\|B\|_{\textrm{ui}}, \quad \|A\otimes B\|_{\textrm{ui}} = \|A\|_{\textrm{ui}}\|B\|_{\textrm{ui}}.
\label{eq:ui-norm-sub}
\eeq
Then, using $\|\sigma^{\a}\|=1$ (the eigenvalues of $\s^{\a}$ are $\pm 1$), we have
\beq
\|H_{SB}\|=\|\sigma^x\otimes B^z\|=\|B^z\|.
\eeq
Using this we find
\bes
\bea
\|\delta^2\lambda[\sigma^x, H_{SB}]/2\| &\leq& \delta^2\lambda(\|\sigma^x H_{SB}\|+\|H_{SB}\sigma^x\|)/2 \\
&\leq& (\pi/2)\delta\|\sigma^x\|\|H_{SB}\| \\
&=& \mathcal{O}(\delta\|B^z\|),
\eea
\ees
where we used the triangle inequality, submultiplicativity, and $\delta\lambda=\pi/2$. So, we arrive at the important conclusion that the pulse width should be small compared to the inverse of the system-bath coupling strength, i.e.,
\beq
\delta \ll 1/\|B^z\|,
\eeq
should be satisfied assuming $\|B^z\|$ is finite. This assumption won't always be satisfied (e.g., it does not hold for the spin-boson model), in which case different analysis techniques are required. In particular, operators norms will have to be replaced by correlation functions, which remain finite even when operator norms are formally infinite (see, e.g., \cite{NLP:09}). But, for now we shall simply assume that all operators norms we shall encounter are indeed finite.

Let's Taylor expand the ``damage" term to lowest order:
\bes
\bea
e^{-\delta^2\lambda[\sigma^x, H_{SB}]/2} &=& I -\delta^2\lambda[\sigma^x, H_{SB}]/2 + \mc{O}(\delta^3)\\
&=& I + \mc{O}(\delta\|B^z\|).
\eea
\ees
Putting everything together, including $e^{-i\delta\lambda\sigma^x} = -i\s^x$, the evolution subject to the real pulse is, from Eq.~\eqref{eq:realX} (again dropping overall phase factors)
\bes
\label{eq:fXfXpulse}
\bea
X f_\tau X f_\tau &=& [\sigma^xe^{-i\delta H_{SB}}(I+\mathcal {O}(\delta\|B^z\|))]e^{-i\tau H_{SB}}[\sigma^xe^{-i\delta H_{SB}}(I+\mathcal {O}(\delta\|B^z\|))] e^{-i\tau H_{SB}}\\
%                  &\approx& [e^{-i\delta H_{SB}}\sigma^x(I+\mathcal {O}(\delta\|B^z\|))]e^{-i\tau H_{SB}}[e^{-i\delta H_{SB}}\sigma^x(I+\mathcal {O}(\delta\|B^z\|))] e^{-i\tau H_{SB}}\\
 %                 &=&\sigma^x(I+\mathcal{O}(\delta\|B^z\|))e^{-i(\tau+\delta)H_{SB}}\sigma^{x}(I+\mathcal{O}(\delta\|B^z\|))e^{-i(\tau+\delta)H_{SB}}\\
                  &=&e^{-i(\tau+\delta)\sigma^x H_{SB}\sigma^x}e^{-i(\tau+\delta)H_{SB}}+\mathcal{O}(\delta\|B^z\|)\\
                  &=&e^{i(\tau+\delta) H_{SB}}e^{-i(\tau+\delta)H_{SB}}+\mathcal{O}(\delta\|B^z\|)\\
                  &=& I+\mc{O}(\delta\|B^z\|),
\eea
\ees
so we see that the real pulse sequence has a first order pulse width correction.

Now let us recall that in fact $H_B\neq0$. How does this impact the analysis? Both the free evolution and the pulse actually include $H_B$:
\bes
\begin{eqnarray}
f_\tau&=&e^{-i\tau(H_{SB}+H_{B})},\label{eq:fHB}\\
X&=&e^{-i\delta(\lambda\sigma^x+H_{SB}+H_B)},
\label{eq:XHB}
\end{eqnarray}
\ees
so we need $\lambda \gg \|H_{SB}+H_B\|$. Set $H_{SB}^\prime = H_{B}+H_{SB}$, and note that the ideal pulse commutes with $H_B$, so that
\beq
\sigma^x(H_{SB}+H_B)\sigma^x=-H_{SB}+H_B.
\eeq
 Substituting Eqs.~\eqref{eq:fHB} and \eqref{eq:XHB} into Eq.~\eqref{eq:fXfXpulse} we then have:
\bes
\label{eq:FXFXHB}
\bea
Xf_\tau X f_\tau &=& e^{-i(\tau+\delta)\sigma^x H_{SB}^\prime \sigma^x}e^{-i(\tau+\delta)H_{SB}^\prime}
+\mathcal{O}(\delta\|H_{SB}^\prime\|)\\
                  &=&e^{-i(\tau+\delta)(-H_{SB}+H_B)}e^{-i(\tau+\delta)(H_{SB}+H_{B})}+\mathcal{O}(\delta\|H_{SB}^\prime\|)
\eea
\ees
Setting $A=H_{SB}+H_B$, $B= -H_{SB}+H_B$, and using the BCH formula \eqref{eq:BCH} again in the form $e^{\epsilon A}e^{\epsilon B}=e^{\epsilon(A+B)}e^{-(\epsilon^2/2)[A,B]+\mathcal{O}(\epsilon^3)}$, we have $A+B = 2H_B$ and $\|[A,B]\|/2 \leq \|H_B-H_{SB}\|\|H_{B}+H_{SB}\| \leq (\|H_{SB}\|+\|H_B\|)^2$, so that Eq.~\eqref{eq:FXFXHB} reduces to
\beq
Xf_\tau X f_\tau =I_S\otimes e^{-2i(\tau+\delta) H_B}+\mathcal{O}[(\tau+\delta)^2(\|H_{SB}\|+\|H_B\|)^2]+\mathcal{O}[\delta (\|H_{SB}\|+\|H_B\|)].
\eeq
Assuming that the pulses are very narrow, i.e., $\d \ll \tau$ (recall that we anyhow need this for ideal pulses), we can neglect $\d$ relative to $\tau$ in the second term, and so the smallness conditions are
\beq
\d \ll \tau \ll 1/(\|B^z\|+\|H_B\|),
\eeq
which replaces the earlier $\delta \ll 1/\|B^z\|$ condition we derived when we ignored $H_B$.

\subsection{Decoupling single qubit general decoherence}

Let us now consider the most general 1-qubit system-bath coupling Hamiltonian
\beq
H_{SB}=\sum_{\alpha=x,y,z}\sigma^\alpha\otimes B^\alpha .
\eeq
Using the anticommutation condition Eq.~\eqref{eq:ac} we have
\beq
\sigma^xH_{SB}\sigma^x = \sigma^x\otimes B^x-\sigma^y\otimes B^y-\sigma^z\otimes B^z,
\eeq
so that the $Xf_\tau X f_\tau$ pulse sequence should cancel both the $y$ and $z$ contributions. The remaining problem is how to deal with the $\sigma^x$ term in $H_{SB}$.

Let us assume that the pulses are ideal ($\d=0$). We can remove the remaining $\sigma^x$ term by inserting the sequence for pure dephasing into a second pulse sequence, designed to remove the $\s_x$ term. This kind of recursive construction is very powerful, and we will see it again in Section~\ref{CDD}.

Let the free evolution again be
\beq
f_\tau = e^{-i\tau H_{SB}}.
\eeq
Then, after applying an $X$-type sequence,
\beq
X f_{2\tau}^\prime\equiv f_\tau X f_\tau = e^{-i2\tau (\sigma^x\otimes B^x + H_B)} + \mc{O}(\tau^2).
\eeq
To remove the remaining $\sigma^x\otimes B^x$ we can apply a $Y$-type sequence to $f_{2\tau}^\prime$:
\bes
\label{eq:XY4}
\bea
f_{4\tau}^{\prime\prime} &=&Y f_{2\tau}^\prime Yf_{2\tau}^\prime \\
&=& Y Xf_\tau Xf_\tau  Y X f_\tau X f_\tau \\
                 &=&Z f_\tau Xf_\tau  Z f_\tau X f_\tau .
\eea
\ees
where as usual we dropped overall phase factors. Clearly,
\beq
f_{4\tau}^{\prime\prime} = e^{-i4\tau H_B} + \mc{O}(\tau^2),
\label{eq:f4tau}
\eeq
so that at $t=4\tau$ the system is completely decoupled from the bath. This pulse sequence is shown in Fig.~\ref{Fig:pulsesequence}, and is the \emph{universal decoupling sequence} (for a single qubit), since it removes a general system-bath interaction.

\begin{figure}[b]
		\includegraphics[width=140mm]{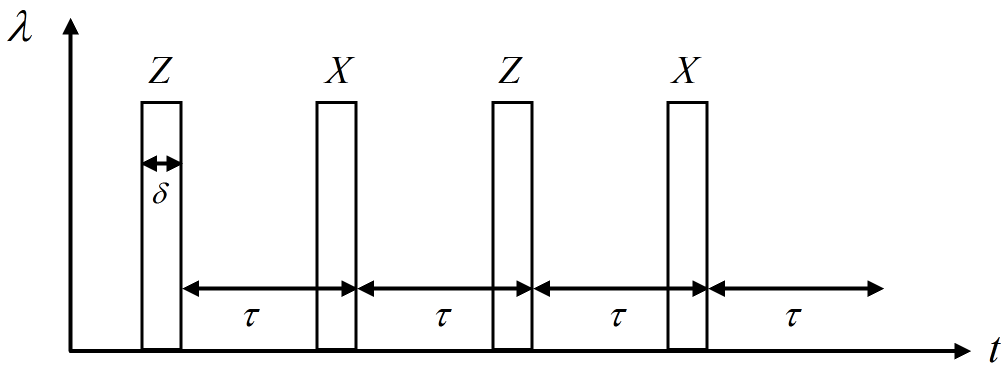}
		\caption{Schematic of the pulse sequence used to suppress general single-qubit decoherence. This pulse sequence is sometimes called XY-4, or the universal decoupling sequence.}
\label{Fig:pulsesequence}
\end{figure}

\section{Dynamical decoupling as symmetrization}
\label{sec:DD-symm}

We saw in Eq.~\eqref{eq:XY4} that the universal decoupling sequence $Z f_\tau Xf_\tau  Z f_\tau X f_\tau $ decouples a single qubit from an arbitrary bath (to first order). We constructed this sequence using a recursive scheme. In this section we would like to adopt a different perspective, which will help us generalize the theory beyond the single qubit case. This perspective is based on symmetrization \cite{Zanardi:1999fk}.

Up to a global phase we have
\beq
\label{eq:sym_sch}
\opZ f_\tau \opX f_\tau \opZ f_\tau \opX  f_\tau = \lp \opZ f_\tau \opZ \rp \lp \opY f_\tau \opY \rp \lp \opX f_\tau \opX \rp \lp \opI f_\tau \opI \rp .
\eeq
On the right hand side of \eqref{eq:sym_sch} we see a clear structure: we are ``cycling'' over the group formed by the elements $\{I,X,Y,Z\}$. Note that because we are not concerned with global phases, this is not the Pauli group, which is the $16$-element group $\{\pm I,\pm X,\pm Y,\pm Z,\pm iI,\pm iX,\pm iY,\pm iZ\}$). Rather, the four element group is the abelian Klein group, whose multiplication table is given by
\begin{center}
\begin{tabular}{|c|c|c|c|c|}
  \hline
  $\times$ & \opI & \opX & \opY & \opZ \\
  \hline
  \opI & \opI & \opX & \opY & \opZ \\
  \opX & \opX & \opI & \opZ & \opY \\
  \opY & \opY & \opZ & \opI & \opX\\
  \opZ & \opZ & \opY & \opX & \opI \\
  \hline
\end{tabular}
\end{center}

Returning to the decoupling discussion, to see why the sequence in Eq.~\eqref{eq:sym_sch} works, note that if we let
$\opA^\alpha=\sigma^\alpha \otimes \opB^\alpha$
we have
\bea
\opI f_\tau \opI  &=& f_\tau = \ee^{-\ii\tau\lp \opA^x + \opA^y + \opA^z  + \opH_{B} \rp}\\
\opX f_\tau \opX  &=& \ee^{-\ii\tau\s^x\opH\s^x} = \ee^{-\ii\tau\lp \opA^x - \opA^y - \opA^z  + \opH_{B}\rp}\\
\opY f_\tau \opY  &=& \ee^{-\ii\tau\s^y\opH\s^y} = \ee^{-\ii\tau\lp -\opA^x + \opA^y - \opA^z  + \opH_{B}\rp}\\
\opZ f_\tau \opZ  &=& \ee^{-\ii\tau\s^z\opH\s^z} = \ee^{-\ii\tau\lp -\opA^x - \opA^y + \opA^z  + \opH_{B}\rp}
\eea
Using the BCH expansion \eqref{eq:BCH} again,
we see that when we add all four of the exponents they cancel all $A^\a$ terms perfectly, so that the right hand side of \eqref{eq:sym_sch} is just
\beq
\lp \opZ f_\tau \opZ \rp \lp \opY f_\tau \opY \rp \lp \opX f_\tau \opX \rp \lp \opI f_\tau \opI \rp
= \ee^{-4\ii\tau H_B} + \mc{O}(\tau^2),
\label{eq:ZXZX}
\eeq
just like in Eq.~\eqref{eq:f4tau}. This is the first order decoupling we were looking for.

From the right hand side of Eq.~\eqref{eq:sym_sch}, we also gain some intuition as to what our
strategy should be beyond the single qubit case.
Again, define
\beq
f_\tau = \exp[-\ii\tau\lp\opH_{SB} + \opH_B\rp],
\eeq
where now $H_{SB}$ and $H_B$ are completely general system-bath and pure-bath operators.
Generalizing from Eq.~\eqref{eq:sym_sch},
consider a group
\beq
\mathcal{G} = \{g_0, \cdots, g_K\}
\eeq
(with $g_0 \equiv \opI$) of unitary transformations $g_j$ acting purely on
the system. Assuming that each such pulse $g_j$ is effectively instantaneous, the pulse sequence shall consist of a full cycle over the group, lasting total time
\beq
T = \lp K+1 \rp\tau.
\eeq
More specifically, we apply the following \emph{symmetrization sequence}:
\bes
\begin{align}
\label{eq:cycle}
\opU (T) &= \prod_{j=0}^{K} g_j^\dagger f_\tau g_j\\
         &= \prod_{j=0}^K \ee^{-\ii\tau\lp g_j^\dagger \opH_{SB} g_j  + \opH_B \rp}\\
         &= e^{-i\tau\lp\sum_{j=0}^K g_j^\dagger \opH_{SB} g_j +(K+1)H_B\rp}+ \mc{O}\lp T^2 \rp\\
         &= \ee^{-\ii T \lp H_{SB}' + H_B \rp} + \mc{O}\lp T^2 \rp, 
         \label{fullunitary}
\end{align}
\ees
where we used Eq.~\eqref{eq:UAU} in the second equality, the BCH formula in the third, and defined the effective, or \emph{average Hamiltonian}
\beq
\opH_{SB}' = \frac{1}{K+1} \sum_{j=0}^{K} g_j^\dagger \opH_{SB} g_j. 
\label{Hsb}
\eeq
Thus the effect of the pulse sequence defined by $\mc{G}$ is to transform the original $H_{SB}$ into the \emph{group-averaged} $H_{SB}'$. If we can choose the \emph{decoupling group} $\mc{G}$ so that $H_{SB}'$ is harmless, we will have achieved our decoupling goal.

Thus, our strategy for general first order decoupling could be one of following:
\begin{enumerate}
  \item Pick a group $\mathcal{G}$ such that $\opH_{SB}'=0$.
  \item Pick a group $\mathcal{G}$ such that
      $\opH_{SB}'=\opI_S\otimes \opB'$.
\end{enumerate}

The first of these is precisely what we saw for decoupling a single qubit using the Pauli (or Klein) group, i.e., Eq.~\eqref{eq:ZXZX}. To see when we can achieve the second strategy (which obviously included the first as a special case with $B'=0$), note that $H_{SB}'$ belongs to the centralizer of the group $\mc{G}$, i.e.
\beq
      H_{SB} \overset{\mc{G}}\longmapsto \opH_{SB}'\in Z(\mc{G}) \equiv \left\{ A \left| [A,g]=0 \; \forall g\in\mathcal{G}\right.\right\}.
\eeq
To prove this we only need to show that $g^{\dagger}\opH_{SB}'g = \opH_{SB}'$ for all $g
\in\mathcal{G}$, since this immediately implies that $[H_{SB}',g]=0$ $\forall g\in\mc{G}$. Indeed,
\bes
\label{eq:proj-cent}
\bea
g^\dagger\opH_{SB}g &=& \frac{1}{K+1} \sum_{j=0}^{K} g^{\dagger} g_j^\dagger \opH_{SB} g_j g\\
&=& \frac{1}{K+1} \sum_{j=0}^{K} \lp g_j g \rp^\dagger \opH_{SB} \lp g_j g\rp \\
&=& \opH_{SB}',
\eea
\ees
since by group closure $\{g_j g\}_{j=0}^K$ also covers all of $\mc{G}$.

The fact that $\opH_{SB}'$ commutes with everything in $\mc{G}$ means that we can apply Schur's Lemma \cite{Hammermesh:grouptheorybook}:
\begin{mylemma}[Schur's Lemma]
\label{lem:schur}
Let $\mathcal{G}=\left\{ g_i \right\}$ be a group. Let $\mathrm{T}(\mathcal{G})$ be an irreducible $d$-dimensional
representation of $\mathcal{G}$ (i.e. not all of the $\mathrm{T}(g_i)$ are similar to a block-diagonal matrix).
If there is a $d\times d$ matrix $A$ such that $\left[A,g_i\right] = 0\; \forall g_i \in \mathcal{G} $, then
$A \propto {I}$.
\end{mylemma}
Thus, it follows from this lemma that, provided we pick $\mc{G}$ so that its matrix representation over the relevant system Hilbert space is irreducible, then indeed $\opH_{SB}' \propto I_S$, since it already commutes with every element of $\mc{G}$.

For example, $\mathcal{H}_S = \lp \C^2 \rp ^{\otimes n} = \C^{2^n}$ for $n$ qubits; the dimension of the irrep should then be $2^n$ in this case. 
Which decoupling group has a $2^n$-dimensional irrep over $\lp \C^2 \rp ^{\otimes n}$? An example is the $n$-fold tensor product of the Pauli group: $\mc{G} = \pm,\pm i \{I\otimes \cdots \otimes I,X\otimes I\otimes \cdots \otimes I,\dots, Z\otimes \cdots \otimes Z\}$. And indeed, this decoupling group suffices to decouple the most general system-bath
Hamiltonian in the case of $n$ qubits:
\beq
H_{SB} = \sum_\a \s_1^{\a_1}\otimes \cdots \otimes \s_n^{\a_n} \otimes B^{\a},
\label{eq:HSB-gen}
\eeq
where $\a = \{\a_1,\dots,\a_n\}$, and $\a_i \in \{0,x,y,z\}$, with the convention that $\s^0=I$. Fortunately, such a system-bath interaction is completely unrealistic, since it involves $n+1$-body interactions. ``Fortunately,'' since the decoupling group we just wrote down has $K-1=4^n$ elements, so that the time it would take to apply just once symmetrization sequence \eqref{eq:cycle} grows exponentially with the number of qubits, and we would only achieve first order decoupling (there is still a correction term proportional to $T^2$).

Actually, this approach using Schur's lemma is a bit too blunt. We have already seen that the Pauli group is too much even for a single qubit; the Pauli group has $16$ elements, but the $4$-element Klein group already suffices. Clearly, the approach suggested by Schur's lemma (looking for a group with a $2^n$-dimensional irrep) is sufficient but not necessary.
Moreover, as we shall see, it is possible to drastically reduce the required resources for decoupling, for example by combining decoupling with DFS encoding, or by focusing on more reasonable models of system-bath interactions.

\section{Combining dynamical decoupling with DFS}
\label{sec:DD-DFS}

We saw that to decouple the general system-bath interaction
$H_{SB}$ in Eq.~\eqref{eq:HSB-gen} would require a group with an exponentially large number of elements. This is not only impractical, it might also destroy any benefit we would hope to get from efficient quantum algorithms. Therefore we now consider  ways to shorten the decoupling sequence. As we'll see, this is possible, at the expense of of using more qubits. There will thus be a space-time tradeoff. For an entry into the original literature on this topic see Ref.~\cite{Byrd:2002:047901}.

\subsection{Dephasing on two qubits: a hybrid DFS-DD approach}

Consider a system consisting of two qubits that are coupled to a bath
by the dephasing interaction
\begin{equation}
H_{SB}=\sigma_{1}^{z}\otimes B_{1}^{z}+\sigma_{2}^{z}\otimes B_{2}^{z}.
\end{equation}
This
Hamiltonian is not invariant under swapping the two qubits
since they couple to different bath operators. To make this more apparent,
rewrite the interaction as
\begin{equation}
H_{SB}=\left(\frac{\sigma_{1}^{z}-\sigma_{2}^{z}}{2}\right)\otimes B_{-}+\left(\frac{\sigma_{1}^{z}+\sigma_{2}^{z}}{2}\right)\otimes B_{+}
\end{equation}
where the redefined bath operators are $B_{\pm}=B_{1}^{z}\pm B_{2}^{z}$.
We find that $\left(\frac{\sigma_{1}^{z}+\sigma_{2}^{z}}{2}\right)$
is a ``collective dephasing'' operator that applies the same dephasing
to both qubits, while the ``differential dephasing" operator $\left(\frac{\sigma_{1}^{z}-\sigma_{2}^{z}}{2}\right)$
applies opposite dephasing to the two qubits. From our DFS studies we already know that we can encode a single logical qubit as $|\bar{0}\rangle=|01\rangle$ and $|\bar{1}\rangle=|10\rangle$, just as in Eq.~\eqref{eq:DFS0110}.
Having chosen a basis that vanishes under the effect of
one part of the interaction hamiltonian, this effectively reduces the
interaction to
\begin{equation}
H_{SB}|_{\textrm{DFS}}=\lp\frac{\sigma_{1}^{z}-\sigma_{2}^{z}}{2}\rp\otimes B_{-}=\bar{\s}^z\otimes B_{-}\end{equation}

If the initial interaction had been symmetric, choosing the DFS would have
reduced it to zero. However, the interaction was not symmetric in
this case, and we are left with the above differential dephasing term. We notice further
that the residual term is the same as a $\bar{\s}^z$, or logical $Z$
operating on the DFS basis [this is a symmetrized version of the logical $Z$ operator in Eq.~\eqref{eq:logXZ}]. We recall that dephasing acting on a single
qubit was decoupled by pulses that implemented the $X$ or $Y$ operators,
and hence expect that the $\bar{\s}^z$ interaction can be decoupled using a $\bar{X}$
or $\bar{Y}$ pulse. We'll use the convention that logical/encoded terms in the Hamiltonian are denoted by $\bar{\s}^\alpha$, while the corresponding unitaries are denoted by $\bar{X}$, $\bar{Y}$, or $\bar{Z}$. Thus
\begin{equation}
\bar{\s}^x=\frac{\s^x_1\s^x_2+\s^y_1\s^y_2}{2}\quad\bar{\s}^y=\frac{\s^y_1\s^x_2-\s^x_1\s^y_2}{2}.
\end{equation}

Restricted to the DFS, the implementation of an $\bar{X}$ pulse using a $\bar{\sigma}^{x}$
is analogous to the implementation of an $X$ pulse by applying $\sigma^{x}$
for an appropriate period of time:
\bes
\label{eq:XX}
\begin{eqnarray}
e^{-i\frac{\pi}{2}\bar{\s}^x} & = & e^{-i\frac{\pi}{4}(\s^x_1\s^x_2+\s^y_1\s^y_2)}\\
 & = & e^{-i\frac{\pi}{4}\s^x_1\s^x_2}e^{-i\frac{\pi}{4}\s^y_1\s^y_2}\quad(\textrm{using }[\s^x_1\s^x_2,\s^y_1\s^y_2]=0)\\
 & = & \frac{1}{\sqrt{2}}[I-i\s^x_1\s^x_2]\frac{1}{\sqrt{2}}[I-i\s^y_1\s^y_2]
 \\
 & = & \frac{1}{2}[I-i\s^x_1\s^x_2-i\s^y_1\s^y_2+\s^z_1\s^z_2] \label{eq:neglectII+ZZ}\\
 & = & -\frac{i}{2}(\s^x_1\s^x_2+\s^y_1\s^y_2)=-i\bar{X},
 \end{eqnarray}
\ees
where the term $I+\s^z_1\s^z_2$  in Eq.~\eqref{eq:neglectII+ZZ} was ignored since it vanishes on the DFS.

Hence, the dynamical decoupling process is effective in the sense that\beq
\bar{X}f_{\tau}\bar{X}f_{\tau}|_{\textrm{DFS}}=\bar{I}\otimes\exp(-2i\tau\tilde{B})+\mathcal{O}[(2\tau)^{2}],
\eeq
where $\tilde{B}$ is a bath operator whose exact form does not matter, since we have obtained a pure-bath operator up to a time $\mathcal{O}[(2\tau)^{2}]$. The notation $\bar{I}$ denotes the identity operator projected to the DFS. What have we learned from this example? That we don't need to remove every term in the system-bath Hamiltonian; instead we can use a DFS encoding along with DD. Next we'll see how this can save us some pulse resources.

\subsection{General decoherence on two qubits: a hybrid DFS-DD approach}

We now consider the most general system-bath Hamiltonian on two qubits:
\begin{equation}
H_{SB}=\sum_{\alpha_{1},\alpha_{2}}(\sigma_{1}^{\alpha_{1}}\otimes\sigma_{2}^{\alpha_{2}})\otimes B^{\alpha_{1}\alpha_{2}},
\end{equation}
where $\alpha_i\in\{0,x,y,z\}$.

Within the framework of the same DFS as earlier (DFS=span$\{|01\rangle,|10\rangle\}$),
we can classify all possible ($4^{2}=16$) system operators as either
\begin{itemize}
\item leaving system states unchanged (i.e., acting as proportional to $\bar{I}$)
\item mapping system states to other states within the DFS; these correspond
to logical operations (these are errors since they occur as a result of interaction with the bath)
\item transitions from the DFS to outside and vice versa (``leakage'')
\end{itemize}
The operators causing these errors and their effects are given in Table~\ref{tab:2DFS}.
\begin{center}
\begin{table}
\begin{tabular}{c|c}
Effect on DFS states & Operators\tabularnewline
\hline\hline
unchanged & $I$,$\s^z_1+\s^z_2$,$\s^z_1\s^z_2$,$\s^x_1\s^x_2-\s^y_1\s^y_2$,$\s^x_1\s^y_2+\s^y_1\s^x_2$\tabularnewline
 & \tabularnewline
\hline
logical op. & $\bar{\sigma}^{z}$,$\bar{\sigma}^{y}$,$\bar{\sigma}^{x}$\tabularnewline
 & \tabularnewline
\hline
leakage & $\s^x_1$,$\s^x_2$,$\s^y_1$,$\s^y_2$,$\s^x_1\s^z_2$,$\s^z_1\s^x_2$,$\s^y_1\s^z_2$,$\s^z_1\s^y_2$\tabularnewline
\end{tabular}
\caption{Classification of all two-qubit error operators on the DFS for collective dephasing.}
\label{tab:2DFS}
\end{table}
\end{center}
For example, $\s^z_1\s^z_2$ acts on $\ket{\bar{0}}=\ket{01}$ and $\ket{\bar{1}}=\ket{10}$ as $-\bar{I}$, while $\s^x_1$ takes both $\ket{\bar{0}}$ and $\ket{\bar{1}}$ out of the DFS, to $\ket{11}$ and $\ket{00}$, respectively.

Along the same lines as single qubit dynamical decoupling, we look for
an operator that anticommutes with all the leakage operators, and an operator that anticommutes with all the logical operators. Performing a calculation very similar to Eq.~\eqref{eq:XX} we find that
\begin{itemize}
\item $\exp(-i\pi\bar{\sigma}^{x}) = Z_1Z_2$ and anticommutes with the \emph{entire} leakage set
\item $\exp(-i\frac{\pi}{2}\bar{\sigma}^{z}) = -i\bar{Z}$ and anticommutes with the
logical error operators $\bar{\s}^x,\bar{\s}^y$.
\end{itemize}
A combination of these operators, along with $\bar{X}$ which we used above, is sufficient to reduce the effect
of the system-bath interaction Hamiltonian to that of a pure bath operator that
acts trivially on the system. First we apply $\bar{X}$ to decouple
the logical error operators $\bar{\s}^z$ and $\bar{\s}^y$, giving us a net unitary evolution
\begin{subequations}
\begin{align}
U_{1}(2\tau)&=\bar{X}f_{\tau}\bar{X}f_{\tau} \\
&=\exp[-2i\tau(\bar{\s}^x\otimes \bar{B}^x+\sum_{j=1}^8 \textrm{leak}_j)\otimes B^{(j)}]+ \mc{O}[(2\tau)^2]\ ,
\end{align}
\end{subequations}
where the sum is over the $8$ leakage operators shown in Table~\ref{tab:2DFS} and $\bar{B}^x$ and $B^{(j)}$ are bath operators. Thus we still have to compensate for the logical $\s^z$ error and the leakage errors. The order in which we do this doesn't matter to first order in $\tau$, so let us remove the leakage errors next. 
This
is accomplished by using a ${ZZ}$
pulse:
\begin{subequations}
\begin{align}
U_{2}(4\tau)&=ZZ\cdot U_{1}(2\tau)\cdot ZZ\cdot U_{1}(2\tau) \\
&=\exp[-4i\tau\bar{\s}^x\otimes \bar{B}^x]+\mc{O}[(4\tau)^2]
\end{align}
\end{subequations}
All that remains now is to remove the logical error operator $\bar{\s}^x$, since it commutes with both the $\bar{X}$ and $ZZ$ pulses we have used so far. This can be performed using $\bar{Z}$, which anticommutes with $\bar{\s}^x$. Hence, the overall time evolution
that compensates for all possible (logical and leakage) errors is
of period $8\tau$ and is of the form
\begin{subequations}
\begin{align}
U_{3}(8\tau)&=\bar{Z}U_{2}(4\tau)\bar{Z}U_{2}(4\tau)\\
&=\bar{I}\otimes e^{-8i\tau B'}+\mc{O}[(8\tau)^2].
\end{align}
\end{subequations}
A qubit encoded into the $|\bar{0}\rangle=|01\rangle$ and $|\bar{1}\rangle=|10\rangle$ DFS is acted on (at time $T=8\tau$) only by the innocuous operators in the first row of Table~\ref{tab:2DFS}. As a result it is completely free of decoherence, up to errors appearing to $\mathcal{O}(T^{2})$, while we used a pulse sequence that has length $8\tau$, shorter by a factor of $2$ compared to the sequence we would have had to use without the DFS encoding (the full two-qubit Pauli or Klein group). This, then, illustrates the space-time tradeoff between using full DD without DFS encoding, \textit{vs} using a hybrid approach, where we use up twice the number of qubits, but gain a factor of two in time. 

However, we could have of course also simply discarded one of the two qubits and used the length-$4$ universal decoupling sequence for a single qubit. In this sense the current example isn't yet evidence of a true advantage. Such an advantage emerges when one considers constraints on which interactions can be controlled. The method we have discussed here essentially requires only an ``XY" type interaction, governed by a Hamiltonian with terms of the form $\s^x\otimes \s^x+\s^y\otimes \s^y$ \cite{Lidar:2001vn}. 

For a discussion of how to generalize the construction we have given here to an arbitrary number of encoded qubits, see \cite{Wu:2002ys}.

\section{Concatenated dynamical decoupling: removing errors of higher order
in time}
\label{CDD}

The dynamical decoupling techniques considered so far have all involved
elimination of decoherence up to first order in time. We now consider
the question of whether it is possible to improve upon these techniques
and remove the effect of noise up to higher orders in time. We saw
in the earlier sections that applying pulses corresponding to the
chosen decoupling group effectively causes a net unitary evolution
\beq
U^{(1)}(T_1)=\prod_{i=0}^{K}g_{i}^{\dagger}U^{(0)}(\tau)g_{i}
\eeq
where $U^{(0)}(\tau)\equiv U_{f}=e^{-iH\tau}$ is the free unitary evolution operator, $T_1 = k \tau$, where we let $k\equiv K+1$ and $\tau \equiv T_0$ is the free evolution duration.

A \emph{concatenated} dynamical decoupling (CDD) sequence is defined recursively for $m\geq 1$ as
\beq
U^{(m)}(T_{m})=\prod_{i=0}^{K}g_{i}^{\dagger}U^{(m-1)}(T_{m-1})g_{i} ,
\eeq
lasting total time 
\beq
T_m = k T_{m-1} = k^m \tau .
\eeq
For example, we could concatenate the universal decoupling sequence $U^{(1)}(T_{1}) = Z f_\tau X f_\tau Z f_\tau X  f_\tau$ [Eq.~\eqref{eq:sym_sch}], where $T_1 = 4\tau$, in this manner. The second order sequence we would obtain is then
\begin{subequations}
\begin{align}
U^{(2)}(T_{2}) &= Z U^{(1)}(T_{1}) X U^{(1)}(T_{1}) Z U^{(1)}(T_{1}) X  U^{(1)}(T_{1})\\
&=Z [Z f_\tau X f_\tau Z f_\tau X  f_\tau] X [Z f_\tau X f_\tau Z f_\tau X  f_\tau] Z [Z f_\tau X f_\tau Z f_\tau X  f_\tau] X  [Z f_\tau X f_\tau Z f_\tau X  f_\tau] \\
&= f_\tau X f_\tau Z f_\tau X  f_\tau Y f_\tau X f_\tau Z f_\tau X  f_\tau f_\tau X f_\tau Z f_\tau X  f_\tau Y f_\tau X f_\tau Z f_\tau X  f_\tau .
\end{align}
\end{subequations}
Note that while some of the pulses have been compressed using equalities such as $Z^2=I$ and (up to a phase) $XZ=Y$, the total duration of the sequence is dictated by the number of free evolution intervals, which is $16$ in this case. Note that while this second order sequence is time-reversal symmetric, this is not a general feature; indeed the first order sequence isn't, nor is the third order sequence, as is easily revealed by writing it down.

To analyze the performance of CDD, we start by rewriting, without loss of generality, $H = H_{SB} + I_S\otimes H_B$ as
\begin{equation}
H\equiv H^{(0)}=H_{C}^{(0)}+H_{NC}^{(0)}
\end{equation}
where $H_{C}^{(0)}$ commutes with the group $\mathcal{G}$ and $H_{NC}^{(0)}$
does not ($H_{C}^{(0)}$ includes $I_S\otimes H_B$ for sure, and maybe also a part of $H_{SB}$). This split is done in anticipation of our considerations below.
We proceed by using the BCH formula, which yields
\begin{subequations}
\begin{align}
U^{(1)}(T_1)&=\exp\left\{-i\tau\sum_{i}g_{i}^{\dagger}H^{(0)} g_{i}+\frac{\tau^{2}}{2}\sum_{i<j}[g_{i}^{\dagger}H^{(0)} g_{i},g_{j}^{\dagger}H^{(0)} g_{j}]+\mathcal{O}(\tau^{3})\right\} \\
&\equiv e^{-i \tau H^{(1)}}
\end{align}
\end{subequations}
where the effective Hamiltonian $H^{(1)}$ can be decomposed as
\begin{subequations}
\begin{align}
\label{eq:H^1}
H^{(1)}  =& H_{C}^{(1)}+\tau H_{NC}^{(1)} + \mathcal{O}(\tau^{2}) \\
H_{C}^{(1)} &\equiv \sum_{i}g_{i}^{\dagger}H^{(0)}g_{i} \\
H_{NC}^{(1)} &\equiv \frac{\tau}{2}\sum_{i<j}[g_{i}^{\dagger}H^{(0)}g_{i},g_{j}^{\dagger}H^{(0)}g_{j}]
\end{align}
\end{subequations}
and where as we have already seen in Eq.~\eqref{eq:proj-cent}, $H_{C}^{(1)}$ lies in the centralizer of $\mathcal{G}$, i.e., $[{H}_{C}^{(1)},g]=0$ $\forall g\in \mathcal{G}$. Thus, the
first order term in the BCH series is harmless. Note the important fact that the bad, non-commuting (with $\mathcal{G}$) term $H_{NC}^{(1)}$ is of $\mathcal{O}(\tau)$, while the good, commuting term $H_{C}^{(1)}$ is of $\mathcal{O}(1)$. 
%We now begin the process of concatenation, which essentially amounts to removing the lowest order term in $\tau$ of $H_{NC}^{(1)} $ using a second iteration of the group action. This new sequence, $U^{(2)}$ will have an effective Hamiltonian, $H^{(2)}$  such that $H_{NC}^{(2)}$ is proportional to $\mathcal{O}(\tau^{3})$ or higher. This can be repeated, until the $m$th iteration where $H_{NC}^{(m)}$ is proportional to $\mathcal{O}(\tau^{m+1})$ or higher.

Next let us see how this plays out in  $U^{(2)}$:
\begin{subequations}
\begin{align}
U^{(2)}(T_2)&=\prod_{i=0}^{K}g_{i}^{\dagger}U^{(1)}(T_1)g_{i}\\
\label{eq:337b}
&=\exp\left\{-i\tau\sum_{i}g_{i}^{\dagger}H^{(1)} g_{i}+\frac{\tau^{2}}{2}\sum_{i<j}[g_{i}^{\dagger}H^{(1)} g_{i},g_{j}^{\dagger}H^{(1)} g_{j}]+\cdots\right\} \\
&\equiv e^{-i \tau H^{(2)}}
\end{align}
\end{subequations}
where the effective Hamiltonian $H^{(2)}$ can be decomposed as
\begin{align}
H^{(2)}  = H_{C}^{(2)}+\tau^2 H_{NC}^{(2)} + \mathcal{O}(\tau^{3}) .
\label{eq:H^2}
\end{align}
It is tempting but slightly counterproductive to try to work out the exact form of the $H_{C}^{(2)}$ and $H_{NC}^{(2)}$ terms. Clearly, $H_{C}^{(2)}$ contains 
both $\sum_{i}g_{i}^{\dagger}H^{(1)}g_{i}$ and the terms involving commutators of only $H_C^{(1)}$ (without any $H_{NC}^{(1)}$). The key point is that the lowest order term in Eq.~\eqref{eq:337b} that does not commute with
$\mathcal{G}$ (and hence is responsible for decoherence), is of the type $\tau^2 \sum_{i<j}[g_{i}^{\dagger}H_C^{(1)}g_{i},g_{j}^{\dagger}H_{NC}^{(1)}g_{j}]$ and hence is of $\mathcal{O}(\tau^{3})$. 

Comparing Eq.~\eqref{eq:H^1} and \eqref{eq:H^2} we see that the order of the non-commuting term increased by $1$, from $\tau$ to $\tau^2$. This will clearly continue as we proceed to higher concatenation levels, so that 
\begin{align}
U^{(m)}(T_m)\equiv e^{-i \tau H^{(m)}},
\end{align}
where
\begin{align}
H^{(m)}  = H_{C}^{(m)}+\tau^m H_{NC}^{(m)} + \mathcal{O}(\tau^{m+1}) .
\end{align}
This is a remarkable result: it tells us that using concatenation we can push the order of the error term $H_{NC}^{(m)}$ to become arbitrarily high. Of course the price is an exponentially growing pulse sequence length, but this is a price that may be worth paying if we can make the error shrink fast enough.

Note, however, that in principle the norm of $H_{NC}^{(m)}$ may grow with $m$ and thus is it not yet obvious at this point that increasing orders of CDD concatenation implies a better performance. Let us show that in fact there is an optimal level of concatenation (see Ref.~\cite{NLP:09} for a rigorous analysis; here we adapt the more intuitive presentation in Ref.~\cite{Khodjasteh:2007zr}).

Given an $m$th level CDD sequence one can define a dimensionless ``error phase",
\beq
\phi_{\textrm{CDD}}(m)= \tau^{m+1} \| H_{NC}^{(m)}\|,
\eeq
Let 
\beq
J\equiv \Vert H_{SB}\Vert, \quad \beta\equiv\Vert H_{B}\Vert .
\eeq
Assuming $J < \b$ and specializing to the universal DD sequence [Eq.~\eqref{eq:sym_sch}] as the base sequence for the concatenation, this error phase can be bounded as \cite[Eq.~(46)]{Khodjasteh:2007zr}
\begin{equation}
\phi_{\textrm{CDD}}(m) \leq T_m(2^{m^{2}}(\beta\tau)^{m}J),
\end{equation}
where $m$ is the degree of concatenation and now $T_m = 4^m \tau$. It follows then that for an $m$-level CDD sequence we have the following two possibilities:
\begin{itemize}
\item Assume that $T_m = T$ is fixed, i.e., $\tau = T/4^m$ can be made arbitrarily small. Then $\phi_{\textrm{CDD}}(m)\leq T (2^{m^{2}}(\beta T /4^{m})^{m}J) = JT (\b T/2^m)^m $, and thus the noise strength decreases monotonically as the concatenation level $m$ increases, as soon as $2^m >\beta T$. In practice, what this implies is that for a fixed $T$, provided $\beta T$ is small enough, the more concatenations of DD the better. However, it is of course not possible in practice to continuously reduce the pulse interval.
\item On the other hand, if $\tau$ has a minimum physically achievable value, such that a higher concatenation level corresponds to a longer pulse sequence, then there is an optimal level of concatenation. To see this, consider how the upper bound $4^m \tau (2^{m^{2}}(\beta\tau)^{m}J)$ behaves as $m$ grows. We see that
\begin{align}
\log\phi_{\textrm{CDD}}(m) \leq \log(2) m^2 + \log(4 \beta \tau)m + \log (J \tau).
\end{align}
This is just a quadratic expression in $m$, a parabola with a minimum. Differentiating with respect to $m$ we easily find that the minimum is at $m_{\textrm{opt}}= - \frac{\log (4 \beta \tau)}{2 \log(2)}$. Thus $m_{\textrm{opt}} > 0$ whenever $T_1\beta = 4\tau \beta < 1$. Thus provided we choose the base pulse sequence length $T_1 < 1/\b$ then it helps to concatenate, up to the level $\lfloor m_{\textrm{opt}} \rfloor$. But for $m$ larger than this the CDD process loses its effectiveness.
\end{itemize}
Concatenation has been tested experimentally and the optimal concatenation level has been observed. See, e.g., \cite{Alvarez:2010ve}.

As a final note, it is possible to obtain the arbitrary-order error suppression using DD without the exponential cost of CDD. For some entries into this literature see, e.g., 
\cite{Uhrig:2007qf,West:2010dq,Wang:10}

\section{Dynamical Decoupling and Representation theory}
\label{sec:DD-rep}

In this section, our goal is to illustrate the connections between first-order DD and the result from representation theory, Theorem \ref{th:rep}, which is a theorem of fundamental importance in the theory of quantum error correction. See Refs.~\cite{Zanardi:1999:77,Viola:1999:4888,Viola:2000:3520} for entries into the original literature on this topic.

\subsection{Information storage and computation under DD}

We define the {\textit{group algebra}} $\C \mc{G}$ of the group $\mc{G}$ over the complex field as: $\C \mc{G} \equiv \{$ all linear combinations of the elements in $\mc{G}$, over $\C \}$. For example, if $\mc{G} = \{ I, \s^x, \s^y, \s^z \}$, the Pauli group, then
\beq
\C \mc{G} = \lb \sum_{\a = 0,x,y,z} a_{\a} \s^{\a} \rb,\quad a_{\a} \in \C.
\eeq

Now consider the group algebra of our decoupling group, $\C\mc{G}$. Clearly it is a matrix algebra of dimension $d \times d$ where $d=\dim(\mc{H}_S)$. We can always choose the group in such a way that $\C \mc{G}$ is $\dgr$-closed. Since every group includes the identity, we can thus invoke Theorem \ref{th:rep}, using which we have:
\beq
\label{eq:grpalb}
\C \mc{G} \cong \bigoplus_{J} I_{n_J} \ox  \mathcal{M}_{d_J},
\eeq
where $J$ is the irrep (irreducible representation) label, $n_J$ is the multiplicity of irrep $J$ and $d_J$ is the dimension of the irrep labelled by $J$. The system Hilbert space $\mc{H}_S$ is correspondingly partitioned into a direct sum of product spaces, which we can write as
\beq
\mc{H}_S \cong \bigoplus_{J} \C^{n_J} \ox \C^{d_J}.
\eeq
Hence, every DD pulse that we apply to the system acts like identity on $\C^{n_J}$ and as some non-trivial operation on $\C^{d_J}$. It's clear that if we store our quantum information in $\C^{n_J}$, then the pulses do not affect it (of course the system-bath interaction can still affect information stored in $\C^{n_J}$).

Now consider the \textit{{commutant}} of the group algebra $\C \mc{G}$ [see Eqs.~\eqref{eq:cmmt} and \eqref{eq:cmmt2}],
\bes
\begin{align}
\C \mc{G}' & \equiv \lb A | [A,\C\mc{G}]=0 \rb \\
& = \bigoplus_{J}  \mathcal{M}_{n_J} \ox I_{d_J} .\label{eq:grpcmmt}
\end{align}
\ees
We can immediately see from this definition and from Eq.~\eqref{Hsb} that, by linearity, the effective system-bath Hamiltonian, $H_{SB}'$, lies in the commutant of the group algebra, i.e., $H_{SB}' \in \C \mc{G}'$. It follows that $H_{SB}'$ can be represented as:
\beq
H_{SB}' = \bigoplus_{J} (H_{SB}')_{n_J} \ox I_{d_J},
\eeq	
i.e., the effective system-bath Hamiltonian has this block-diagonal representation, with blocks labeled by the irrep index $J$, each of dimension $n_J d_J$, and where each non-trivial factor $(H_{SB}')_{n_J}$ is an $n_J\times n_J$ matrix. It is clear that in this case, we can encode our quantum information in $\C^{d_J}$, since $H_{SB}'$ will act as identity on it. Note that it doesn't necessarily matter that the pulses themselves have a non-trivial effect on $\C_{d_J}$: because we know everything about the pulses, they can compensated for by applying appropriate transformations.

However, it is often desirable not to have to compensate for the action of the DD pulses. In that case we might we want to be able to store information in $\C^{n_J}$ rather than $\C^{d_J}$. If so, clearly we need to somehow make the effect of $H_{SB}'$ trivial on $\C^{n_J}$ as well. Since $H_{SB}'$ is determined by our group of pulses, $\mc{G}$, it boils down to choosing the appropriate set of pulses, i.e., picking the group $\mc{G}$ such that
\bes
\begin{align}
H_{SB}' & \in \C \mc{G}' \cap \C \mc{G} \\
& {=} \bigoplus_{J} \lambda_J I_{n_J} \ox I_{d_J}, \text{   } \lambda_J \in \C
\label{eq:intersect}
\end{align}
\ees
where Eq.~\eqref{eq:intersect} can be easily deduced by examining Eqs.~\eqref{eq:grpalb} and \eqref{eq:grpcmmt}. This means that if we pick an appropriate $\mc{G}$, $H_{SB}'$ will have the form:
\beq
H_{SB}'=
\begin{pmatrix}
\lambda_1 I_{n_1 d_1}&&&& \\
&\lambda_2 I_{n_2 d_2}&&& \\
&&\lambda_3 I_{n_3 d_3}&& \\
&&&\ddots& \\
&&&&&
\end{pmatrix}
\eeq

We have already seen examples where $H'_{SB}$ is of this form. E.g.,
\begin{itemize}
\item When the conditions of Schur's Lemma (\ref{lem:schur}) are satisfied then $H_{SB}' \propto I_d$. Thus, that result was a special case resulting from this more general structure.

\item If $\lambda_J = 0$ for all $J$, then $H_{SB}'=0$ and the system bath interaction is annihilated (to first order in $T$). This describes the kind of situation we obtained with the Klein group symmetrization [see Eq.~\eqref{eq:f4tau}].
\end{itemize}

So, now we have established that it is possible to protect information stored both in $\C^{n_J}$ and $\C^{d_J}$. Information stored in the latter is protected since the effective system-bath interaction acts like identity on that space; though we'd have to compensate for the influence of the pulses we would apply. Storing in the former requires us to more cleverly choose the pulse group $\mc{G}$ such that the effective system-bath interaction becomes a block-diagonal matrix, with each block proportional to identity. 

Our choice of information storage location might depend on which of $\max_J n_J$ and $\max_J d_J$ is greater, since that would provide us with a larger-dimensional space and hence more qubits. Or it might depend on whether we wish to apply computation as well, in which case using $\C^{n_J}$ would be preferred. Indeed, when Eq.~\eqref{eq:intersect} is satisfied we can make use of the commutant $\C \mc{G}'$ to perform computation! 

We now discuss some examples which illustrate these points:

\subsection{Examples}

In all of the following examples, our system is a set of $N$ qubits, i.e., ${\mc{H}}_S = (\C^2)^{\ox N}= \C^{2^{N}}$.

\subsubsection{Example 1: $\mc{G} = (\text{SU}(2))^{\ox N}$}
In this case the pulses are products of all $N$ arbitrary single-qubit unitaries, and the only operator which commutes with $\mc{G}$ is the identity operator, i.e., $H_{SB}' \in \C \mc{G}' = \mathbb{C}$. This is the Schur's Lemma situation again. Thus, in this case we have the choice of storing information in either the left-hand ($\C^{n_J}$) or right-hand factor ($\C^{d_J}$), or both if we do not mind compensating for the action of the pulses on the right-hand factor. It is interesting that in spite of the huge cost of directly implementing this group, it can be well approximated by picking the $\text{SU}(2)$ rotation on each qubit at random \cite{Viola:2005:060502}.

\subsubsection{Example 2: $\mc{G} =\text{Collective SU}(2)$}
Let us add to Example 1 the constraint that the same unitary matrix acts on every qubit, but every element of $SU(2)$ is implemented (again, this can be approximated using random elements of $SU(2)$ \cite{Viola:2005:060502}). Then $H_{SB}' \in \C \mc{G}' = \C S_{N}$, where $S_N$ is the permutation group on $N$ elements.
Thus the effective system-bath Hamiltonian is \emph{not} proportional to identity, and instead we have here the case where $H_{SB}'$ acts trivially only on the right-hand factors. Its action on the left-hand factors is to apply permutations.

The collective-SU(2) group is generated by the sum of the Pauli matrices, i.e., $\{\sum_{i} \s^{\a}_i\}_{\a=x,y,z}$. We already encountered these sums in the study of collective decoherence [see Eqs.~\eqref{eq:CD} and \eqref{eq:totspin}]. In this case, our pulse group has the same generating Hamiltonian as the system part of the system-bath interaction in collective decoherence: $H_{SB}^{\text{coll.dec.}} = \sum_{\a=x,y,z} \lp \sum_{i} \s^{\a}_i \rp \ox B^{\a}$. Since the DD group acts as identity on the left-hand factors, and the DD group behaves like the collective decoherence operators, while the effective system-bath interaction $H_{SB}'$ acts like the exchange operators \eqref{eq:Eij} we encountered in our study of computation over the DFS for collective decoherence, we see that the current situation is the reverse (or dual) of the situation back in the DFS case. In other words, we can invoke the machinery we developed then [see Eqs.~\eqref{eq:colld1}-\eqref{eq:colld3}] but we should flip the role of $d_J$ and $n_J$: $\mathcal{H}_S =\bigoplus_{J=0(1/2)}^{N/2} \mathbb{C}^{n_J}\otimes\mathbb{C}^{d_J}$, where now
\bes
\bea
n_J &=&2J+1, \\
d_J &=& \frac{(2J+1)N!}{(N/2+1+J)!(N/2-J)!}.
\eea
\ees
We can always pick an irrep $J$ so that $d_J > n_J$, and so we have here the same code rates as in the case of a DFS.

\subsubsection{Example 3: $\mc{G} = S_n$}
In this case $\mc{G}$ is the permutation group. As we saw in the last example, the permutation group is dual to collective SU(2). So in this case, $H_{SB}'\in\C \mc{G}' =$``collective decoherence". In fact, since the permutation group can be obtained by swaps (or transpositions), we could as well take $\mc{G} = \{ \text{SWAP}_{i,j} \}$. And we know that,
\beq
\text{SWAP}_{i,j} = \ee^{-\ii \pi \vec{\s}_i \cdot \vec{\s}_j /4},
\eeq
which is generated by the Heisenberg interaction, so the decoupling group is implementable in physical systems (such as quantum dots) where the exchange interaction is controllable (see Ref.~\cite{Wu:2002zr} for a discussion of how to use efficiently implement $\mc{G}$ in this case). The Hilbert space again splits as  $\mathcal{H}_S =\bigoplus_{J=0(1/2)}^{N/2} \mathbb{C}^{n_J}\otimes\mathbb{C}^{d_J}$, where now we have the irrep dimension and multiplicity formulas we encountered during the DFS study:
\bes
\bea
n_J &=& \frac{(2J+1)N!}{(N/2+1+J)!(N/2-J)!} \\
d_J &=&2J+1.
\eea
\ees
We have the option of encoding into the right-hand factor, where $H_{SB}'$ acts as identity, but then the space dimension is only $2J+1$. Alternatively, we can encode into the left-hand factors, where the effective system-bath interaction has non-trivial action, but it acts as collective decoherence, so that our DFS encoding will completely hide the quantum information from the action of $H_{SB}'$. This has the significant advantage (over right-hand factor encoding) of providing us with a code space of dimension $n_J$. This example leads us to the interesting conclusion that in this case in fact Eq.~\eqref{eq:intersect} applies, i.e., the effective system-bath interaction acts trivially everywhere.

\subsubsection{Example 4: Linear System-Bath coupling}
We consider a system-bath interaction of the form,
\beq
H_{SB} = \sum_{\a=x,y,z}  \sum_{i} \s^{\a}_i  \ox B^{\a}_i.
\eeq
Each qubit in this case has its own bath. In this noise model, we don't consider bilinear terms in the system such as $\s^{\a}_i\s^{\b}_j \ox B^{\a \b}_{ij}$, because this is a 3-body interaction which is typically much weaker in nature and also very hard to engineer.

The decoupling group we select for this is,
\beq
\mc{G}= \lb I^{\ox N}, X^{\ox N}, Y^{\ox N}, Z^{\ox N} \rb.
\eeq
We choose $N$ to be even. Therefore, $\mc{G}$ becomes abelian. And from representation theory, we know that all the irreps of an abelian group are 1-dimensional (scalars, so $d_J=1$ $\forall J$) and the number of irreps is the order of the group. Here $\abs{\mc{G}} = 4$. The irreps are:
\begin{align}
  \begin{array}{c|cccc}
  J & I^{\ox N} & X^{\ox N} & Y^{\ox N} & Z^{\ox N} \\
  \hline
  1 & 1 & 1 & 1 & 1 \\
  2 & 1 & 1 & -1 & -1 \\
  3 & 1 & -1 & 1 & -1 \\
  4 & 1 & -1 & -1 & 1 \\
  \end{array}
\label{eq:irreps-table}
\end{align}
The group algebra is,
\bes
\begin{align}
\C \mc{G} &\cong \bigoplus_{J} I_{n_J} \ox \mc{M}_{d_J} \\
& = \bigoplus_{J=1}^{4} c_J I_{2^{N-2}},\label{eq:2^N-2}
\end{align}
\ees
where in Eq.~\eqref{eq:2^N-2} $n_J = 2^{N-2}$ because $\sum_{J} n_J d_J = 2^N$ and all the $n_J$'s are equal (by use of a standard multiplicity formula from group theory \cite{Hammermesh:grouptheorybook}, or by noting the symmetry between the $X$, $Y$ and $Z$ operators), we have $4 n_J=2^N$, making $n_J = 2^{N-2}$.

Thus the group algebra can represented as a block diagonal matrix, each block being proportional to an identity of dimension $2^{N-2}$ with proportionality constant $c_J$, i.e., we have
$(N-2)$ qubits in each block that will be unaffected by the system-bath interaction.

Let us pick the first (trivial) irrep to encode our qubits into, i.e., $\{1,1,1,1\}$. In this irrep each pulse acts as $1$, so we're looking for  code-states which are ``stabilized'' by the group (each group element acts as identity). After a bit of thought it is clear that such states are of the form:
\beq
\label{eq:psir}
\ket{\ps_r} \equiv \frac{1}{\sqrt{2}} \lp \ket{r} + \ket{\bar{r}} \rp,
\eeq
where $\bar{r} = \text{NOT}(r)$ and $r\in \{0,1\}^N$ is an \textit{even} weight binary string of $N$ bits, i.e., $r$ contains an even number of $1$'s. Then it is easy to see that the action of any member of the decoupling group $\mc{G}$, leaves $\ket{\ps_r}$ unchanged, so indeed $\ket{\ps_r}$ belongs to the trivial irrep.

Why did we pick this decoupling group? Because it has a couple of very interesting and useful features which we list:

{\textit{Feature 1:}}  We can show that $H_{SB}'=0$.
\begin{proof}
\bes
\begin{align}
H_{SB}' &= \frac{1}{\abs{\mc{G}}} \sum_{j=0}^{\abs{\mc{G}}} g^{\dgr}_j H_{SB}' g_j \\
&= \frac{1}{4} \ls I H_{SB} I + X^{\ox N} H_{SB} X^{\ox N} + Y^{\ox N} H_{SB} Y^{\ox N} + Z^{\ox N} H_{SB} Z^{\ox N} \rs \\
&{=} \frac{1}{4} \Big[  H_{SB}  + \sum_{i} \big(  \s^X_i \ox B^X_i - \s^Y_i \ox B^Y_i - \s^Z_i \ox B^Z_i  \nonumber \\
				& \qquad \qquad \qquad         -\s^X_i \ox B^X_i + \s^Y_i \ox B^Y_i - \s^Z_i \ox B^Z_i  \nonumber \\
				& \qquad \qquad \qquad         -\s^X_i \ox B^X_i - \s^Y_i \ox B^Y_i + \s^Z_i \ox B^Z_i \big) \big] \label{eq:flipsign}\\
&= 0,
\end{align}
\ees
where to arrive at Eq.~\eqref{eq:flipsign} we used the properties of the Pauli group to do the multiplication. For example,
\bes
\begin{align}
X^{\ox N} \s^Y_i X^{\ox N} &= (\s^X_1\ox \dots \ox \s^X_i \ox \dots \ox \s^X_N)(\s^Y_i) (\s^X_1\ox \dots \ox \s^X_i \ox \dots \ox \s^X_N) \\
&= -\s^Y_i,
\end{align}
\ees
and so on.
\end{proof}
Thus, this decoupling group eliminates the system-bath interaction completely, to first order. Recall that this means that we can use the left-hand factor to encode and store information. This allows us to perform computation using the commutant, which has non-trivial action on the left-hand factor.

{\textit{Feature 2:}} We can do computation on the decoherence-protected qubits.

For $N=2$, i.e., for $2$ physical qubits, we have no logical qubits, as the only one state possible according to Eq.~\eqref{eq:psir} is $\ket{\psi_r} = \frac{1}{\sqrt{2}} (\ket{00} + \ket{11})$. This agrees with the fact that, since each irrep provides us with $N-2$ logical qubits, we have zero logical qubits for $2$ physical qubits.

Let us list the possible states in the case of $N=4$, using Eq.~\eqref{eq:psir}
\bes
\begin{align}
\ket{\psi_{0000}}&=\frac{1}{\sqrt{2}} (\ket{0000} + \ket{1111}) \equiv \ket{\bar{0}\bar{0}}, \\
\ket{\psi_{0011}}&=\frac{1}{\sqrt{2}} (\ket{0011} + \ket{1100}) \equiv \ket{\bar{1}\bar{0}}, \\
\ket{\psi_{0101}}&=\frac{1}{\sqrt{2}} (\ket{0101} + \ket{1010}) \equiv \ket{\bar{0}\bar{1}}, \\
\ket{\psi_{0110}}&=\frac{1}{\sqrt{2}} (\ket{0110} + \ket{1001}) \equiv \ket{\bar{1}\bar{1}},
\end{align}
\ees
and since we have $4$ orthonormal states, we can use them as $2$ qubits. Again, this agrees with the fact that $N-2=2$ in this case. 

How do we perform computations on these states? For that we use the commutant of the group,
\bes
\begin{align}
\C \mc{G}' & \cong \bigoplus \mathcal{M}_{n_J} \ox I_{d_J} \\
& = \bigoplus_{J=1}^{4} c_J \mathcal{M}_{2^{N-2}}
\end{align}
\ees
where $c_J$ are scalars; they are the columns of Table~\eqref{eq:irreps-table} (e.g., for $X^{\otimes N}$ we have $c_1=c_2=1$ and $c_3=c_4=-1$). We can check that the commutant can be generated by $\{ X_1 X_{j+1}\}_{j=1}^{N-2} \cup \{ Z_{j+1} Z_N \}_{j=1}^{N-2}$. For example, we don't require $Y_i Y_j$ to be in the generating set, since $Y_i Y_j = X_i X_j Z_i Z_j$. Therefore, for $N=4$, the generating set for the commutant becomes $\{ X_1 X_2, X_1 X_3, Z_2 Z_4, Z_3 Z_4 \}$.

Let's check the action of $X_1 X_2$ on our first logical state $\ket{\psi_{0000}} \equiv \ket{\bar{0}\bar{0}}$,
\bes
\begin{align}
X_1 X_2 \ket{\bar{0}\bar{0}} & = (X_1 \ox X_2 \ox I \ox I) \lp \frac{1}{\sqrt{2}} (\ket{0000} + \ket{1111}) \rp \\
&= \frac{1}{\sqrt{2}} (\ket{0011} + \ket{1100}) \equiv \ket{\bar{1}\bar{0}}.
\end{align}
\ees
We can similarly check that $X_1 X_2 \ket{\bar{0}\bar{1}} = \ket{\bar{1}\bar{1}}$. Therefore $X_1 X_2$ is logical Pauli-$X$ on the first logical qubit. We write this as $X_1 X_2 = \bar{X_1}$, where the bar denotes a logical operator. Similarly we can verify that:
\begin{align}
\label{eq:lpauli}
X_1 X_{j+1} = \bar{X_j}, \\
Z_{j+1} Z_N = \bar{Z_j}.
\end{align}
And having obtained logical-$X$ and logical-$Z$ (note that they anti-commute, as they should), we can implement {any one qubit gate we like} using only two-body interactions. And moreover, these logical gates lie in the commutant of the decoupling group. Therefore, we can compute while at the same time applying DD. 

But, in order to perform universal quantum computation, we also need to be able to perform entangling operations on two qubits, for example the CNOT gate. Let's examine $\bar{X_i} \bar{X_j}$ which is a logical entangling operation :
\bes
\begin{align}
\bar{X_i} \bar{X_j} & = X_1 X_{i+1} X_1 X_{j+1} \\
&= X_{i+1} X_{j+1},
\end{align}
\ees
which is a physical $2$-body interaction. So, we have managed to implement a logical entangling operation using only a $2$-body interaction. 

It is well known that if we can implement any Hamiltonian of the form,
\beq
\label{eq:univ}
H_S = \sum_{i} \o^{X}_i(t) \s^{X}_i + \sum_{i} \o^{Z}_i(t) \s^{Z}_i + \sum_{\a\in\{X,Z\}}\sum_{i,j} J_{ij}(t) \s^{\a}_i \s^{\a}_j,
\eeq
then we can perform universal quantum computation \cite{nielsen2000quantum}.
Therefore, in our case, replacing the Pauli operators in Eq.~\eqref{eq:univ} with their logical counterparts and expanding them in terms of their decompositions \eqref{eq:lpauli} in the physical qubit space, we obtain:
\bes
\bea
\bar{H}_S &=& \sum_{i=1}^{N-2} \bar{\o}^{X}_i(t) \bar{\s}^{X}_{i} + \sum_{i=1}^{N-2} \bar{\o}^{Z}_{i}(t) \bar{\s}^{i}_N + \sum_{\a\in\{X,Z\}}\sum_{i,j} \bar{J}_{ij}(t) \bar{\s}^{\a}_{i} \bar{\s}^{\a}_{j}\\
&=& \sum_{i=1}^{N-2} \o^{X}_{i+1}(t) \s_1^x\s^x_{i+1}\ + \sum_{i=1}^{N-2} {\o}^{Z}_{i+1}(t) {\s}^{z}_{i+1}\s^z_N + \sum_{\a\in\{X,Z\}}\sum_{i,j} J_{i+1,j+1}(t) \bar{\s}^{\a}_{i+1} \bar{\s}^{\a}_{j+1}
\eea
\ees
Remarkably,  we have obtained \textit{decoupling and universal quantum computation} on the logical qubits using only two-body interactions on the physical qubits. A similar approach has been studied numerically in the context of the $4$-qubit DFS code, with gates protected by CDD, showing evidence for a highly robust set of universal gates \cite{West:10}. Related ideas apply in the context of adiabatic quantum computation 
\cite{PhysRevLett.100.160506}.

\section{Conclusions}
\label{sec:conc}

This review has covered a selection of topics in the theory of decoherence-free subspaces, noiseless subsystems, and dynamical decoupling. We have seen how these tools allow one to hide information from the environment, and when this hiding is imperfect, how dynamical decoupling allows us to suppress the remaining residual decoherence. Moreover, we have shown explicitly how universal quantum computation is compatible with decoherence avoidance and suppression.   

Many important topics were left out in this brief review. For example, apart from CDD we did not address high-order decoupling methods, in particular schemes based on 
optimized pulse intervals \cite{Uhrig:2007:100504,West:2010:130501,Wang:2011:022306,Xia:2011uq}. Nor did we address the filter function approach to DD \cite{Uys:2009:040501}, optimized continuous modulation \cite{Gordon:2008:010403,Clausen:2010:040401}, or randomized decoupling, which is well suited to strongly time-dependent baths \cite{Santos:2006:150501}. It is important to stress that beyond decoherence avoidance and suppression, the theory of noiseless subsystems gave rise also to important advances in the theory of quantum error correcting codes, such as operator quantum error correction \cite{Kribs:2005:180501}. However, perhaps our greatest omission has been the abundance of experimental results which have both confirmed and driven the theoretical developments described here. For an entry into that literature, as well as many additional theoretical topics, see Ref.~\cite{Lidar-Brun:book}.
Nevertheless, hopefully we have given the reader the tools and inspiration to delve deeper into the large and fascinating literature on decoherence avoidance and suppression.

\acknowledgments
I am grateful to all the students in my 
Spring 2007 and 
Fall 2011 Quantum Error Correction 
courses, 
in particular Jose Raul Gonzalez Alonso, Chris Cantwell, Kung-Chuan Hsu, Siddharth Muthu Krishnan, Ching-Yi Lai, Hokiat Lim, 
Osonde Osoba, 
Kristen Pudenz, Greg Quiroz, 
Bilal Shaw, Mark Wilde, 
Sunil Yeshwanth,  and Yicong Zheng, for their meticulous notes, which formed the basis for this review. 
Thanks also to Gerardo Paz who helped with the section on CDD. I'm indebted to Steve Huntsman who found and helped correct many typos.
This work was
supported by the NSF Center for Quantum Information and
Computation for Chemistry, Grant No. CHE-1037992, NSF Grant No. CHE-924318, and by the ARO MURI Grant No. W911NF-11-1-0268. 

%\bibliographystyle{apsrev4-1}
%\bibliography{\string~/Dropbox/DW2-background/refs}

\begin{thebibliography}{39}%
\makeatletter
\providecommand \@ifxundefined [1]{%
 \@ifx{#1\undefined}
}%
\providecommand \@ifnum [1]{%
 \ifnum #1\expandafter \@firstoftwo
 \else \expandafter \@secondoftwo
 \fi
}%
\providecommand \@ifx [1]{%
 \ifx #1\expandafter \@firstoftwo
 \else \expandafter \@secondoftwo
 \fi
}%
\providecommand \natexlab [1]{#1}%
\providecommand \enquote  [1]{``#1''}%
\providecommand \bibnamefont  [1]{#1}%
\providecommand \bibfnamefont [1]{#1}%
\providecommand \citenamefont [1]{#1}%
\providecommand \href@noop [0]{\@secondoftwo}%
\providecommand \href [0]{\begingroup \@sanitize@url \@href}%
\providecommand \@href[1]{\@@startlink{#1}\@@href}%
\providecommand \@@href[1]{\endgroup#1\@@endlink}%
\providecommand \@sanitize@url [0]{\catcode `\\12\catcode `\$12\catcode
  `\&12\catcode `\#12\catcode `\^12\catcode `\_12\catcode `\%12\relax}%
\providecommand \@@startlink[1]{}%
\providecommand \@@endlink[0]{}%
\providecommand \url  [0]{\begingroup\@sanitize@url \@url }%
\providecommand \@url [1]{\endgroup\@href {#1}{\urlprefix }}%
\providecommand \urlprefix  [0]{URL }%
\providecommand \Eprint [0]{\href }%
\providecommand \doibase [0]{http://dx.doi.org/}%
\providecommand \selectlanguage [0]{\@gobble}%
\providecommand \bibinfo  [0]{\@secondoftwo}%
\providecommand \bibfield  [0]{\@secondoftwo}%
\providecommand \translation [1]{[#1]}%
\providecommand \BibitemOpen [0]{}%
\providecommand \bibitemStop [0]{}%
\providecommand \bibitemNoStop [0]{.\EOS\space}%
\providecommand \EOS [0]{\spacefactor3000\relax}%
\providecommand \BibitemShut  [1]{\csname bibitem#1\endcsname}%
\let\auto@bib@innerbib\@empty
%</preamble>
\bibitem [{\citenamefont {Nielsen}\ and\ \citenamefont
  {Chuang}(2000)}]{nielsen2000quantum}%
  \BibitemOpen
  \bibfield  {author} {\bibinfo {author} {\bibfnamefont {M.}~\bibnamefont
  {Nielsen}}\ and\ \bibinfo {author} {\bibfnamefont {I.}~\bibnamefont
  {Chuang}},\ }\href@noop {} {\emph {\bibinfo {title} {Quantum Computation and
  Quantum Information}}},\ Cambridge Series on Information and the Natural
  Sciences\ (\bibinfo  {publisher} {Cambridge University Press},\ \bibinfo
  {year} {2000})\BibitemShut {NoStop}%
\bibitem [{\citenamefont {Lidar}\ and\ \citenamefont
  {Whaley}(2003)}]{Lidar:2003fk}%
  \BibitemOpen
  \bibfield  {author} {\bibinfo {author} {\bibfnamefont {D.~A.}\ \bibnamefont
  {Lidar}}\ and\ \bibinfo {author} {\bibfnamefont {K.~B.}\ \bibnamefont
  {Whaley}},\ }\href {\doibase 10.1007/3-540-44874-8_5} {\emph {\bibinfo
  {title} {Irreversible Quantum Dynamics}}},\ edited by\ \bibinfo {editor}
  {\bibfnamefont {F.}~\bibnamefont {Benatti}}\ and\ \bibinfo {editor}
  {\bibfnamefont {R.}~\bibnamefont {Floreanini}},\ \bibinfo {series} {Lecture
  Notes in Physics}, Vol.\ \bibinfo {volume} {622}\ (\bibinfo  {publisher}
  {Springer Berlin / Heidelberg},\ \bibinfo {year} {2003})\ pp.\ \bibinfo
  {pages} {83--120}\BibitemShut {NoStop}%
\bibitem [{\citenamefont {Yang}\ \emph {et~al.}(2011)\citenamefont {Yang},
  \citenamefont {Wang},\ and\ \citenamefont {Liu}}]{Yang:2011:2}%
  \BibitemOpen
  \bibfield  {author} {\bibinfo {author} {\bibfnamefont {W.}~\bibnamefont
  {Yang}}, \bibinfo {author} {\bibfnamefont {Z.-Y.}\ \bibnamefont {Wang}}, \
  and\ \bibinfo {author} {\bibfnamefont {R.-B.}\ \bibnamefont {Liu}},\ }\href
  {\doibase 10.1007/s11467-010-0113-8} {\bibfield  {journal} {\bibinfo
  {journal} {Front. Phys.}\ }\textbf {\bibinfo {volume} {6}},\ \bibinfo {pages}
  {2} (\bibinfo {year} {2011})}\BibitemShut {NoStop}%
\bibitem [{\citenamefont {Lidar}\ and\ \citenamefont
  {Brun}(2013)}]{Lidar-Brun:book}%
  \BibitemOpen
  \bibinfo {editor} {\bibfnamefont {D.}~\bibnamefont {Lidar}}\ and\ \bibinfo
  {editor} {\bibfnamefont {T.}~\bibnamefont {Brun}},\ eds.,\ \href
  {http://www.cambridge.org/9780521897877} {\emph {\bibinfo {title} {Quantum
  Error Correction}}}\ (\bibinfo  {publisher} {Cambridge University Press},\
  \bibinfo {address} {{Cambride, UK}},\ \bibinfo {year} {2013})\BibitemShut
  {NoStop}%
\bibitem [{\citenamefont {Hammermesh}(1989)}]{Hammermesh:grouptheorybook}%
  \BibitemOpen
  \bibfield  {author} {\bibinfo {author} {\bibfnamefont {M.}~\bibnamefont
  {Hammermesh}},\ }\href@noop {} {\emph {\bibinfo {title} {Group Theory and its
  Application to Physical Problems}}}\ (\bibinfo  {publisher} {Dover
  Publications},\ \bibinfo {year} {1989})\BibitemShut {NoStop}%
\bibitem [{\citenamefont {Loss}\ and\ \citenamefont
  {DiVincenzo}(1998)}]{Loss:1998zr}%
  \BibitemOpen
  \bibfield  {author} {\bibinfo {author} {\bibfnamefont {D.}~\bibnamefont
  {Loss}}\ and\ \bibinfo {author} {\bibfnamefont {D.~P.}\ \bibnamefont
  {DiVincenzo}},\ }\href {http://link.aps.org/doi/10.1103/PhysRevA.57.120}
  {\bibfield  {journal} {\bibinfo  {journal} {Physical Review A}\ }\textbf
  {\bibinfo {volume} {57}},\ \bibinfo {pages} {120} (\bibinfo {year}
  {1998})}\BibitemShut {NoStop}%
\bibitem [{\citenamefont {Mizel}\ and\ \citenamefont
  {Lidar}(2004)}]{Mizel:2004ly}%
  \BibitemOpen
  \bibfield  {author} {\bibinfo {author} {\bibfnamefont {A.}~\bibnamefont
  {Mizel}}\ and\ \bibinfo {author} {\bibfnamefont {D.~A.}\ \bibnamefont
  {Lidar}},\ }\href {http://link.aps.org/doi/10.1103/PhysRevB.70.115310}
  {\bibfield  {journal} {\bibinfo  {journal} {Physical Review B}\ }\textbf
  {\bibinfo {volume} {70}},\ \bibinfo {pages} {115310} (\bibinfo {year}
  {2004})}\BibitemShut {NoStop}%
\bibitem [{\citenamefont {Bacon}\ \emph {et~al.}(2000)\citenamefont {Bacon},
  \citenamefont {Kempe}, \citenamefont {Lidar},\ and\ \citenamefont
  {Whaley}}]{Bacon:2000qf}%
  \BibitemOpen
  \bibfield  {author} {\bibinfo {author} {\bibfnamefont {D.}~\bibnamefont
  {Bacon}}, \bibinfo {author} {\bibfnamefont {J.}~\bibnamefont {Kempe}},
  \bibinfo {author} {\bibfnamefont {D.~A.}\ \bibnamefont {Lidar}}, \ and\
  \bibinfo {author} {\bibfnamefont {K.~B.}\ \bibnamefont {Whaley}},\ }\href
  {http://link.aps.org/doi/10.1103/PhysRevLett.85.1758} {\bibfield  {journal}
  {\bibinfo  {journal} {{Phys.~Rev.~Lett.}}\ }\textbf {\bibinfo {volume}
  {85}},\ \bibinfo {pages} {1758} (\bibinfo {year} {2000})}\BibitemShut
  {NoStop}%
\bibitem [{\citenamefont {Woodworth}\ \emph {et~al.}(2006)\citenamefont
  {Woodworth}, \citenamefont {Mizel},\ and\ \citenamefont
  {Lidar}}]{Woodworth:2006bh}%
  \BibitemOpen
  \bibfield  {author} {\bibinfo {author} {\bibfnamefont {R.}~\bibnamefont
  {Woodworth}}, \bibinfo {author} {\bibfnamefont {A.}~\bibnamefont {Mizel}}, \
  and\ \bibinfo {author} {\bibfnamefont {D.~A.}\ \bibnamefont {Lidar}},\ }\href
  {http://stacks.iop.org/0953-8984/18/i=21/a=S02} {\bibfield  {journal}
  {\bibinfo  {journal} {Journal of Physics: Condensed Matter}\ }\textbf
  {\bibinfo {volume} {18}},\ \bibinfo {pages} {S721} (\bibinfo {year}
  {2006})}\BibitemShut {NoStop}%
\bibitem [{\citenamefont {Knill}\ \emph {et~al.}(2000)\citenamefont {Knill},
  \citenamefont {Laflamme},\ and\ \citenamefont {Viola}}]{Knill:2000dq}%
  \BibitemOpen
  \bibfield  {author} {\bibinfo {author} {\bibfnamefont {E.}~\bibnamefont
  {Knill}}, \bibinfo {author} {\bibfnamefont {R.}~\bibnamefont {Laflamme}}, \
  and\ \bibinfo {author} {\bibfnamefont {L.}~\bibnamefont {Viola}},\ }\href
  {http://link.aps.org/doi/10.1103/PhysRevLett.84.2525} {\bibfield  {journal}
  {\bibinfo  {journal} {{Phys.~Rev.~Lett.}}\ }\textbf {\bibinfo {volume}
  {84}},\ \bibinfo {pages} {2525} (\bibinfo {year} {2000})}\BibitemShut
  {NoStop}%
\bibitem [{\citenamefont {Kempe}\ \emph {et~al.}(2001)\citenamefont {Kempe},
  \citenamefont {Bacon}, \citenamefont {Lidar},\ and\ \citenamefont
  {Whaley}}]{Kempe:2001uq}%
  \BibitemOpen
  \bibfield  {author} {\bibinfo {author} {\bibfnamefont {J.}~\bibnamefont
  {Kempe}}, \bibinfo {author} {\bibfnamefont {D.}~\bibnamefont {Bacon}},
  \bibinfo {author} {\bibfnamefont {D.~A.}\ \bibnamefont {Lidar}}, \ and\
  \bibinfo {author} {\bibfnamefont {K.~B.}\ \bibnamefont {Whaley}},\ }\href
  {http://link.aps.org/doi/10.1103/PhysRevA.63.042307} {\bibfield  {journal}
  {\bibinfo  {journal} {Physical Review A}\ }\textbf {\bibinfo {volume} {63}},\
  \bibinfo {pages} {042307} (\bibinfo {year} {2001})}\BibitemShut {NoStop}%
\bibitem [{\citenamefont {DiVincenzo}\ \emph {et~al.}(2000)\citenamefont
  {DiVincenzo}, \citenamefont {Bacon}, \citenamefont {Kempe}, \citenamefont
  {Burkard},\ and\ \citenamefont {Whaley}}]{DiVincenzo:2000kx}%
  \BibitemOpen
  \bibfield  {author} {\bibinfo {author} {\bibfnamefont {D.~P.}\ \bibnamefont
  {DiVincenzo}}, \bibinfo {author} {\bibfnamefont {D.}~\bibnamefont {Bacon}},
  \bibinfo {author} {\bibfnamefont {J.}~\bibnamefont {Kempe}}, \bibinfo
  {author} {\bibfnamefont {G.}~\bibnamefont {Burkard}}, \ and\ \bibinfo
  {author} {\bibfnamefont {K.~B.}\ \bibnamefont {Whaley}},\ }\href {\doibase
  10.1038/35042541} {\bibfield  {journal} {\bibinfo  {journal} {Nature}\
  }\textbf {\bibinfo {volume} {408}},\ \bibinfo {pages} {339} (\bibinfo {year}
  {2000})}\BibitemShut {NoStop}%
\bibitem [{\citenamefont {{R. Bhatia}}(1997)}]{Bhatia:book}%
  \BibitemOpen
  \bibfield  {author} {\bibinfo {author} {\bibnamefont {{R. Bhatia}}},\
  }\href@noop {} {\emph {\bibinfo {title} {{Matrix Analysis}}}},\ \bibinfo
  {series} {{Graduate Texts in Mathematics}}\ No.\ \bibinfo {number} {169}\
  (\bibinfo  {publisher} {Springer-Verlag},\ \bibinfo {address} {{New York}},\
  \bibinfo {year} {1997})\BibitemShut {NoStop}%
\bibitem [{\citenamefont {Ng}\ \emph {et~al.}(2011)\citenamefont {Ng},
  \citenamefont {Lidar},\ and\ \citenamefont {Preskill}}]{NLP:09}%
  \BibitemOpen
  \bibfield  {author} {\bibinfo {author} {\bibfnamefont {H.-K.}\ \bibnamefont
  {Ng}}, \bibinfo {author} {\bibfnamefont {D.~A.}\ \bibnamefont {Lidar}}, \
  and\ \bibinfo {author} {\bibfnamefont {J.~P.}\ \bibnamefont {Preskill}},\
  }\href@noop {} {\bibfield  {journal} {\bibinfo  {journal} {Phys. Rev. A}\ }
  (\bibinfo {year} {2011})}\BibitemShut {NoStop}%
\bibitem [{\citenamefont {Zanardi}(1999{\natexlab{a}})}]{Zanardi:1999fk}%
  \BibitemOpen
  \bibfield  {author} {\bibinfo {author} {\bibfnamefont {P.}~\bibnamefont
  {Zanardi}},\ }\href {\doibase
  http://dx.doi.org/10.1016/S0375-9601(99)00365-5} {\bibfield  {journal}
  {\bibinfo  {journal} {Physics Letters A}\ }\textbf {\bibinfo {volume}
  {258}},\ \bibinfo {pages} {77} (\bibinfo {year}
  {1999}{\natexlab{a}})}\BibitemShut {NoStop}%
\bibitem [{\citenamefont {Byrd}\ and\ \citenamefont
  {Lidar}(2002)}]{Byrd:2002:047901}%
  \BibitemOpen
  \bibfield  {author} {\bibinfo {author} {\bibfnamefont {M.~S.}\ \bibnamefont
  {Byrd}}\ and\ \bibinfo {author} {\bibfnamefont {D.~A.}\ \bibnamefont
  {Lidar}},\ }\href {\doibase 10.1103/PhysRevLett.89.047901} {\bibfield
  {journal} {\bibinfo  {journal} {Phys. Rev. Lett.}\ }\textbf {\bibinfo
  {volume} {89}},\ \bibinfo {pages} {047901} (\bibinfo {year}
  {2002})}\BibitemShut {NoStop}%
\bibitem [{\citenamefont {Lidar}\ and\ \citenamefont
  {Wu}(2001)}]{Lidar:2001vn}%
  \BibitemOpen
  \bibfield  {author} {\bibinfo {author} {\bibfnamefont {D.~A.}\ \bibnamefont
  {Lidar}}\ and\ \bibinfo {author} {\bibfnamefont {L.~A.}\ \bibnamefont {Wu}},\
  }\href {http://link.aps.org/doi/10.1103/PhysRevLett.88.017905} {\bibfield
  {journal} {\bibinfo  {journal} {{Phys.~Rev.~Lett.}}\ }\textbf {\bibinfo
  {volume} {88}},\ \bibinfo {pages} {017905} (\bibinfo {year}
  {2001})}\BibitemShut {NoStop}%
\bibitem [{\citenamefont {Wu}\ \emph {et~al.}(2002)\citenamefont {Wu},
  \citenamefont {Byrd},\ and\ \citenamefont {Lidar}}]{Wu:2002ys}%
  \BibitemOpen
  \bibfield  {author} {\bibinfo {author} {\bibfnamefont {L.~A.}\ \bibnamefont
  {Wu}}, \bibinfo {author} {\bibfnamefont {M.~S.}\ \bibnamefont {Byrd}}, \ and\
  \bibinfo {author} {\bibfnamefont {D.~A.}\ \bibnamefont {Lidar}},\ }\href
  {http://link.aps.org/doi/10.1103/PhysRevLett.89.127901} {\bibfield  {journal}
  {\bibinfo  {journal} {{Phys.~Rev.~Lett.}}\ }\textbf {\bibinfo {volume}
  {89}},\ \bibinfo {pages} {127901} (\bibinfo {year} {2002})}\BibitemShut
  {NoStop}%
\bibitem [{\citenamefont {Khodjasteh}\ and\ \citenamefont
  {Lidar}(2007)}]{Khodjasteh:2007zr}%
  \BibitemOpen
  \bibfield  {author} {\bibinfo {author} {\bibfnamefont {K.}~\bibnamefont
  {Khodjasteh}}\ and\ \bibinfo {author} {\bibfnamefont {D.~A.}\ \bibnamefont
  {Lidar}},\ }\href {http://link.aps.org/doi/10.1103/PhysRevA.75.062310}
  {\bibfield  {journal} {\bibinfo  {journal} {Physical Review A}\ }\textbf
  {\bibinfo {volume} {75}},\ \bibinfo {pages} {062310} (\bibinfo {year}
  {2007})}\BibitemShut {NoStop}%
\bibitem [{\citenamefont {{\'A}lvarez}\ \emph {et~al.}(2010)\citenamefont
  {{\'A}lvarez}, \citenamefont {Ajoy}, \citenamefont {Peng},\ and\
  \citenamefont {Suter}}]{Alvarez:2010ve}%
  \BibitemOpen
  \bibfield  {author} {\bibinfo {author} {\bibfnamefont {G.~A.}\ \bibnamefont
  {{\'A}lvarez}}, \bibinfo {author} {\bibfnamefont {A.}~\bibnamefont {Ajoy}},
  \bibinfo {author} {\bibfnamefont {X.}~\bibnamefont {Peng}}, \ and\ \bibinfo
  {author} {\bibfnamefont {D.}~\bibnamefont {Suter}},\ }\href
  {http://link.aps.org/doi/10.1103/PhysRevA.82.042306} {\bibfield  {journal}
  {\bibinfo  {journal} {Physical Review A}\ }\textbf {\bibinfo {volume} {82}},\
  \bibinfo {pages} {042306} (\bibinfo {year} {2010})}\BibitemShut {NoStop}%
\bibitem [{\citenamefont {Uhrig}(2007{\natexlab{a}})}]{Uhrig:2007qf}%
  \BibitemOpen
  \bibfield  {author} {\bibinfo {author} {\bibfnamefont {G.~S.}\ \bibnamefont
  {Uhrig}},\ }\href {http://link.aps.org/doi/10.1103/PhysRevLett.98.100504}
  {\bibfield  {journal} {\bibinfo  {journal} {Physical Review Letters}\
  }\textbf {\bibinfo {volume} {98}},\ \bibinfo {pages} {100504} (\bibinfo
  {year} {2007}{\natexlab{a}})}\BibitemShut {NoStop}%
\bibitem [{\citenamefont {West}\ \emph
  {et~al.}(2010{\natexlab{a}})\citenamefont {West}, \citenamefont {Fong},\ and\
  \citenamefont {Lidar}}]{West:2010dq}%
  \BibitemOpen
  \bibfield  {author} {\bibinfo {author} {\bibfnamefont {J.~R.}\ \bibnamefont
  {West}}, \bibinfo {author} {\bibfnamefont {B.~H.}\ \bibnamefont {Fong}}, \
  and\ \bibinfo {author} {\bibfnamefont {D.~A.}\ \bibnamefont {Lidar}},\ }\href
  {http://link.aps.org/doi/10.1103/PhysRevLett.104.130501} {\bibfield
  {journal} {\bibinfo  {journal} {Physical Review Letters}\ }\textbf {\bibinfo
  {volume} {104}},\ \bibinfo {pages} {130501} (\bibinfo {year}
  {2010}{\natexlab{a}})}\BibitemShut {NoStop}%
\bibitem [{\citenamefont {Wang}\ and\ \citenamefont
  {Liu}(2011{\natexlab{a}})}]{Wang:10}%
  \BibitemOpen
  \bibfield  {author} {\bibinfo {author} {\bibfnamefont {Z.-Y.}\ \bibnamefont
  {Wang}}\ and\ \bibinfo {author} {\bibfnamefont {R.-B.}\ \bibnamefont {Liu}},\
  }\href {\doibase 10.1103/PhysRevA.83.022306} {\bibfield  {journal} {\bibinfo
  {journal} {Phys. Rev. A}\ }\textbf {\bibinfo {volume} {83}},\ \bibinfo
  {pages} {022306} (\bibinfo {year} {2011}{\natexlab{a}})}\BibitemShut
  {NoStop}%
\bibitem [{\citenamefont {Zanardi}(1999{\natexlab{b}})}]{Zanardi:1999:77}%
  \BibitemOpen
  \bibfield  {author} {\bibinfo {author} {\bibfnamefont {P.}~\bibnamefont
  {Zanardi}},\ }\href {\doibase
  http://dx.doi.org/10.1016/S0375-9601(99)00365-5} {\bibfield  {journal}
  {\bibinfo  {journal} {Physics Letters A}\ }\textbf {\bibinfo {volume}
  {258}},\ \bibinfo {pages} {77} (\bibinfo {year}
  {1999}{\natexlab{b}})}\BibitemShut {NoStop}%
\bibitem [{\citenamefont {Viola}\ \emph {et~al.}(1999)\citenamefont {Viola},
  \citenamefont {Lloyd},\ and\ \citenamefont {Knill}}]{Viola:1999:4888}%
  \BibitemOpen
  \bibfield  {author} {\bibinfo {author} {\bibfnamefont {L.}~\bibnamefont
  {Viola}}, \bibinfo {author} {\bibfnamefont {S.}~\bibnamefont {Lloyd}}, \ and\
  \bibinfo {author} {\bibfnamefont {E.}~\bibnamefont {Knill}},\ }\href
  {http://link.aps.org/doi/10.1103/PhysRevLett.83.4888} {\bibfield  {journal}
  {\bibinfo  {journal} {Physical Review Letters}\ }\textbf {\bibinfo {volume}
  {83}},\ \bibinfo {pages} {4888} (\bibinfo {year} {1999})}\BibitemShut
  {NoStop}%
\bibitem [{\citenamefont {Viola}\ \emph {et~al.}(2000)\citenamefont {Viola},
  \citenamefont {Knill},\ and\ \citenamefont {Lloyd}}]{Viola:2000:3520}%
  \BibitemOpen
  \bibfield  {author} {\bibinfo {author} {\bibfnamefont {L.}~\bibnamefont
  {Viola}}, \bibinfo {author} {\bibfnamefont {E.}~\bibnamefont {Knill}}, \ and\
  \bibinfo {author} {\bibfnamefont {S.}~\bibnamefont {Lloyd}},\ }\href
  {http://link.aps.org/doi/10.1103/PhysRevLett.85.3520} {\bibfield  {journal}
  {\bibinfo  {journal} {Physical Review Letters}\ }\textbf {\bibinfo {volume}
  {85}},\ \bibinfo {pages} {3520} (\bibinfo {year} {2000})}\BibitemShut
  {NoStop}%
\bibitem [{\citenamefont {Viola}\ and\ \citenamefont
  {Knill}(2005)}]{Viola:2005:060502}%
  \BibitemOpen
  \bibfield  {author} {\bibinfo {author} {\bibfnamefont {L.}~\bibnamefont
  {Viola}}\ and\ \bibinfo {author} {\bibfnamefont {E.}~\bibnamefont {Knill}},\
  }\href {http://link.aps.org/doi/10.1103/PhysRevLett.94.060502} {\bibfield
  {journal} {\bibinfo  {journal} {Physical Review Letters}\ }\textbf {\bibinfo
  {volume} {94}},\ \bibinfo {pages} {060502} (\bibinfo {year}
  {2005})}\BibitemShut {NoStop}%
\bibitem [{\citenamefont {Wu}\ and\ \citenamefont {Lidar}(2002)}]{Wu:2002zr}%
  \BibitemOpen
  \bibfield  {author} {\bibinfo {author} {\bibfnamefont {L.~A.}\ \bibnamefont
  {Wu}}\ and\ \bibinfo {author} {\bibfnamefont {D.~A.}\ \bibnamefont {Lidar}},\
  }\href {http://link.aps.org/doi/10.1103/PhysRevLett.88.207902} {\bibfield
  {journal} {\bibinfo  {journal} {Physical Review Letters}\ }\textbf {\bibinfo
  {volume} {88}},\ \bibinfo {pages} {207902} (\bibinfo {year}
  {2002})}\BibitemShut {NoStop}%
\bibitem [{\citenamefont {West}\ \emph
  {et~al.}(2010{\natexlab{b}})\citenamefont {West}, \citenamefont {Lidar},
  \citenamefont {Fong},\ and\ \citenamefont {Gyure}}]{West:10}%
  \BibitemOpen
  \bibfield  {author} {\bibinfo {author} {\bibfnamefont {J.~R.}\ \bibnamefont
  {West}}, \bibinfo {author} {\bibfnamefont {D.~A.}\ \bibnamefont {Lidar}},
  \bibinfo {author} {\bibfnamefont {B.~H.}\ \bibnamefont {Fong}}, \ and\
  \bibinfo {author} {\bibfnamefont {M.~F.}\ \bibnamefont {Gyure}},\ }\href
  {\doibase 10.1103/PhysRevLett.105.230503} {\bibfield  {journal} {\bibinfo
  {journal} {Phys. Rev. Lett.}\ }\textbf {\bibinfo {volume} {105}},\ \bibinfo
  {pages} {230503} (\bibinfo {year} {2010}{\natexlab{b}})}\BibitemShut
  {NoStop}%
\bibitem [{\citenamefont {Lidar}(2008)}]{PhysRevLett.100.160506}%
  \BibitemOpen
  \bibfield  {author} {\bibinfo {author} {\bibfnamefont {D.~A.}\ \bibnamefont
  {Lidar}},\ }\href {http://link.aps.org/doi/10.1103/PhysRevLett.100.160506}
  {\bibfield  {journal} {\bibinfo  {journal} {{Phys.~Rev.~Lett.}}\ }\textbf
  {\bibinfo {volume} {100}},\ \bibinfo {pages} {160506} (\bibinfo {year}
  {2008})}\BibitemShut {NoStop}%
\bibitem [{\citenamefont {Uhrig}(2007{\natexlab{b}})}]{Uhrig:2007:100504}%
  \BibitemOpen
  \bibfield  {author} {\bibinfo {author} {\bibfnamefont {G.~S.}\ \bibnamefont
  {Uhrig}},\ }\href {\doibase 10.1103/PhysRevLett.98.100504} {\bibfield
  {journal} {\bibinfo  {journal} {Phys. Rev. Lett.}\ }\textbf {\bibinfo
  {volume} {98}},\ \bibinfo {pages} {100504} (\bibinfo {year}
  {2007}{\natexlab{b}})}\BibitemShut {NoStop}%
\bibitem [{\citenamefont {West}\ \emph
  {et~al.}(2010{\natexlab{c}})\citenamefont {West}, \citenamefont {Fong},\ and\
  \citenamefont {Lidar}}]{West:2010:130501}%
  \BibitemOpen
  \bibfield  {author} {\bibinfo {author} {\bibfnamefont {J.~R.}\ \bibnamefont
  {West}}, \bibinfo {author} {\bibfnamefont {B.~H.}\ \bibnamefont {Fong}}, \
  and\ \bibinfo {author} {\bibfnamefont {D.~A.}\ \bibnamefont {Lidar}},\ }\href
  {\doibase 10.1103/PhysRevLett.104.130501} {\bibfield  {journal} {\bibinfo
  {journal} {Phys. Rev. Lett.}\ }\textbf {\bibinfo {volume} {104}},\ \bibinfo
  {pages} {130501} (\bibinfo {year} {2010}{\natexlab{c}})}\BibitemShut
  {NoStop}%
\bibitem [{\citenamefont {Wang}\ and\ \citenamefont
  {Liu}(2011{\natexlab{b}})}]{Wang:2011:022306}%
  \BibitemOpen
  \bibfield  {author} {\bibinfo {author} {\bibfnamefont {Z.-Y.}\ \bibnamefont
  {Wang}}\ and\ \bibinfo {author} {\bibfnamefont {R.-B.}\ \bibnamefont {Liu}},\
  }\href {\doibase 10.1103/PhysRevA.83.022306} {\bibfield  {journal} {\bibinfo
  {journal} {Phys. Rev. A}\ }\textbf {\bibinfo {volume} {83}},\ \bibinfo
  {pages} {022306} (\bibinfo {year} {2011}{\natexlab{b}})}\BibitemShut
  {NoStop}%
\bibitem [{\citenamefont {Xia}\ \emph {et~al.}(2011)\citenamefont {Xia},
  \citenamefont {Uhrig},\ and\ \citenamefont {Lidar}}]{Xia:2011uq}%
  \BibitemOpen
  \bibfield  {author} {\bibinfo {author} {\bibfnamefont {Y.}~\bibnamefont
  {Xia}}, \bibinfo {author} {\bibfnamefont {G.~S.}\ \bibnamefont {Uhrig}}, \
  and\ \bibinfo {author} {\bibfnamefont {D.~A.}\ \bibnamefont {Lidar}},\ }\href
  {http://link.aps.org/doi/10.1103/PhysRevA.84.062332} {\bibfield  {journal}
  {\bibinfo  {journal} {Phys. Rev. A}\ }\textbf {\bibinfo {volume} {84}},\
  \bibinfo {pages} {062332} (\bibinfo {year} {2011})}\BibitemShut {NoStop}%
\bibitem [{\citenamefont {Uys}\ \emph {et~al.}(2009)\citenamefont {Uys},
  \citenamefont {Biercuk},\ and\ \citenamefont {Bollinger}}]{Uys:2009:040501}%
  \BibitemOpen
  \bibfield  {author} {\bibinfo {author} {\bibfnamefont {H.}~\bibnamefont
  {Uys}}, \bibinfo {author} {\bibfnamefont {M.~J.}\ \bibnamefont {Biercuk}}, \
  and\ \bibinfo {author} {\bibfnamefont {J.~J.}\ \bibnamefont {Bollinger}},\
  }\href {\doibase 10.1103/PhysRevLett.103.040501} {\bibfield  {journal}
  {\bibinfo  {journal} {Phys. Rev. Lett.}\ }\textbf {\bibinfo {volume} {103}},\
  \bibinfo {pages} {040501} (\bibinfo {year} {2009})}\BibitemShut {NoStop}%
\bibitem [{\citenamefont {Gordon}\ \emph {et~al.}(2008)\citenamefont {Gordon},
  \citenamefont {Kurizki},\ and\ \citenamefont {Lidar}}]{Gordon:2008:010403}%
  \BibitemOpen
  \bibfield  {author} {\bibinfo {author} {\bibfnamefont {G.}~\bibnamefont
  {Gordon}}, \bibinfo {author} {\bibfnamefont {G.}~\bibnamefont {Kurizki}}, \
  and\ \bibinfo {author} {\bibfnamefont {D.~A.}\ \bibnamefont {Lidar}},\ }\href
  {\doibase 10.1103/PhysRevLett.101.010403} {\bibfield  {journal} {\bibinfo
  {journal} {Phys. Rev. Lett.}\ }\textbf {\bibinfo {volume} {101}},\ \bibinfo
  {pages} {010403} (\bibinfo {year} {2008})}\BibitemShut {NoStop}%
\bibitem [{\citenamefont {Clausen}\ \emph {et~al.}(2010)\citenamefont
  {Clausen}, \citenamefont {Bensky},\ and\ \citenamefont
  {Kurizki}}]{Clausen:2010:040401}%
  \BibitemOpen
  \bibfield  {author} {\bibinfo {author} {\bibfnamefont {J.}~\bibnamefont
  {Clausen}}, \bibinfo {author} {\bibfnamefont {G.}~\bibnamefont {Bensky}}, \
  and\ \bibinfo {author} {\bibfnamefont {G.}~\bibnamefont {Kurizki}},\ }\href
  {\doibase 10.1103/PhysRevLett.104.040401} {\bibfield  {journal} {\bibinfo
  {journal} {Phys. Rev. Lett.}\ }\textbf {\bibinfo {volume} {104}},\ \bibinfo
  {pages} {040401} (\bibinfo {year} {2010})}\BibitemShut {NoStop}%
\bibitem [{\citenamefont {Santos}\ and\ \citenamefont
  {Viola}(2006)}]{Santos:2006:150501}%
  \BibitemOpen
  \bibfield  {author} {\bibinfo {author} {\bibfnamefont {L.~F.}\ \bibnamefont
  {Santos}}\ and\ \bibinfo {author} {\bibfnamefont {L.}~\bibnamefont {Viola}},\
  }\href {\doibase 10.1103/PhysRevLett.97.150501} {\bibfield  {journal}
  {\bibinfo  {journal} {Phys. Rev. Lett.}\ }\textbf {\bibinfo {volume} {97}},\
  \bibinfo {pages} {150501} (\bibinfo {year} {2006})}\BibitemShut {NoStop}%
\bibitem [{\citenamefont {Kribs}\ \emph {et~al.}(2005)\citenamefont {Kribs},
  \citenamefont {Laflamme},\ and\ \citenamefont {Poulin}}]{Kribs:2005:180501}%
  \BibitemOpen
  \bibfield  {author} {\bibinfo {author} {\bibfnamefont {D.}~\bibnamefont
  {Kribs}}, \bibinfo {author} {\bibfnamefont {R.}~\bibnamefont {Laflamme}}, \
  and\ \bibinfo {author} {\bibfnamefont {D.}~\bibnamefont {Poulin}},\ }\href
  {\doibase 10.1103/PhysRevLett.94.180501} {\bibfield  {journal} {\bibinfo
  {journal} {Phys. Rev. Lett.}\ }\textbf {\bibinfo {volume} {94}},\ \bibinfo
  {pages} {180501} (\bibinfo {year} {2005})}\BibitemShut {NoStop}%
\end{thebibliography}


\end{document}